\documentclass[12pt,draftcls,onecolumn]{IEEEtran}

\usepackage{cite}
\usepackage{graphicx}
\usepackage[cmex10]{amsmath} 
\usepackage{amssymb,amsthm}
\usepackage{flushend}
\usepackage{bbm}
\usepackage{enumitem}

\newtheorem{Theorem}{Theorem}
\newtheorem{lemma}{Lemma}
\newtheorem{Remark}{Remark}
\newtheorem{Corollary}{Corollary}

\newtheorem{definition}{Definition}

\makeatletter%
\if@twocolumn%
\newcommand{\Figwidth}{\columnwidth}%
\newcommand{\hspaceonetwocol}[2]{\hspace{#2}}
\newcommand{\includeonetwocol}[2]{#2}
\def\twocolbreak{\nonumber\\ &}%
\def\twocolnewline{\nonumber\\}%
\def\twocolAlignMarker{&}%
\else
\newcommand{\Figwidth}{4.5in}%
\newcommand{\hspaceonetwocol}[2]{\hspace{#1}}
\newcommand{\includeonetwocol}[2]{#1}
\def\twocolbreak{}%
\def\twocolnewline{}%
\def\twocolAlignMarker{}%
\fi%
\makeatother%

\def\msubK{{\ell_{m}}}
\def\hfeat{h_{\ell_{m}}^{m}}
\def\hnod{h_{i,m}}
\def\Prob{\mathbb{P}}
\def\ChiS{{\Lambda}}
\def\TreeSideInfo{{\tilde{\boldsymbol \tau}}}
\def\TreeLabels{{\boldsymbol \tau}}
\def\RootEstimator{{\hat{\tau}_0}}
\def\RootLabel{{\tau_0}}
\def\CroppedTree{{T^t}}
\def\LLRCroppedTree{{\Gamma_0^t}}
\def\TreeIter{{\hat{t}}}
\def\TreeDepth{{t}}

\allowdisplaybreaks

\begin{document}
\title{Recovering a Single Community with Side Information}

\author{Hussein Saad, {\em Student Member, IEEE}, and Aria Nosratinia, {\em Fellow, IEEE}
\thanks{The authors are with the Department of Electrical Engineering, University of Texas at
    Dallas, Richardson, TX 75083-0688 USA, E-mail:
    hussein.saad@utdallas.edu; aria@utdallas.edu.}
\thanks{This work was presented in part in ISIT 2018, Vail, CO.~\cite{ISIT-2018-1,ISIT-2018-2}.}
}
\maketitle


\begin{abstract}

We study the effect of the quality and quantity of side information on the recovery of a hidden community of size $K=o(n)$ in a graph of size $n$.  Side information for each node in the graph is modeled by a random vector with the following features: either the dimension of the vector is allowed to vary with $n$, while log-likelihood ratio (LLR) of each component with respect to the node label is fixed, or the LLR is allowed to vary and the vector dimension is fixed. These two models represent the variation in quality and quantity of side information. Under maximum likelihood detection, we calculate tight necessary and sufficient conditions for exact recovery of the labels. We demonstrate how side information needs to evolve with $n$ in terms of either its quantity, or quality, to improve the exact recovery threshold. A similar set of results are obtained for weak recovery. Under belief propagation, tight necessary and sufficient conditions for weak recovery are calculated when the LLRs are constant, and sufficient conditions when the LLRs vary with $n$. Moreover, we design and analyze a local voting procedure using side information that can achieve exact recovery when applied after belief propagation. The results for belief propagation are validated via simulations on finite synthetic data-sets, showing that the asymptotic results of this paper can also shed light on the performance at finite $n$.
\end{abstract}

\IEEEpeerreviewmaketitle


\section{Introduction}
\label{into.}

Detecting communities (or clusters) in graphs is a fundamental problem that has been studied in various fields, statistics~\cite{SBM_first,Min_max_SBM,stat1,stat2,stat3}, computer science~\cite{SCT,CS1,CS2,CS3,CS4} and theoretical statistical physics~\cite{Phys,Phys2}. It has many applications: finding like-minded people in social networks~\cite{social}, improving recommendation systems~\cite{recommendation}, detecting protein complexes~\cite{protein}. In this paper, we consider the problem of finding a single sub-graph (community) hidden in a large graph, where the community size is much smaller than the graph size. Applications of finding a hidden community include fraud activity detection~\cite{Hidd_1, Hidd_2} and correlation mining~\cite{pmlr-v31-firouzi13a}. 

Several models have been studied for random graphs that exhibit a community structure~\cite{survey1}. A widely used model in the context of community detection is the stochastic block model (SBM)~\cite{exact_general_sbm}. In this paper, the stochastic block model for one community is considered~\cite{montnari,infor_limits,recovering_one, Kadavankandy:SingleCommunity}. The stochastic block model for one community consists of a graph of size $n$ with a community of size $K$, where any two nodes are connected with probability $p$ if they are both within the community, and with probability $q$ otherwise.

The problem of finding a hidden community upon observing \textit{only} the graph has been studied in~\cite{montnari,infor_limits,recovering_one}. The information limits\footnote{The extremal phase transition threshold is also known as {\em information theoretic limit}~\cite{exact_general_sbm} or {\em information limit}~\cite{infor_limits}. We use the latter term throughout this paper.} of weak recovery and exact recovery  have been studied in~\cite{infor_limits}. Weak recovery is achieved when the expected number of misclassified nodes is $o(K)$, and exact recovery when all labels are recovered with probability approaching one. The limits of belief propagation for weak recovery have been characterized~\cite{recovering_one,montnari} in terms of a signal-to-noise ratio parameter $\lambda = \frac{K^{2}(p-q)^{2}}{(n-k)q}$. The utility of a voting procedure after belief propagation to achieve exact recovery was pointed out in~\cite{recovering_one}.

Graphical models are popular because they represent many large data sets and give insight on the performance of inference algorithms, but also in many inference problems they do not capture all data that is both relevant and available. In many practical applications, non-graphical relevant information is available that can aid the inference. For example, social networks such as Facebook and Twitter have access to other information other than the graph edges such as date of birth, nationality, school. A citation network has the authors’ names, keywords, and therefore may provide significant additional information beyond the co-authoring relationships. 
This paper characterizes the utility of side information in single-community detection, in particular exploring {\em when and by how much can side information improve the information limit, as well as the phase transition of belief propagation, in single-community detection.}

We model a varying quantity and quality of side information by associating with each node a vector (i.e., non-graphical) observation whose dimension represents the quantity of side information and whose (element-wise) log-likelihood ratios (LLRs) with respect to node labels represents the quality of side information.  The contributions of this paper can be summarized as follows:

\begin{itemize}

\item The information limits in the presence of side information are characterized. When the dimension of side information for each node varies but its LLR is fixed across $n$, tight necessary and sufficient conditions are calculated for both weak and exact recovery. Also, it is shown that under the same sufficient conditions, weak recovery is achievable even when the size of the community is random and unknown. We also find conditions on the graph and side information where achievability of weak recovery implies achievability of exact recovery. Subject to some mild conditions on the exponential moments of LLR, the results apply to both discrete as well as continuous-valued side information.
  
When the side information for each node has fixed dimension but varying LLR, we find tight necessary and sufficient conditions for exact recovery, and necessary conditions for weak recovery. Under varying LLR, our results apply to side information with finite alphabet.

\item The phase transition of belief propagation in the presence of side information is characterized, where we assume the side information per node has a fixed dimension.  When the LLRs are fixed across $n$, tight necessary and sufficient conditions are calculated for weak recovery. Furthermore, it is shown that when belief propagation fails, no local algorithm can achieve weak recovery. It is also shown than belief propagation is strictly inferior to the maximum likelihood detector. Numerical results on finite synthetic data-sets validate our asymptotic analysis and show the relevance of our asymptotic results to even graphs of moderate size. We also calculate conditions under which  belief propagation followed by a local voting procedure achieves exact recovery.

When the side information has variable LLR across $n$, the belief propagation misclassification rate was calculated using density evolution. Our results generalize~\cite{Kadavankandy:SingleCommunity}, where it was shown that belief propagation achieves weak recovery for $\lambda >0$ only for binary side information consisting of noisy labels with vanishing noise.

\end{itemize}

We now present a brief review of the literature in the area of side information for community detection and highlight the distinctions of the present work.  In the context of detecting two or more  communities: Mossel and Xu~\cite{MoselXu:ACM16} showed that, under certain condition, belief propagation with noisy label information has the same residual error as the maximum a-posteriori estimator for two symmetric communities. Cai {\em et. al}~\cite{Cai:BP-BEC} studied  weak recovery of two symmetric communities under belief propagation upon observing a vanishing fraction of labels. Neither~\cite{MoselXu:ACM16} nor~\cite{Cai:BP-BEC} establishes a converse. For two symmetric communities, Saad and Nosratinia~\cite{our2,Saad:JSTSP18} studied exact recovery under side information. Asadi~\cite{Asadi} studied the effect of i.i.d.\ vectors of side information on the phase transition of exact recovery for more than two communities. Kanade {\em et. al}~\cite{Kanade:GlobalLocal} showed that observation of a vanishing number of labels is unhelpful to {\em correlated recovery}\footnote{Correlated recovery denotes probability of error that is strictly better than a random guess, and is not a subject of this paper.} phase transition. For single community detection, Kadavankandy {\em et al.}~\cite{Kadavankandy:SingleCommunity} studied belief propagation with noisy label information with vanishing noise (unbounded LLRs).

The issue of side information in the context of single-community detection has not been addressed in the literature except for~\cite{Kadavankandy:SingleCommunity} whose results are generalized in this paper. Analyzing the effect of side information on information limit of weak recovery is a novel contribution of this work. A converse for the local algorithms such as belief propagation with side information has not been available prior to this work. The study of side information whose LLRs vary with $n$ is largely novel. And finally, while this work (inevitably) shares many tools and techniques with other works in the area of stochastic block models and community detection, the treatment of side information with variable LLR (as a function of $n$) presents new challenges for the bounding of errors by the application of Chernoff bound and large deviations, which are addressed in this work.


\section{System Model and Definitions}
\label{sys.}

Let $\boldsymbol G$ be a realization  from a random ensemble of graphs $\mathcal{G}(n,K,p,q)$, where each graph has $n$ nodes and contains a hidden community $C^{*}$ with size $|C^{*}| = K$. The underlying distribution of the graph is as follows:  an edge connects a pair of nodes with probability $p$ if both nodes are in $C^{*}$ and with probability $q$ otherwise. $G_{ij}$ is the indicator of an edge between nodes $i, j$. For each node $i$, a vector of dimension $M$ is observed consisting of side information, whose distribution depends on the label $x_i$ of the node. By convention $x_{i} = 1$ if $i \in C^{*}$ and $x_{i} = 0$ if $i \notin C^{*}$. For node $i$, the entries of the side information vector are each denoted $y_{i,m}$  and can be interpreted as different features of the side information. The side information for the entire graph is collected into the matrix ${\boldsymbol Y}_{n\times M}$.  The column vector $\boldsymbol{y}_{m}=[y_{1,m}, \ldots,y_{n,m}]^t$ collects the side information feature $m$ for all nodes $i$. 

 The vector of true labels is denoted $\boldsymbol{x}^{*} \in \{0,1\}^{n}$. $P$ and $Q$ are Bernoulli distributions with parameters $p, q$, respectively, and 
\[
L_{G}(i,j) = \log\Big(\frac{P(G_{ij})}{Q(G_{ij})}\Big)
\] is the log-likelihood ratio of edge $G_{ij}$ with respect to $P$ and $Q$.
  
In this paper, we address the problem of {\em single-community detection}, i.e.,  recovering $\boldsymbol{x}^{*}$ from $\boldsymbol G$ and $\boldsymbol Y$, under the following conditions: $K = o(n)$ while $\lim_{n\rightarrow\infty}K=\infty$, $p \geq q$, $\frac{p}{q} = \theta(1)$ and $\lim\sup_{n \to \infty} p < 1$.

An estimator $\hat{\boldsymbol{x}}(\boldsymbol{G},\boldsymbol{Y})$ is said to achieve exact recovery of $\boldsymbol{x}^*$ if, as $n \to \infty$, $\mathbb{P}(\hat{\boldsymbol{x}} = \boldsymbol{x}^{*}) \to 1$. An estimator $\hat{\boldsymbol{x}}(\boldsymbol{G},\boldsymbol{Y})$ is said to achieve weak recovery if, as $n \to \infty$, $\frac{d(\hat{\boldsymbol{x}},\boldsymbol{x}^{*})}{K} \to 0$ in probability, where $d(\cdot,\cdot)$ denotes the Hamming distance. It was shown in~\cite{infor_limits} that the latter definition is equivalent to the existence of an estimator $\hat{\boldsymbol{x}}$ such that $\mathbb{E}[d(\hat{\boldsymbol{x}},\boldsymbol{x}^{*})] = o(K)$. This equivalence will be used throughout our paper.


\section{Information Limits}
\label{info.lim.}

\subsection{Fixed-Quality Features}
\label{Varying_features}

In this subsection, the side information for each node is allowed to evolve with $n$ by having a varying number of independent and identically distributed scalar observations, each of which has a finite (imperfect) amount of information about the node label. By allowing the dimension of the side information per-node to vary and its scalar components to be identically distributed, the side information is represented with fixed-quality quanta. The results of this section demonstrate that as $n$ grows, the number of these side information quanta per-node must increase in a prescribed fashion in order to have a positive effect on the threshold for recovery.

For all $n$, for all $i=1,\ldots, n$, define the distributions:
\[
V(\upsilon) \triangleq \Prob(y_{i,m}=\upsilon|x_i=1)  \qquad U(\upsilon)\triangleq \Prob(y_{i,m}=\upsilon|x_i=-1)
\]
Thus the components of the side information for each node (features) are identically distributed for all nodes and all graph sizes $n$; we also assume all features are independent conditioned on the node labels $\boldsymbol{x}^*$. The dimension $M$ of the side information per node is allowed to vary as the size of the graph $n$ changes.

In addition, we assume $U,V$ are such that the resulting LLR random variable, defined below, has bounded support:
\[
L_{S}(i,m) = \log\Big(\frac{V(y_{i,m})}{U(y_{i,m})}\Big)
\] 
Throughout the paper, $L_S$ will continue to denote the LLR random variable of one side information feature, and $L_{G}$ denotes the random variable of the LLR of a graph edge.

\begin{definition}
\label{def.4}
\begin{align}
\psi_{QU}(t,m_{1},m_{2})  & \triangleq m_{1}\log(\mathbb{E}_{Q}[e^{tL_{G}}]) + m_{2}\log(\mathbb{E}_{U}[e^{tL_{S}}]) \\
\psi_{PV}(t,m_{1},m_{2})  & \triangleq m_{1}\log(\mathbb{E}_{P}[e^{tL_{G}}]) + m_{2}\log(\mathbb{E}_{V}[e^{tL_{S}}]) \\
E_{QU}(\theta,m_{1},m_{2}) &\triangleq \sup_{t\in[0,1]}  t\theta - \psi_{QU}(t,m_{1},m_{2}) \\ 
E_{PV}(\theta,m_{1},m_{2}) &\triangleq \sup_{t\in[-1,0]} t\theta - \psi_{PV}(t,m_{1},m_{2})
\end{align}
where $\theta$, $m_{1}$ and $m_{2}$ $\in \mathbb{R}$.
\end{definition}

\subsubsection{Weak Recovery}
\label{weak.lim.1}
 
\begin{Theorem}
\label{The.1.new}
For single community detection under bounded-LLR side information, weak recovery is achieved if and only if:
\begin{equation}
\begin{split}
(K-1) D(P||Q) &+ M D(V||U) \to \infty \text{ , } \\
\liminf_{n\to\infty} \Big[(K-1)D(P||Q) &+ 2M D(V||U) \Big]> 2\log(\frac{n}{K}) 
\end{split}
\label{Eq.Hajek}
\end{equation}
\ \ 
\end{Theorem}

\begin{proof}
For necessity please see Appendix~\ref{App.3}. For sufficiency, please see Appendix~\ref{App.4}.
\end{proof}

\begin{Remark}
\label{ex-ten2}
The condition of bounded support for the LLRs can be somewhat weakened to  Eqs.~\eqref{eq.lem2.1.1} and~\eqref{eq.lem2.1.4}.  As an example $U \sim \mathcal{N}(0,1)$ and $V \sim \mathcal{N}(\mu,1)$ with $\mu\neq 0$ satisfies~\eqref{eq.lem2.1.1}, \eqref{eq.lem2.1.4} and the theorem continues to hold even though the LLR is not bounded.
\end{Remark}

\begin{Remark}
\label{Re.1.new}
Theorem~\ref{The.1.new} shows that if $M$ grows with $n$ slowly enough, e.g., if $M$ is fixed and independent of $n$, or if $M = o(\log(\frac{n}{K}))$, side information does not affect the information limits. 
\end{Remark}

\begin{Remark}
\label{ex-ten}
If the features are conditionally independent but {\em not} identically distributed, it is easy to show the necessary and sufficient conditions are:
\begin{gather}
(K-1) D(P||Q) + \sum_{m=1}^{M} D(V_{m}||U_{m}) \to \infty \text{ , }  \nonumber\\
\liminf_{n\to\infty} (K-1)D(P||Q) + 2\sum_{m=1}^{M} D(V_{m}||U_{m}) > 2\log(\frac{n}{K}) \nonumber
\end{gather}
where $V_{m}$ and $U_{m}$ are analogous to $U$ and $V$ earlier, except specialized to each feature.
\end{Remark}

The assumption that the size of the community $|C^*|$ is known a-priori is not always reasonable: we might need to detect a small community whose size is not known in advance. In that case, the performance is characterized by the following lemma. 

\begin{lemma}
\label{Suff.Random}
For single-community detection under bounded-LLR side information, if the size of the community is not known in advance but obeys a probability distribution satisfying:
\begin{equation}
\mathbb{P}\Big(\Big | \; |C^{*}| - K \Big |  \leq \frac{K}{\log(K)}\Big) \geq 1 - o(1)
\end{equation}
for some known $K=o(n)$. If conditions~\eqref{Eq.Hajek} hold, then: 
\begin{equation}
\mathbb{P}\Big(\frac{ | \hat{C}\triangle C^{*} |}{K}  \leq 2\epsilon + \frac{1}{\log(K)}\Big) \geq 1 - o(1)
\end{equation}
where $$\epsilon = \big(\min(\log(K), (K-1)D(P||Q) + MD(V||U) )\big)^{-\frac{1}{2}} = o(1).$$
\end{lemma}

\begin{proof}
Please see Appendix~\ref{App.4.1}
\end{proof}

\subsubsection{Exact Recovery}
\label{exact.lim.1}

The sufficient conditions for exact recovery are derived using a two-step algorithm (see Table~\ref{Alg.1}). Its first step consists of any algorithm achieving weak recovery, e.g. maximum likelihood (see Lemma~\ref{Suff.Random}). The second step applies a local voting procedure.

\begin{table*}
\caption{Algorithm for exact recovery.}
\label{Alg.1}
\centering
\normalsize
\begin{tabular}{|p{4.5in}|}
\hline
Algorithm 1\\
\hline
\begin{enumerate}
\item 
Input: $n$, $K$, $\boldsymbol{G}$, $\boldsymbol{Y}$, $\delta \in (0,1): n\delta, \frac{1}{\delta} \in \mathbb{N}$.
\item
Consider a partition of the nodes $\{S_k\}$ with $|S_k|=n\delta$. $\boldsymbol{G}_{k}$  and
$\boldsymbol{Y}_{k}$ are the subgraph and side information corresponding to ${S_k}^c$, i.e., after each member of partition has been withheld.
\item
\label{Alg1-step3}
Consider estimator $\hat{C}_k(\boldsymbol{G}_{k}, \boldsymbol{Y}_{k})$ that produces $|\hat{C}_k|=\lceil K(1-\delta)\rceil$  and further assume it achieves weak recovery.

\item
For all $S_k$ and all $i\in S_k$ calculate $r_{i} = (\sum_{j \in \hat{C}_{k}} L_{G}(ij)) + \sum_{m=1}^{M} L_{S}(i,m)$ 

\item 
Output: $\tilde{C}=\{ \text{Nodes corresponding to $K$ largest }r_{i}\}$.
\end{enumerate}



\\
\hline
\end{tabular}
\end{table*}

\begin{lemma}
\label{The.5.new}
Define $C^*_k = C^* \cap {S_k}^c$ and assume $\hat{C}_k$ achieves weak recovery, i.e.
\begin{equation}
\mathbb{P}\big(|\hat{C}_{k} \triangle C^{*}_{k}| \leq \delta K \text{ for } 1\leq k \leq \frac{1}{\delta}\big) \to 1 \; .
\label{suff._exact.eq.1.new}
\end{equation}
If
\begin{equation}
\liminf_{n\to\infty} E_{QU}\big(\log(\frac{n}{K}),K,M\big) > \log(n)
\end{equation}
then $\mathbb{P}(\tilde{C} = C^{*}) \to 1$.
\end{lemma}

\begin{proof}
Please see Appendix~\ref{App.5}.
\end{proof}

Then the main result of this section follows:

\begin{Theorem}
\label{The.3.new}
In single community detection under bounded-LLR side information, assume~\eqref{Eq.Hajek} holds, then exact recovery is achieved if and only if: 
\begin{equation}
\label{Cond.1.The.3.new}
\liminf_{n\to\infty} E_{QU}\big(\log(\frac{n}{K}),K,M\big) > \log(n)
\end{equation}
\end{Theorem}

\begin{proof}
For sufficiency, please see Appendix~\ref{App.6}. For necessity see Appendix~\ref{App.7}.
\end{proof}

\begin{Remark}
\label{Re.2.12}
The assumption that~\eqref{Eq.Hajek} holds is necessary because otherwise weak recovery is not achievable, and by extension, exact recovery.
\end{Remark}

\begin{Remark}
\label{Re.2}
Theorem~\ref{The.3.new} shows if $M$ grows with $n$ slowly enough, e.g., $M$ is fixed and independent of $n$ or $M = o(K)$, side information will not affect the information limits of exact recovery.
\end{Remark}

To illustrate the effect of side information on information limits, consider the following example:
\begin{align}
K = \frac{c n}{\log(n)}, \quad q = \frac{b \log^{2}(n)}{n},  \quad p = \frac{a \log^{2}(n)}{n}
\end{align}
for positive constants $c, a \geq b$. Then, $KD(P||Q) = O(\log(n))$, and hence, weak recovery is achieved without side information, and by extension, with side information. Moreover, {\em exact} recovery without side information is achieved if and only if:
\begin{equation}
 \sup_{t\in [0,1]} tc(a-b) + bc - bc(\frac{a}{b})^{t} > 1 
\label{no_side}
\end{equation}
  
\begin{figure}
\centering
\includegraphics[width=\Figwidth]{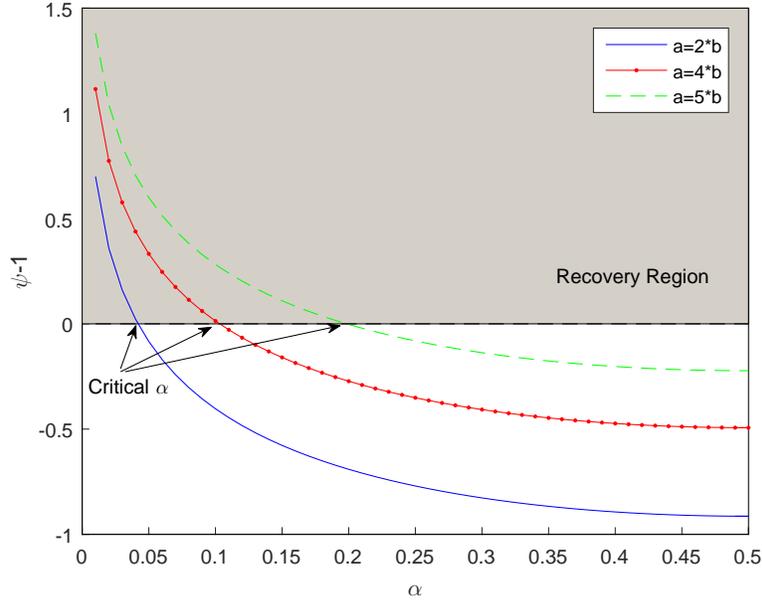}
\caption{Exact recovery threshold, $\psi-1$ for different values of $\alpha$ at $c=b=1$.}
\label{fig.1}
\end{figure}

Assume noisy label side information with error probability $\alpha \in (0,0.5)$. By Theorem~\ref{The.3.new}, exact recovery is achieved if and only if:
\begin{align}
 \sup_{t\in [0,1]} & tc(a-b) +  bc - bc(\frac{a}{b})^{t} - \twocolbreak
\frac{M}{\log(n)}\log( (1-\alpha)^{t}\alpha^{(1-t)} + (1-\alpha)^{(1-t)}\alpha^{t}) > 1 
\label{with_side}
\end{align} 
If $M = o(\log(n))$, then~\eqref{with_side} reduces to~\eqref{no_side}, thus  side information does not improve the information limits of exact recovery. If $M>o(\log(n))$, then $\log( (1-\alpha)^{t}\alpha^{(1-t)} + (1-\alpha)^{(1-t)}\alpha^{t})<0$ since $t\in[0,1]$. It follows that~\eqref{with_side} is less restrictive than~\eqref{no_side}, thus improving the information limit. 

Let $\psi$ denote the left hand side of~\eqref{with_side} with $M = \log(n)$, i.e.,
\begin{align}
\psi & = \sup_{t\in [0,1]} tc(a-b) + bc - bc(\frac{a}{b})^{t} \twocolbreak \includeonetwocol{}{\hspace{0.4in}} - \log( (1-\alpha)^{t}\alpha^{(1-t)} + (1-\alpha)^{(1-t)}\alpha^{t})
\end{align}
The behavior of $\psi$ against $\alpha$ describes the influence of side information on exact recovery and is depicted in Fig.~\ref{fig.1}. 
 
\subsection{Variable-Quality Features}
\label{Varying_noise} 

In this section, the number of features, $M$, is assumed to be constant but the LLR of each feature is allowed to vary with $n$. 

\subsubsection{Weak Recovery}
\label{weak.lim.}\hfill

Recall that the probability distribution side information feature $m$ is $V_{m}$  when the node is inside and outside the community, and $U_{m}$ when the node is outside the community.

\begin{Theorem}[Necessary Conditions for Weak Recovery]
\label{The.1}
For single community detection under bounded-LLR side information, weak recovery is achieved only if:
\begin{equation}
\begin{split}
&(K-1) D(P||Q) + \sum_{m=1}^{M} (D(V_{m}||U_{m})+D(U_{m}||V_{m})) \to \infty \\
&\liminf_{n\to\infty} (K-1)D(P||Q) + 2\sum_{m=1}^{M}D(V_{m}||U_{m}) \geq 2\log(\frac{n}{K}) 
\label{Cond.2.The.1}
\end{split}
\end{equation}
\end{Theorem}

\ \ 

\begin{proof}
The proof follows similar to Theorem~\ref{The.1.new}.
\end{proof}

\subsubsection{Exact Recovery}
\label{exact.lim.}\hfill

We begin by concentrating on the following regime, and will subsequently show its relation to the set of problems that are both feasible and interesting.
\begin{equation}
\label{reg.1}
K = \rho \frac{n}{\log(n)}, \qquad p = a \frac{\log(n)^{2}}{n} \qquad  q = b \frac{\log(n)^{2}}{n}
\end{equation}
with constants $\rho \in (0,1)$ and $a \geq b > 0$. 

The alphabet for each feature $m$ is denoted with $\{u_{1}^{m}, u_{2}^{m}, \cdots, u_{L_{m}}^{m}\}$, where $L_{m}$ is the cardinality of feature $m$ which, in this section, is assumed to be bounded and constant across $n$. The likelihoods of the features are defined as follows:
\begin{align}
\alpha_{+,\msubK}^{m} \triangleq \mathbb{P}(y_{i,m} = u_{\msubK}^{m} | x_{i} = 1) \\
 \alpha_{-,\msubK}^{m} \triangleq \mathbb{P}(y_{i,m} = u_{\msubK}^{m} | x_{i} = 0) 
\end{align}
Recall that in our side information model, all features are independent conditioned on the labels. To ensure that the quality of the side information is increasing with $n$, both $\alpha_{+,\msubK}^{m}$ and $\alpha_{-,\msubK}^{m}$ are assumed to be either constant or monotonic in $n$.

To better understand the behavior of information limits, we categorize side information outcomes based on the trends of LLR and likelihoods. For simplicity we speak of trends for one feature; extension to multiple features is straight forward. An outcome is called \textit{informative} if $h_{\ell} = O(\log(n))$ and \textit{non-informative} if $h_{\ell} = o(\log(n))$. An outcome is called \textit{rare} if $\log(\alpha_{\pm,\ell}) = O(\log(n))$ and \textit{not rare} if $\log(\alpha_{\pm,\ell}) = o(\log(n))$. Among the four different combinations, the \textit{worst} case is when the outcome is both \textit{non-informative} and \textit{not rare} for nodes inside and outside the community. We will show that if such an outcome exists, then side information will not improve the information limit. The \textit{best} case is when the outcome is \textit{informative} and \textit{rare} for the nodes inside the community, or for the nodes outside the community, but not both. Two cases are in between: (1) an outcome that is \textit{non-informative} and \textit{rare} for nodes inside and outside the community and (2) an outcome that is \textit{informative} and \textit{not rare} for nodes inside and outside the community. It will be shown that the last three cases can affect the information limit under certain conditions. 

For convenience we define:
\begin{align}
  T  &\triangleq \log\big(\frac{a}{b}\big) 
  \end{align}
 We introduce the following functions whose value, as shown in the sequel, characterizes the exact recovery threshold:
  \begin{align}
	\eta_{1}(\rho, a,b) &\triangleq \rho\Big (b + \frac{a-b}{T} \log\big(\frac{a-b}{ebT}\big)\Big ) \label{eta1}\\
  \eta_{2}(\rho, a,b,\beta) &\triangleq \rho b + \frac{\rho(a-b)-\beta}{T} \log\big(\frac{\rho(a-b) -\beta}{\rho ebT}\big)  + \beta \label{eta2}\\
	\eta_{3}(\rho, a,b,\beta) &\triangleq \rho b + \frac{\rho(a-b)+\beta}{T} \log\big(\frac{\rho(a-b) +\beta}{\rho ebT}\big) \label{eta3} 
	\end{align}
For example in the regime~\eqref{reg.1}, one can conclude using~\eqref{Cond.1.The.3.new} that exact recovery without side information is achieved if and only if $\eta_1>1$.


The LLR of each feature is denoted:
\begin{equation}
    \hfeat  \triangleq \log\Big(\frac{\alpha_{+,\msubK}^{m}}{\alpha_{-,\msubK}^{m}}\Big )  \label{h}
\end{equation}
We also define the following functions of the likelihood and LLR of side information, whose evolution with $n$ is critical to the phase transition of exact recovery~\cite{Saad:JSTSP18}.
\begin{align}
f_1(n) &\triangleq \sum_{m=1}^M \hfeat, \\
f_2(n) &\triangleq\sum_{m=1}^M \log(\alpha_{+,\msubK}^{m}), \\
f_3(n) &\triangleq\sum_{m=1}^M \log(\alpha_{-,\msubK}^{m})
\end{align}
In the following, the side information outcomes $[u_{\ell_1}^1, \ldots,u_{\ell_M}^M]$ are represented by their index $[\ell_1,\ldots,\ell_M]$ without loss of generality. Throughout, dependence on $n$ of outcomes and their likelihood is implicit.

\begin{Theorem}
\label{The.3.5}
In the regime characterized by~\eqref{reg.1}, assume $M$ is constant and $\alpha_{+,\msubK}^{m}$ and $\alpha_{-,\msubK}^{m}$ are either constant or monotonic in $n$. Then, necessary and sufficient conditions for exact recovery depend on side information statistics  in the following manner:
\begin{enumerate}
\item If there exists any sequence (over $n$) of  side information outcomes $[\ell_1,\ldots,\ell_M]$ such that $f_1(n)$, $f_2(n)$, $f_3(n)$ are all $o(\log(n))$,
then $\eta_{1}(\rho,a,b) > 1$ must hold.

\item 
\label{Th3.item2} If there exists any sequence (over $n$) of side information outcomes $[\ell_1,\ldots,\ell_M]$ such that $f_1(n)= o(\log(n))$ and $f_2(n), f_3(n)$ evolve according to $-\beta\log(n)+ o(\log(n))$ with $\beta > 0$, then $\eta_{1}(\rho,a,b) + \beta > 1$  must hold.

\item 
\label{Th3.item3} If there exists any sequence (over $n$) of side information outcomes $[\ell_1,\ldots,\ell_M]$ such that $f_1(n) = \beta_{1}\log(n)+ o(\log(n))$ with $ 0 < \beta_1 < \rho(a-b-bT)$ and furthermore $f_2(n)=o(\log(n))$, then $\eta_{2}(\rho,a,b,\beta_1) > 1$ must hold.

\item If there exists any sequence (over $n$) of side information outcomes $[\ell_1,\ldots,\ell_M]$ such that $f_1(n) = \beta_{2}\log(n)+ o(\log(n))$ with $ 0 < \beta_2 < \rho(a-b-bT)$ and furthermore $f_3(n)=o(\log(n))$, then $\eta_{3}(\rho,a,b,\beta_2) > 1$ must hold.

\item If there exists any sequence (over $n$) of side information outcomes $[\ell_1,\ldots,\ell_M]$ such that $f_1(n) = \beta_{3}\log(n)+ o(\log(n))$ with $ 0 < \beta_3 < \rho(a-b-bT)$ and furthermore $f_2(n)=-\beta_{3}'\log(n) + o(\log(n))$, then $\eta_{2}(\rho,a,b,\beta_3) + \beta_{3}' > 1$ must hold.

\item If there exists any sequence (over $n$) of side information outcomes $[\ell_1,\ldots,\ell_M]$ such that $f_1(n) = \beta_{4}\log(n)+ o(\log(n))$ with $ 0 < \beta_4 < \rho(a-b-bT)$ and furthermore $f_3(n)=-\beta_{4}'\log(n) + o(\log(n))$, then $\eta_{3}(\rho,a,b,\beta_4) + \beta_{4}' > 1$ must hold.
\end{enumerate}
\end{Theorem}

\begin{proof}
For necessity, see Appendix~\ref{App.8}. For sufficiency, see Appendix~\ref{App.10}.
\end{proof}

\begin{Remark}
The six items in Theorem~\ref{The.3.5} are concurrent. For example, if some side information outcome sequences fall under Item~\ref{Th3.item2} and some fall under Item~\ref{Th3.item3}, then the necessary and sufficient condition for exact recovery is $\min (\eta_{1}(\rho,a,b,\beta), \eta_{2}(\rho,a,b,\beta_1)) > 1$.
\end{Remark}

\begin{Remark}
Theorem~\ref{The.3.5} does not address  $f_1(n) = \omega(\log(n))$ because it leads to a trivial problem. For example, for noisy label side information, if the noise parameter $\alpha = e^{-n}$, then side information alone is sufficient for exact recovery. Also, when $f_1(n) = \beta \log(n)$ with $| \beta | \geq \rho(a-b-bT)$, a necessary condition is easily obtained but a matching sufficient condition for this case remains unavailable.
\end{Remark}

In the following, we specialize the results of Theorem~\ref{The.3.5} to noisy-labels and partially-revealed-label side information.
\begin{Corollary}
For side information consisting of noisy labels with error probability $\alpha \in (0,0.5)$, Theorem~\ref{The.3.5} combined with Lemma~\ref{rem.f1.f2} state that exact recovery is achieved if and only if:
$$
\begin{cases}
\eta_{1}(\rho, a,b) > 1,  &\text{ when } \log(\frac{1-\alpha}{\alpha}) = o(\log(n)) \label{suff.BSC.1}  \\ 
\eta_{2}(\rho, a,b,\beta) > 1, &\text{ when } \log(\frac{1-\alpha}{\alpha}) = (\beta+o(1))\log(n), \twocolbreak \quad \quad 0<\beta<\rho(a-b-bT) \label{suff.BSC.2}
\end{cases}
$$
\end{Corollary}

Figure~\ref{fig.2} shows the error exponent for the noisy label side information as a function of $\beta$.
\begin{figure}
\centering
\includegraphics[width=\Figwidth]{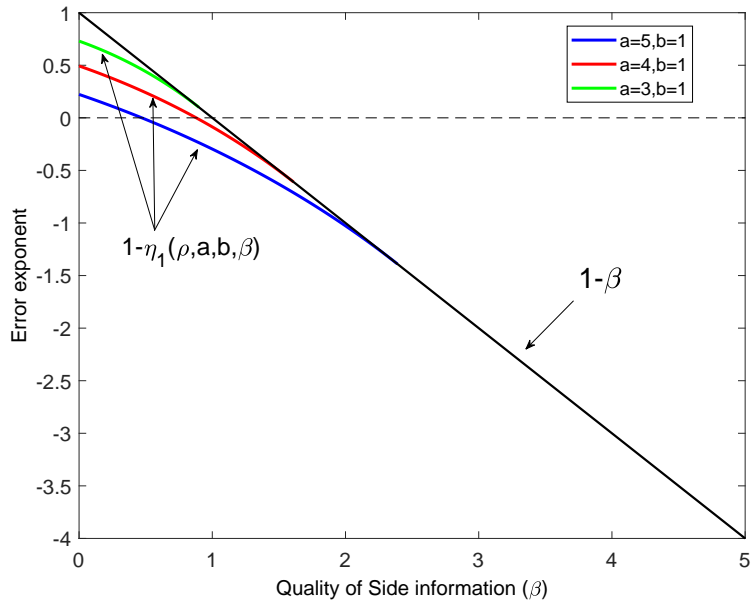}
\caption{Error exponent for noisy side information.}
\label{fig.2}
\end{figure}

\begin{Corollary}
For side information consisting of a fraction $1-\epsilon$ of the labels revealed, Theorem~\ref{The.3.5} states that exact recovery is achieved if and only if:
$$
\begin{cases}
\eta_{1}(\rho, a,b) > 1, &\text{ when } \log(\epsilon) = o(\log(n)) \label{suff.BEC.1}  \\ 
\eta_{1}(\rho, a,b) + \beta > 1, &\text{ when } \log(\epsilon) = (-\beta+o(1))\log(n), \twocolbreak \qquad \beta>0  \label{suff.BEC.2}
\end{cases}
$$
\end{Corollary}

Figure~\ref{fig.3} shows the error exponent for partially revealed labels, as a function of $\beta$.
\begin{figure}
\centering
\includegraphics[width=\Figwidth]{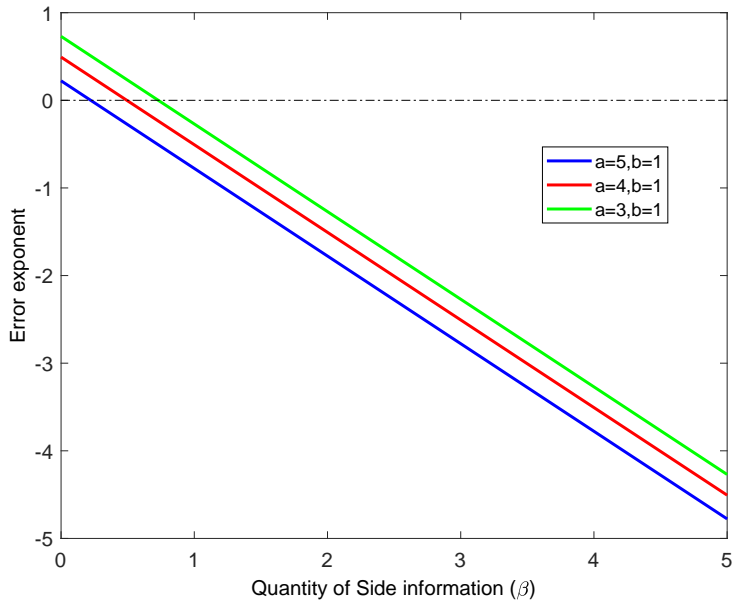}
\caption{Error exponent for partially revealed side information.}
\label{fig.3}
\end{figure}

We now comment on the coverage of the regime~\eqref{reg.1}. If the average degree of a node is $o(\log n)$, then the graph will have isolated nodes and exact recovery is impossible. If the average degree of the node is $\omega (\log n)$, then the problem is trivial. Therefore the regime of interest is when the average degree is $\Omega(\log n)$. This restricts $Kp$ and $Kq$ in a manner that is reflected in~\eqref{reg.1}. Beyond that, in the system model of this paper $K=o(n)$, so $\frac{\log(\frac{n}{K})}{\log(n)}$ is either $o(1)$ or approaching a constant $C \in (0,1]$. The regime~\eqref{reg.1} focuses on the former, but the proofs are easily modified to cover the latter. For the convenience of the reader, we highlight the places in the proof where a modification is necessary to cover the latter case.


\section{Belief Propagation}
\label{BP}
Belief propagation for recovering a single community was studied {\em without} side information in~\cite{recovering_one,montnari} in terms of a signal-to-noise ratio parameter $\lambda = \frac{K^{2}(p-q)^{2}}{(n-k)q}$, showing that {\em weak recovery} is achieved if and only if $\lambda > \frac{1}{e}$. Moreover, belief propagation followed by a local voting procedure was shown to achieve {\em exact recovery} if  $\lambda > \frac{1}{e}$, as long as information limits allow exact recovery. 

In this section $M=1$, i.e. we consider scalar side information random variables that are discrete and take value from an alphabet size $L$. Extension to a vector side information is straight forward as long as dimensionality is constant across $n$; the extension is outlined in Corollary~\ref{Extension}.

Denote the expectation of the likelihood ratio of the side information conditioned on $x = 1$ by:
\begin{equation}
\ChiS \triangleq  \sum_{\ell=1}^{L} \frac{\alpha_{+,\ell}^{2}}{\alpha_{-,\ell}}
\end{equation}
By definition, $\ChiS= \tilde{\chi}^2 +1$, where $\tilde{\chi}^2$ is the chi-squared divergence between the conditional distributions of side information. Thus, $\ChiS \geq 1$.

\subsection{Bounded LLR}
\label{BP.bounded}

We begin by demonstrating the performance of belief propagation algorithm on a random tree with side information. Then, we show that the same performance is possible on a random {\em graph} drawn from $\mathcal{G}(n,K,p,q)$, using a coupling lemma~\cite{recovering_one} expressing local approximation of random graphs by trees.

\subsubsection{Belief Propagation on a Random Tree with Side Information} 
\label{BP.tree}\hfill

We model random trees with side information in a manner roughly parallel to random graphs. Let $T$ be an infinite tree with nodes $i$, each of them possessing a label $\tau_{i} \in \{0,1\}$.  The root is node $i=0$. The subtree of depth $t$ rooted at node $i$ is denoted $T_{i}^{t}$. For brevity, the subtree rooted at $i=0$ with depth $t$ is denoted $\CroppedTree$. Unlike the random graph counterpart, the tree and its node labels are generated together as follows: $\tau_{0}$ is a Bernoulli-$\frac{K}{n}$ random variable. For any $i\in T$, the number of its children with label $1$ is a random variable $H_i$ that is Poisson with parameter $Kp$ if $\tau_i=1$, and Poisson with parameter $Kq$ if $\tau_i=0$. The number of children of node $i$ with label $0$ is a random variable $F_i$ which is Poisson with parameter $(n-K)q$, regardless of the label of node $i$. The side information $\tilde{\tau}_i$ takes value in a finite alphabet $\{u_{1},\cdots,u_{L}\}$. The set of all labels in $T$ is denoted with $\TreeLabels$, all side information with $\TreeSideInfo$, and the labels and side information of $\CroppedTree$ with $\TreeLabels^t$ and $\TreeSideInfo^t$ respectively.  The likelihood of side information continues to be denoted by $\alpha_{+,\ell},\alpha_{-,\ell}$, as earlier.

The problem of interest is to infer the label $\RootLabel$ given observations $\CroppedTree$ and $\TreeSideInfo^t$. The error probability of an estimator $\RootEstimator(\CroppedTree,\TreeSideInfo^t)$ can be written as:
\begin{align}
\label{error.general}
p_{e}^{t} &\triangleq \frac{K}{n} \mathbb{P}(\RootEstimator = 0 | \RootLabel = 1) + \frac{n-K}{n} \mathbb{P}(\RootEstimator = 1 | \RootLabel = 0)
\end{align}
The maximum a posteriori (MAP) detector minimizes $p_{e}^{t}$ and can be written in terms of the log-likelihood ratio as $\hat{\tau}_{MAP} = {\mathbbm 1}_{\{\LLRCroppedTree \geq \nu\}}$, where $\nu = \log(\frac{n-K}{K})$ and:
\begin{align}
\label{likeli}
\LLRCroppedTree = \log\bigg(\frac{\mathbb{P}(\CroppedTree,\TreeSideInfo^t | \RootLabel = 1)}{\mathbb{P}(\CroppedTree,\TreeSideInfo^t | \RootLabel = 0)}\bigg)
\end{align}
The probability of error of the MAP estimator can be bounded as follows~\cite{error2}:
\begin{align}
\label{bound.main}
\frac{K(n-K)}{n^{2}} \rho^{2} &\leq p_{e}^{t} \leq \frac{\sqrt{K(n-K)}}{n} \rho
\end{align}
where $\rho = \mathbb{E}\big[ e^{\frac{\LLRCroppedTree}{2}} \big| \RootLabel = 0\big]$. 

\begin{lemma}
\label{recursive}
Let ${\mathcal N}_{i}$ denote the children of node $i$, $N_i\triangleq|{\mathcal N}_i|$ and $h_{i} \triangleq \log\big(\frac{\mathbb{P}(\tilde{\tau}_{i} | \tau_i =1)}{\mathbb{P}(\tilde{\tau}_{i} | \tau_i =0)}\big)$. Then,
\begin{align}
\Gamma_{i}^{t+1} &= -K(p-q) + h_{i} + \sum_{k \in {\mathcal N}_{i}} \log\bigg(\frac{\frac{p}{q}e^{\Gamma_{k}^{t} - \nu}+1}{e^{\Gamma_{k}^{t} - \nu}+1}\bigg)
\end{align}
\end{lemma}

\begin{proof}
See Appendix~\ref{BP-App.1}
\end{proof}

\paragraph{Lower and Upper Bounds on $\rho$}\hfill

Define for $ t\geq 1$ and any node $i$:
\begin{align}
\psi_{i}^{t} &= -K(p-q) +  \sum_{j \in {\mathcal N}_i} M(h_{j} + \psi_{j}^{t-1}) \label{psi.1}
\end{align} 
where 
\[
M(x) \triangleq \log\Big(\frac{\frac{p}{q}e^{x-\nu} + 1}{e^{x-\nu}+1}\Big) = \log\Big(1+ \frac{\frac{p}{q}-1}{1 + e^{-(x-\nu)}}\Big).
\]
Then, $\Gamma_{i}^{t+1} = h_{i} + \psi_{i}^{t+1}$ and $\psi_{i}^{0} = 0$ $\forall i \in \CroppedTree$. Let $Z_{0}^{t}$ and $Z_{1}^{t}$ denote random variables drawn according to the distribution of $\psi_{i}^{t}$ conditioned on $\tau_i= 0$ and $\tau_i = 1$, respectively. Similarly, let $U_{0}$ and $U_{1}$ denote  random variables drawn according to the distribution of $h_{i}$ conditioned on $\tau_i=0$ and $\tau_i = 1$, respectively. Thus, $\rho = \mathbb{E}\big[e^{\frac{1}{2}(Z_{0}^{t}+U_{0})}\big] = \mathbb{E}\big[e^{\frac{U_{0}}{2}}\big] \mathbb{E}\big[e^{\frac{Z_{0}^{t}}{2}}\big]$. 
Define:\begin{align}
b_{t} &\triangleq \mathbb{E}\Big[\frac{e^{Z_{1}^{t}+U_{1}}}{1+ e^{Z_{1}^{t}+U_{1}-\nu} }\Big] \label{b_t}\\
a_{t} &\triangleq \mathbb{E}\big[e^{Z_{1}^{t}+U_{1}}\big] \label{a_t}
\end{align} 

\begin{lemma}
\label{lower.upper}
Let $B = (\frac{p}{q})^{1.5}$. Then:
\begin{equation}
\mathbb{E}[e^{\frac{U_{0}}{2}}] e^{\frac{-\lambda}{8}b_{t}} \leq \rho \leq \mathbb{E}[e^{\frac{U_{0}}{2}}] e^{\frac{-\lambda}{8B}b_{t}}  
\end{equation}
\end{lemma}

\begin{proof}
See Appendix~\ref{BP-App.2}.
\end{proof}

Thus to bound $\rho$, lower and upper bounds on $b_{t}$ are needed. 

\begin{lemma}
\label{upper_b}
For all $t \geq 0$, if $\lambda \leq \frac{1}{\ChiS e}$, then $b_{t} \leq \ChiS e$.
\end{lemma}

\begin{proof}
See Appendix~\ref{BP-App.3}.
\end{proof}

\begin{lemma}
\label{Lem.lower.b.1}
Define $C = \lambda(2+\frac{p}{q})$ and $\ChiS' = \mathbb{E}[e^{3U_{0}}]$. Assume that $b_{t} \leq \frac{\nu}{2(C-\lambda)}$. Then,
\begin{equation} 
\label{lower.eq.1}
b_{t+1}  \geq \ChiS e^{\lambda b_{t}} (1 - \frac{\ChiS'}{\ChiS} e^{\frac{-\nu}{2}})
\end{equation}
\end{lemma}

\begin{proof}
See Appendix~\ref{BP-App.4}.
\end{proof}

\begin{lemma}
\label{Lem.lower.b.2}
The sequences $a_{t}$ and $b_{t}$ are non-decreasing in $t$.
\end{lemma}

\begin{proof}
The proof follows directly from~\cite[Lemma 5]{recovering_one}.
\end{proof}

\begin{lemma}
\label{Lem.lower.b.main}
Define $\log^{*}(\nu)$ to be the number of times the logarithm function must be iteratively applied to $\nu$ to get a result less than or equal to one. Let $C = \lambda(2+\frac{p}{q})$ and $\ChiS' = \mathbb{E}[e^{3U_{0}}]$. Suppose $\lambda > \frac{1}{\ChiS e}$. Then there are constants $\bar{t}_{o}$ and $\nu_{o}$ depending only on $\lambda$ and $\ChiS$ such that:
\begin{equation}
\label{lower.eq.2}
b_{\bar{t}_{o} + \log^{*}(\nu)+2} \geq \ChiS e^{\frac{\lambda \nu}{2(C-\lambda)}} (1-\frac{\ChiS'}{{\ChiS}} e^{\frac{-\nu}{2}})
\end{equation}
whenever $\nu \geq \nu_{o}$ and $\nu \geq 2\ChiS (C-\lambda)$.
\end{lemma}

\begin{proof}
See Appendix~\ref{BP-App.5}.
\end{proof}

\paragraph{Achievability and Converse for the MAP Detector}

\begin{lemma}
\label{lower_up_map}
Let $\ChiS' = \mathbb{E}[e^{3U_{0}}]$, $C = \lambda(2+\frac{p}{q})$ and $B = (\frac{p}{q})^{1.5}$.  If $ 0 < \lambda \leq \frac{1}{\ChiS e}$, then:
\begin{align} 
\label{lower_p_t}
p_{e}^{t} \geq \frac{K(n-K)}{n^{2}} \mathbb{E}^{2}[e^{\frac{U_{0}}{2}}] e^{\frac{-\lambda \ChiS e}{4}}
\end{align}

If $\lambda > \frac{1}{\ChiS e}$, then:
\begin{align}
\label{upper_p_t}
p_{e}^{t} \leq \sqrt{\frac{K(n-K)}{n^{2}}} \mathbb{E}[e^{\frac{U_{0}}{2}}] e^{  \frac{-\lambda \ChiS}{8B}  e^{\frac{\lambda \nu}{2(C-\lambda)}} (1-\frac{\ChiS'}{{\ChiS}} e^{\frac{-\nu}{2}}) }
\end{align}
Moreover, since $\nu \to \infty$:
\begin{align}
\label{upper_p_t2}
p_{e}^{t} \leq \sqrt{\frac{K(n-K)}{n^{2}}} \mathbb{E}[e^{\frac{U_{0}}{2}}] e^{  -\nu(r + \frac{1}{2}) } = \frac{K}{n} e^{-\nu (r + o(1))}
\end{align}
for some $r >0$.
\end{lemma}

\begin{proof}
The proof follows directly from~\eqref{bound.main} and Lemmas~\ref{upper_b} and~\ref{Lem.lower.b.main}.
\end{proof}

\subsubsection{Belief Propagation Algorithm for Community Recovery with Side Information}
\label{BP.hidden.bounded} \hfill

In this section, the inference problem defined on the random tree is coupled to the problem of recovering a hidden community with side information. This can be done via a coupling lemma~\cite{recovering_one} that shows that under certain conditions, the neighborhood of a fixed node $i$ in the graph is locally a tree with probability converging to one, and hence, the belief propagation algorithm defined for random trees in Section~\ref{BP.tree} can be used on the graph as well. The proof of the coupling lemma depends only on the tree structure, implying that it also holds for our system model, where the side information is independent of the tree structure given the labels. 

Define $\boldsymbol{G}_{u}^{\TreeIter}$ to be the subgraph containing all nodes that are at a distance at most $\TreeIter$ from node $u$ and define $\boldsymbol{x}_{u}^{\TreeIter}$ and $\boldsymbol{Y}_{u}^{\TreeIter}$ to be the set of labels and side information of all nodes in $\boldsymbol{G}_{u}^{\TreeIter}$, respectively. 

\begin{lemma}[Coupling Lemma~\cite{recovering_one}]
\label{couple}
Suppose that $\TreeIter(n)$ are positive integers such that  $(2+np)^{\TreeIter(n)} = n^{o(1)}$. Then:
\begin{itemize}
    \item 
If the size of community is deterministic and known, i.e., $|C^{*}| = K$, then for any node $u$ in the graph, there exists a coupling between $(\boldsymbol{G},\boldsymbol{x},\boldsymbol{Y})$ and $(T, \TreeLabels, \TreeSideInfo)$ such that:
\begin{equation}
\label{couple.eq.1}
\mathbb{P}( (\boldsymbol{G}_{u}^{\TreeIter}, \boldsymbol{x}_{u}^{\TreeIter},\boldsymbol{Y}_{u}^{\TreeIter} ) = (T^{\TreeIter}, \TreeLabels^{\TreeIter}, \TreeSideInfo^{\TreeIter}) ) \geq 1 - n^{-1 + o(1)}
\end{equation}
where for convenience of notation, the dependence of $\TreeIter$ on $n$ is made implicit.
\item
If $|C^{*}|$ obeys a probability distribution so that $\mathbb{P}(| |C^{*}| - K | \geq \sqrt{3K\log(n)}) \leq n^{\frac{-1}{2} + o(1)}$ with $K \geq 3\log(n)$, then for any node $u$, there exists a coupling between $(\boldsymbol{G},\boldsymbol{x},\boldsymbol{y})$ and $(T, \TreeLabels, \TreeSideInfo)$ such that:
\begin{equation}
\label{couple.eq.2}
\mathbb{P}( (\boldsymbol{G}_{u}^{\TreeIter}, \boldsymbol{x}_{u}^{\TreeIter},\boldsymbol{Y}_{u}^{\TreeIter} ) = (T^{\TreeIter}, \TreeLabels^{\TreeIter}, \TreeSideInfo^{\TreeIter}) )
\geq 1 - n^{\frac{-1}{2} + o(1)}
\end{equation}
\end{itemize}
\end{lemma}

Now, we are ready to present the belief propagation algorithm for community recovery with bounded side information. Define the message transmitted from node $i$ to its neighboring node $j$ at iteration $t+1$ as:
\begin{align}
\label{BP.1}
R_{i \to j}^{t+1} = h_{i} - K(p-q) + \sum_{k \in {\mathcal N}_i\backslash j} M(R_{k \to i}^{t})
\end{align}
where $h_{i} = \log(\frac{\mathbb{P}(y_{i}|x_{i}=1)}{\mathbb{P}(y_{i}|x_{i}=0)})$, ${\mathcal N}_{i}$ is the set of neighbors of node $i$ and $M(x) = \log(\frac{\frac{p}{q}e^{x-\nu} + 1}{e^{x-\nu}+1})$. The messages are initialized to zero for all nodes $i$, i.e., $R_{i \to j}^{0} = 0$ for all $i \in \{ 1,\cdots,n\}$ and $j \in {\mathcal N}_i$. Define the belief of node $i$ at iteration $t+1$ as:
\begin{align}
\label{BP.2}
R_{i}^{t+1} = h_{i} - K(p-q) + \sum_{k \in {\mathcal N}_i} M(R_{k \to i}^{t})
\end{align}

Algorithm~\ref{tab} presents the proposed belief propagation algorithm for community recovery with side information.
\begin{table*}
\caption{Belief propagation algorithm for community recovery with side information.}
\label{tab}
\centering
\normalsize
\begin{tabular}{|p{4.5in}|}
\hline
Belief Propagation Algorithm \\
\hline
\begin{enumerate}
    \item Input: $n, K, \TreeDepth\in\mathbb{N}$, $\boldsymbol{G}$ and $\boldsymbol{Y}$.
    \item For all nodes $i$ and $j \in {\mathcal N}_i$, set $R^{0}_{i\to j} = 0$.
    \item For all nodes $i$ and $j \in {\mathcal N}_i$, run $\TreeDepth-1$ iterations of belief propagation as in~\eqref{BP.1}.
    \item For all nodes $i$, compute its belief $R_{i}^{\TreeDepth}$ based on~\eqref{BP.2}.
    \item Output $\tilde{C}=\{ \text{Nodes corresponding to $K$ largest }R_{i}^{\TreeDepth}\}$.
\end{enumerate}




\\
\hline
\end{tabular}
\end{table*}

If in Algorithm~\ref{tab} we have $\TreeDepth=\TreeIter(n)$, according to Lemma~\ref{couple} with probability converging to one $R_{i}^{\TreeDepth} = \Gamma_{i}^{\TreeDepth}$, where $\Gamma_{i}^{\TreeDepth}$ was the log-likelihood defined for the random tree. Hence, the performance of Algorithm~\ref{tab} is expected to be the same as the MAP estimator defined as $\hat{\tau}_{MAP} = {\mathbbm 1}_{\{\Gamma_{i}^{\TreeDepth} \geq \nu\}}$, where $\nu = \log(\frac{n-K}{K})$. The only difference is that the MAP estimator decides based on $\Gamma_{i}^{\TreeDepth} \geq \nu$ while Algorithm~\ref{tab} selects the $K$ largest $R_{i}^{\TreeDepth}$. To manage this difference, let $\hat{C}$ define the community recovered by the MAP estimator, i.e. $\hat{C} = \{i: R_{i}^{\TreeDepth} \geq \nu \}$. Since $\tilde{C}$ is the set of nodes with the $K$ largest $R_{i}^{\TreeDepth}$. Then, 
\begin{align}
| C^{*} \triangle \tilde{C} | & \leq | C^{*} \triangle \hat{C} | + | \hat{C} \triangle \tilde{C} | \nonumber\\ 
& = | C^{*} \triangle \hat{C} | + | |\hat{C}| - K | \label{fix.est.1} 
\end{align}
Moreover,
\begin{align}
||\hat{C}| - K | & \leq  ||\hat{C}| - |C^{*}| | + ||C^{*}| - K | \twocolbreak \leq | C^{*} \triangle \hat{C} | + ||C^{*}| - K | \label{fix.est.2}
\end{align}
Using~\eqref{fix.est.2} and substituting in~\eqref{fix.est.1}:
\begin{align}
| C^{*} \triangle \tilde{C} | & \leq 2 | C^{*} \triangle \hat{C} | + |  |C^{*}| - K|  \label{fix.est}
\end{align}
We will use~\eqref{fix.est} to prove weak recovery.

\paragraph{Weak Recovery}
\label{BP_weak}

\begin{Theorem}[Achievability]
\label{The.6}
Suppose that $(np)^{\log^{*}(\nu)} = n^{o(1)}$ and $\lambda > \frac{1}{\ChiS e}$. Let $\TreeIter(n) = \bar{t}_{o} + \log^{*}(\nu) +2$, where $\bar{t}_{o}$ is a constant depending only on $\lambda$ and $\ChiS$. Apply Algorithm~\ref{tab} with $\TreeDepth=\TreeIter(n)$ resulting in estimated community $\tilde{C}$. Then:
\begin{equation}
\frac{\mathbb{E}[|C^{*} \triangle \tilde{C} |]}{K} \to 0
\end{equation}
for either $|C^{*}| = K$ or random $|C^{*}|$ such that $K \geq 3\log(n)$ and $\mathbb{P}(| |C^{*}| - K | \geq \sqrt{3K\log(n)}) \leq n^{\frac{-1}{2} + o(1)}$.
\end{Theorem}

\begin{proof}
See Appendix~\ref{BP-App.5.1}.
\end{proof}

\begin{Theorem}[Converse]
\label{The.7}
Suppose that $\lambda \leq \frac{1}{\ChiS e}$. Let $\TreeIter \in \mathbb{N}$ depend on $n$ such that $(2+np)^{\TreeIter} = n^{o(1)}$. Then, for any local estimator $\hat{C}$ of $x^{*}_{u}$ that has access to observations of the graph and side information limited to a neighborhood of radius  $\TreeIter$ from $u$,
\begin{align}
\frac{\mathbb{E}[|C^{*} \triangle \hat{C} |]}{K} \geq (1-\frac{K}{n}) \mathbb{E}^{2}[e^{\frac{U_{0}}{2}}] e^{\frac{-\lambda \ChiS e}{4}} - o(1)
\end{align} 
\end{Theorem}

\begin{proof}
See Appendix~\ref{BP-App.5.2}.
\end{proof}

\begin{Corollary}
\label{Extension}
The same result holds for side information consisting of multiple features, i.e., constant $M \geq 1$. In other words, using the same notation as in Section~\ref{exact.lim.}, weak recovery is possible if and only if $\lambda > \frac{1}{\ChiS e}$ where $\ChiS = \sum_{\ell_{1}=1}^{L_{1}} \cdots \sum_{\ell_M=1}^{L_M} (\prod_{m=1}^{M} \frac{(\alpha_{+,\msubK}^{m})^{2}}{\alpha_{-,\msubK}^{m}})$.
\end{Corollary}

\paragraph{Exact Recovery}
\label{BP_exact}\hfill

In Section~\ref{exact.lim.1}, it was shown that under certain conditions any estimator that achieves weak recovery on a random cluster size will also achieve exact recovery if followed by a local voting process. This can be used to demonstrate sufficient conditions for exact recovery under belief propagation. To do so, we employ a modified form of the algorithm in Table~\ref{Alg.1}, where in Step~\ref{Alg1-step3} for weak recovery we use the belief propagation algorithm presented in Table~\ref{tab}.







\begin{Theorem}
\label{The.8}
Suppose that $(np)^{\log^{*}(\nu)} = n^{o(1)}$ and $\lambda > \frac{1}{\ChiS e}$. Let $\delta \in (0,1)$ such that $\frac{1}{\delta} \in \mathbb{N}$, $ n\delta \in \mathbb{N}$ and $\lambda (1-\delta) > \frac{1}{\ChiS e}$. Let $\TreeIter = \bar{t}_{o} + \log^{*}(n) +2$, where $\bar{t}_{o}$ is a constant depending only on $\lambda(1-\delta)$ and $\ChiS$ as described in Lemma~\ref{Lem.lower.b.main}. Assume that~\eqref{Cond.1.The.3.new} holds. Let $\tilde{C}$ be the estimated community produced by the modified version of Algorithm~\ref{Alg.1} with $\TreeDepth=\TreeIter(n)$. Then $\mathbb{P}(\tilde{C} = C^{*}) \to 1$ as $n \to \infty$.
\end{Theorem}

\begin{proof}
See Appendix~\ref{BP-App.5.3}.
\end{proof}

\paragraph{Comparison with Information Limits}
\label{compare.bounded}\hfill

Since $K \to \infty$ and the LLRs are bounded, the weak recovery result in Theorem~\ref{The.1.new} reduces to $\liminf_{n \to \infty} \frac{KD(P||Q)}{2\log(\frac{n}{K})} > 1$. This condition can be written as~\cite{recovering_one}:
\begin{align}
\lambda > C \frac{K}{n} \log(\frac{n}{K})
\end{align}
for some positive constant $C$. Thus, weak recovery only demands a vanishing $\lambda$. On the other hand, belief propagation achieves weak recovery for $\lambda > \frac{1}{ \ChiS e}$, where $\ChiS$ is greater than one and bounded as long as  LLR is bounded. This implies a gap between the information limits and belief propagation limits for weak recovery. Since $\ChiS \geq 1$, side information diminishes the gap.

For exact recovery, the following regime is considered:
\begin{align}
K = \frac{c n}{\log(n)}, \text{   } q = \frac{b \log^{2}(n)}{n},  \text{   } p = 2q
\end{align}
for fixed positive $b, c$ as $n \to \infty$. In this regime, $KD(P||Q) =O(\log(n)) $, and hence, weak recovery is always asymptotically possible. Also, $\lambda = c^{2}b$. Moreover, exact recovery is asymptotically possible if $cb(1 - \frac{1 + \log\log(2)}{\log(2)}) >1$. For belief propagation, we showed that exact recovery is possible if $cb(1 - \frac{1 + \log\log(2)}{\log(2)}) >1$ and $\lambda > \frac{1}{\ChiS e}$. 

Figure~\ref{fig.4} compares
the regions where weak recovery is achieved for belief propagation with
and without side information, as well as exact recovery with bounded-LLR
side information. Side information with $L=2$ is considered, where each node observes a
noisy label with cross-over probability $\alpha=0.3$. 
In Region~$1$, the belief propagation algorithm followed by voting achieves exact recovery with no need for side information. In Region~$2$,  belief propagation followed by voting achieves exact recovery with side information, but not without. In Region~$3$, weak recovery is achieved by belief propagation with no need for side information, but exact recovery is not asymptotically possible. In Region~$4$, weak recovery is achieved by the belief propagation as long as side information is available; exact recovery is not  asymptotically possible. In Region~$5$, exact recovery is asymptotically possible, but belief propagation without side information or with side information whose $\alpha =0.3$ cannot achieve even weak recovery (needs smaller $\alpha$, i.e., better side information).  In Region~$6$,  weak recovery,  but not exact recovery, is asymptotically possible via optimal algorithms, but belief propagation without side information or with side information whose $\alpha =0.3$ cannot achieve even weak recovery.
  
\begin{figure}
\centering
\includegraphics[width=\Figwidth]{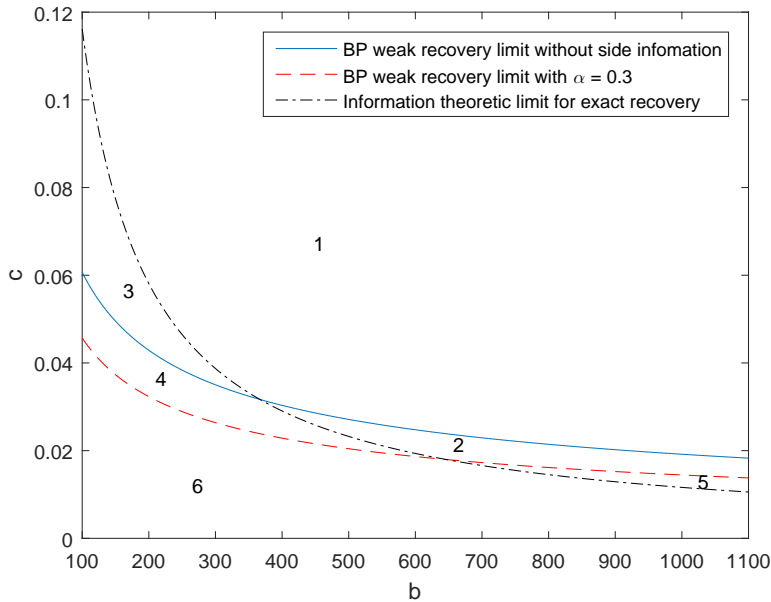}
\caption{Phase diagram with $K = c \frac{n}{\log(n)}$, $q = \frac{b\log^{2}(n)}{n}$, $p = 2q$ and $\alpha = 0.3$ for $b, c$ fixed as $n \to \infty$.}
\label{fig.4}
\end{figure}

Figure~\ref{fig.5} explores the effect of different values of $\alpha$, showing that as quality of side information improves (smaller $\alpha$), the gap between the belief propagation limit and the information limit decreases.

\begin{figure}
\centering
\includegraphics[width=\Figwidth]{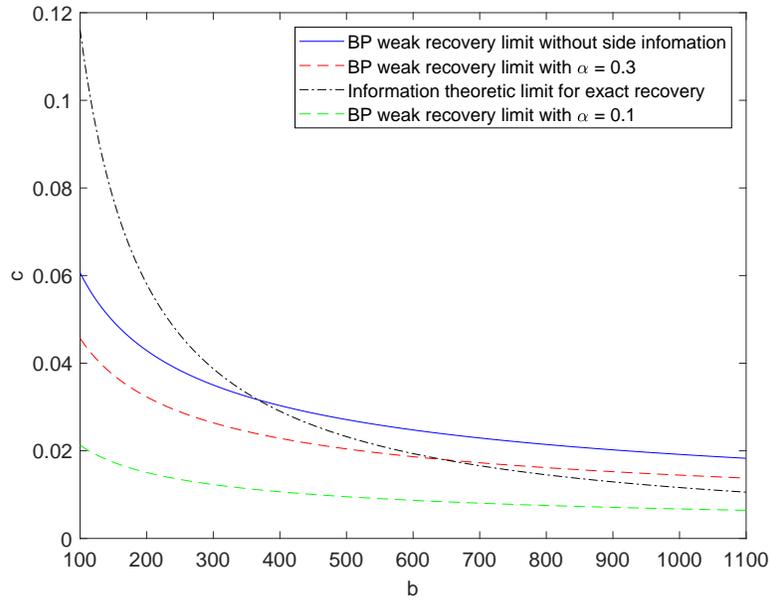}
\caption{Phase diagram with $K = c \frac{n}{\log(n)}$, $q = \frac{b\log^{2}(n)}{n}$, $p = 2q$ and $\alpha = 0.3, 0.1$ for $b, c$ fixed as $n \to \infty$.}
\label{fig.5}
\end{figure}

\paragraph{Application to Finite Data}
\label{numerical.bounded}\hfill

This section explores the relevance of asymptotic results, obtained in this paper, to finite data. The setup consists of a graph with $n=10^{4}, K = 100, \TreeDepth = 10$ and side information consisting of noisy labels with error probability $\alpha$. We study the performance of Algorithm~\ref{tab} on this data set. The following performance metric is used $\zeta = \frac{1}{2K}\sum_{i=1}^{n} |x_{i}^{*} - \hat{x}_{i}|$. The normalization by $2K$, and the fact that the algorithm is guaranteed to return a community of known size $K$, defines the range of the error metric $\zeta\in [0,1]$. Two scenarios are considered: First, $q = 5\times 10^{-4}$ and $p=10q$, which results in $\lambda \approx 0.041<\frac{1}{e}$. The results are reported for different values of $\alpha$ in Table~\ref{tb.1}, which show that when $\lambda < \frac{1}{\ChiS e}$, significant residual error exists. On the other hand, when $\lambda >> \frac{1}{\ChiS e}$, error occurrences are rare. In the second scenario,   $q = 5\times 10^{-4}$ and $p=80q$, resulting in $\lambda \approx 3.152 > \frac{1}{e}$. The results are reported for different values of $\alpha$ in Table~\ref{tb.2}. In this scenario, the performance of belief propagation without side information is much better compared with the first scenario because $\lambda > \frac{1}{e}$. The results also show that the performance is improved as $\alpha$ decreases.

\begin{table}
\caption{Performance of belief propagation for $\lambda < \frac{1}{e}$.}
\centering
\begin{tabular}{ |c|c|c|c|} 
 \hline
$\alpha$  & $\zeta$ w/o side & $\lambda \times \ChiS e \approx $ & $\zeta$ with side \\ \hline
0.1 & 0.95 & 0.903 & 0.75	 \\ \hline
0.01 & 0.95 & 10 & 0.4 \\ \hline
0.001 & 0.95 & 100 & 0.05 \\
\hline
\end{tabular}
\label{tb.1}
\\[0.2in]
\caption{Performance of belief propagation for $\lambda > \frac{1}{e}$.}
\centering
\begin{tabular}{ |c|c|c|c|} 
 \hline
$\alpha$  & $\zeta$ w/o side & $\lambda \times \ChiS e \approx $ & $\zeta$ with side \\ \hline
0.1 & 0.125 & 70 & 0.1 \\ \hline
0.01 & 0.125 & 840 & 0.03 \\ \hline
0.001 & 0.125 & 8551 & 0.02 \\
\hline
\end{tabular}
\label{tb.2}
\end{table}

\subsection{Unbounded LLR} 

The results of the previous section suggest that when $\ChiS \to \infty$ arbitrarily slowly, belief propagation achieves weak recovery for any fixed $\lambda >0$. In this section we prove this result for scalar side information with finite cardinality and $\ChiS$ that grows at a specific rate. 

The proof technique uses density evolution of $\Gamma_{i}^{t}$. More precisely, we assume that $\nu$, $\frac{\alpha_{+,\ell}}{\alpha_{-,\ell}}$, and $\lambda$ are constants independent of $n$, while $nq, Kq \overset{n\rightarrow\infty}{\xrightarrow{\hspace{0.2in}}} \infty$, which implies that $\frac{p}{q} \overset{n\rightarrow\infty}{\xrightarrow{\hspace{0.2in}}} 1$. This assumption allows us to precisely characterize the conditional probability density function of $\Gamma_{i}^{t}$ (asymptotically Gaussian), and hence, calculate the fraction of misclassified labels via the Q-function. Then, $\frac{n}{K}$ is allowed to grow and the behavior of the fraction of misclassified labels is studied as $\nu$ and the LLR of the side information grow.

Recall the definition of $\psi_{i}^{t}$ from~\eqref{psi.1} and $\Gamma_{i}^{t}$ from~\eqref{likeli} as well as the definitions of $Z_{0}^{t}$, $Z_{1}^{t}$, $U_{0}$ and $U_{1}$ defined directly afterward.

\begin{lemma}
\label{mean_variance_BP}
Assume $\lambda$, $\frac{\alpha_{+,\ell}}{\alpha_{-,\ell}}$ and $\nu$ are constants independent of $n$ while $nq, Kq \overset{n\rightarrow\infty}{\xrightarrow{\hspace{0.2in}}} \infty$. Then, for all $t \geq 0$:
\begin{align}
\mathbb{E}[Z_{0}^{t+1}] &= \frac{-\lambda}{2}b_{t} + o(1) \label{mean_1} \\
\mathbb{E}[Z_{1}^{t+1}] &= \frac{\lambda}{2}b_{t} + o(1) \label{mean_2}\\
\text{var}(Z_{0}^{t+1}) &= \text{var}(Z_{1}^{t+1}) = \lambda b_{t} + o(1) \label{var_1}
\end{align}
\end{lemma}

\begin{proof}
See Appendix~\ref{BP-App.6}.
\end{proof}

The following lemma shows that the distributions of $Z_{1}^{t}$ and $Z_{0}^{t}$ are asymptotically Gaussian. 

\begin{lemma}
\label{Gaussian}
Assume $\lambda$, $\frac{\alpha_{+,\ell}}{\alpha_{-,\ell}}$ and $\nu$ are constants independent of $n$ while $nq, Kq \overset{n\rightarrow\infty}{\xrightarrow{\hspace{0.2in}}} \infty$. Let $\phi(x)$ be the cumulative distribution function (CDF) of a standard normal distribution. Define $v_{0} =0$ and $v_{t+1} =\lambda \mathbb{E}_{Z,U_{1}}[\frac{1}{e^{-\nu} + e^{-(\frac{v_{t}}{2} +\sqrt{v_{t}}Z)-U_{1}}}]$, where $Z \sim \mathcal{N}(0,1)$. Then, for all $t\geq 0$:
\begin{align}
&\sup_{x} \big| \mathbb{P}\big( \frac{Z_{0}^{t+1} + \frac{v_{t+1}}{2} }{\sqrt{v_{t+1}}} \leq x \big) - \phi(x) \big|  \to 0 \label{gauss.1} \\
&\sup_{x} \big| \mathbb{P}\big( \frac{Z_{1}^{t+1} - \frac{v_{t+1}}{2} }{\sqrt{v_{t+1}}} \leq x \big) - \phi(x) \big|  \to 0 \label{gauss.2}
\end{align}
\end{lemma}

\begin{proof}
See Appendix~\ref{BP-App.7}.
\end{proof}

\begin{lemma}
\label{MAP_unbounded}
Assume $\lambda$, $\frac{\alpha_{+,\ell}}{\alpha_{-,\ell}}$ and $\nu$ are constants independent of $n$ while $nq, Kq \overset{n\rightarrow\infty}{\xrightarrow{\hspace{0.2in}}} \infty$.  Let $\hat{C}$ define the community recovered by the MAP estimator, i.e. $\hat{C} = \{i: \Gamma_{i}^{t} \geq \nu \}$. Then,
\begin{align}
\lim_{nq,Kq \to \infty} \lim_{n\to\infty} \frac{\mathbb{E}[\hat{C} \triangle C^{*}]}{K} &= \frac{n-K}{K} \mathbb{E}_{U_{0}}[Q(\frac{\nu+\frac{v_{t}}{2} - U_{0}}{\sqrt{v_{t}}})] \twocolbreak + \mathbb{E}_{U_{1}}[Q(\frac{-\nu+\frac{v_{t}}{2} + U_{1}}{\sqrt{v_{t}}})]
\end{align}
where $v_{0} =0$ and $v_{t+1} =\lambda \mathbb{E}_{Z,U_{1}}[\frac{1}{e^{-\nu} + e^{-(\frac{v_{t}}{2} +\sqrt{v_{t}}Z)-U_{1}}}]$, and $Z \sim \mathcal{N}(0,1)$.
\end{lemma}

\begin{proof}
Let $p_{e,0}, p_{e,1}$ denote Type I and Type II errors for recovering $\tau_0$. Then the proof follows from Lemmas~\ref{mean_variance_BP} and~\ref{Gaussian}, and because 
\[
\frac{\mathbb{E}[\hat{C} \triangle C^{*}]}{K} = \frac{n}{K} p_{e}^{t}  = \frac{n-K}{K}  p_{e,0} + p_{e,1}.
\]
\end{proof}

Lemma~\ref{MAP_unbounded} applies for side information with cardinality $L \geq 1$, and hence, generalizes~\cite{Kadavankandy:SingleCommunity} which was limited to $L=2$. Now $\frac{n}{K}$ is allowed to grow and the behavior of the fraction of misclassified labels is studied as $\nu$ and the LLR of the side information grows without bound. The following lemma shows that if $\ChiS \to \infty$ such that  $| h_{\ell} | = | \log(\frac{\alpha_{+,\ell}}{\alpha_{-,\ell}}) | < \nu$, belief propagation achieves weak recovery for any fixed $\lambda >0$ upon observing the tree structure of depth $t^{*}+2$ and side information with finite $L$, where $t^{*} = \log^{*}(\nu)$ is the number of times the logarithm function must be iteratively applied to $\nu$ to get a result less than or equal to one.

\begin{lemma}
\label{The.weak.unbound}
Let $\hat{C}$ be the output of the MAP estimator for the root of a random tree of depth $t^{*}+2$ upon observing the tree structure and side information with cardinality $L < \infty$. Assume as $\frac{n}{K} \to \infty$, $\ChiS \to \infty$ such that $| h_{\ell}| < \nu$. Then for any fixed $\lambda >0$:
\begin{align}
& \lim_{\frac{n}{K} \to \infty} \lim_{nq,Kq \to \infty} \lim_{n\to\infty} \frac{\mathbb{E}[\hat{C} \triangle C^{*}]}{K} = 0
\end{align}
\end{lemma}

\begin{proof}
See Appendix~\ref{BP-App.8}.
\end{proof}

Although Lemma~\ref{The.weak.unbound} is for $L$-ary side information, it focuses on one asymptotic regime of side information where $| h_{\ell} | < \nu$. To study other asymptotic regimes of side information, one example is considered for $L=2$, i.e., side information takes values in $\{0,1\}$. For constants $\eta$, $\beta$ $\in (0,1)$ and $\gamma > 0$, define: 
\begin{align}
\alpha_{+,1} &= \mathbb{P}(y = 1 | x^{*} = 1) = \eta \beta \nonumber \\
\alpha_{-,1} &= \mathbb{P}(y = 1 | x^{*} = 0) = \frac{\eta(1-\beta)}{(\frac{n-K}{K})^{\gamma}} \label{case.1.side}
\end{align}
Thus, $ \ChiS \to \infty$ and $h_{1} = (1+o(1))\gamma \log(\frac{n-K}{K})$ and $h_{2} = (1+o(1))\log(1- \eta \beta)$. For $0 < \gamma < 1$, Lemma~\ref{The.weak.unbound} shows that belief propagation achieves weak recovery for any fixed $\lambda >0$. This implies that belief propagation achieves weak recovery also for $\gamma \geq 1$ because  $\gamma \geq 1$ implies higher-quality side information. This generalizes the results obtained in~\cite{Kadavankandy:SingleCommunity} which was only for $\gamma =1$.

\subsubsection{Belief Propagation Algorithm for Community Recovery with Unbounded Side Information}
\label{BP.hidden.unbounded}\hfill

Lemma~\ref{MAP_unbounded} characterizes the performance of the optimal estimator of the root of a random tree upon observing the tree of depth $t$ and the side information. Similar to Section~\ref{BP.hidden.bounded}, the inference problem defined on the random tree is coupled to the problem of recovering a hidden community with side information. This is done via Lemma~\ref{couple}, which together with Equation~\eqref{fix.est} allow us to use Algorithm~\ref{tab} (as long as $(np)^{t} = n^{o(1)}$). Let $\tilde{C}$ be the output of Algorithm~\ref{tab}, i.e.,  the set of nodes with the $K$ largest $R_{i}^{t}$. Then, using Equation~\eqref{fix.est} we have: $\frac{\mathbb{E}[\tilde{C} \triangle C^{*}]}{K} \leq 2\frac{\mathbb{E}[\hat{C} \triangle C^{*}]}{K}$. Thus, the results of Lemma~\ref{The.weak.unbound} and the special case~\eqref{case.1.side} hold. This also suggests that belief propagation (Algorithm~\ref{tab}) achieves weak recovery for any $\lambda >0$ when $\ChiS$ grows with $\frac{n}{K}$ arbitrarily slowly.


\section{Conclusion}

This paper studies the effect of the quality and quantity of side information on the recovery of a hidden community of size $K=o(n)$. Under maximum likelihood detection, tight necessary and sufficient conditions are calculated for exact recovery, where we demonstrate how side information must evolve with $n$ in terms of either quantity or quality to improve the exact recovery threshold. A similar set of results are obtained for weak recovery. Under belief propagation, tight necessary and sufficient conditions for weak recovery are calculated when the LLRs are constant, and sufficient conditions when the LLRs vary with $n$. 
It is established that belief propagation followed by a local voting procedure achieves exact recovery, and its performance gap with respect to ML is reduced by side information.
Simulations on finite synthetic data-sets show that the asymptotic results of this paper are relevant in assessing the performance of belief propagation at finite $n$.


\appendices

\section{Auxiliary Lemmas For Information Limits}
\label{App.1}

\begin{lemma}
\label{Lem.1}
Define
\begin{align}
\hat{E}_{QU}(\theta,m_{1},m_{2}) &\triangleq \sup_{t \in \mathbb{R}} t\theta - m_{1}\log_{Q}(\mathbb{E}[e^{tL_{G}}]) \twocolbreak \includeonetwocol{}{\hspace{0.2in}} - m_{2}\log_{U}(\mathbb{E}[e^{tL_{S}}]) \nonumber \\
\hat{E}_{PV}(\theta,m_{1},m_{2}) &\triangleq \sup_{t \in \mathbb{R}} t\theta - m_{1}\log_{P}(\mathbb{E}[e^{tL_{G}}]) \twocolbreak \includeonetwocol{}{\hspace{0.2in}} - m_{2}\log_{V}(\mathbb{E}[e^{tL_{S}}])\nonumber
\end{align}
For $\theta \in [-m_{1}D(Q||P)-m_{2}D(U||V), m_{1}D(P||Q)+m_{2}D(V||U)]$, the following holds:
\begin{align}
\hat{E}_{QU}(\theta,m_{1},m_{2}) &= E_{QU}(\theta,m_{1},m_{2}) \label{eq.Lem1.1.1} \\ 
\hat{E}_{PV}(\theta,m_{1},m_{2}) &= E_{PV}(\theta,m_{1},m_{2}) \label{eq.Lem1.1.2}
\end{align}
Moreover, for $\delta: -m_{1}D(Q||P)-m_{2}D(U||V) \leq \theta \leq \theta+\delta \leq m_{1}D(P||Q)+m_{2}D(V||U)]$, the following holds:
\begin{align}
E_{QU}(\theta,m_{1},m_{2}) &\leq E_{QU}(\theta+\delta,m_{1},m_{2}) \twocolbreak \leq E_{QU}(\theta,m_{1},m_{2}) + \delta \label{eq.Lem1.2.1} \\
E_{PV}(\theta,m_{1},m_{2}) &\geq E_{PV}(\theta+\delta,m_{1},m_{2}) \twocolbreak \geq E_{PV}(\theta,m_{1},m_{2}) - \delta \label{eq.Lem1.2.2}
\end{align}
\end{lemma}

\begin{proof}
Equations~\eqref{eq.Lem1.1.1} and~\eqref{eq.Lem1.1.2} follow since $E_{PV}(\theta,m_{1},m_{2}) = E_{QU}(\theta,m_{1},m_{2}) - \theta$ and because:
\begin{align}
&E_{QU}(-m_{1}D(Q||P)-m_{2}D(U||V),m_{1},m_{2}) = 0 \nonumber \\
& E_{PV}(m_{1}D(P||Q)+m_{2}D(V||U),m_{1},m_{2}) =0 \nonumber \\
&\psi_{QU}'(m_{1},m_{2},0) = \psi_{PV}'(m_{1},m_{2},-1) = \twocolbreak \includeonetwocol{}{\hspace{1in}} -m_{1}D(Q||P) -m_{2}D(U||V) \nonumber \\ 
&\psi_{QU}'(m_{1},m_{2},1) = \psi_{PV}'(m_{1},m_{2},0) =  \twocolbreak \includeonetwocol{}{\hspace{1in}} 
m_{1}D(P||Q)  + m_{2}D(V||U)
\end{align}

Equations~\eqref{eq.Lem1.2.1} and~\eqref{eq.Lem1.2.2} follow since $E_{PV}(E_{QU})$ is decreasing (increasing) for $\theta \in [-m_{1}D(Q||P)-m_{2}D(U||V), m_{1}D(P||Q)+m_{2}D(V||U)]$.
\end{proof}


\begin{lemma}
\label{Lem.2} 
Assume $|L_{G}| \leq B$ and $|L_{S}| \leq B'$ for some positive constants $B$ and $B'$. Define $B'' = \max\{B,B'\}$. Then, for $t\in [-1,1]$ and $\eta \in [0,1]$,
\begin{align}
& \psi_{QU}''(m_{1},m_{2},t) \leq 2e^{5B''}\Big(  \min \big\{ m_{1}D(Q||P)+  m_{2}D(U||V) , \twocolbreak
\includeonetwocol{}{\hspace{1in}} m_{1}D(P||Q)+  m_{2}D(V||U) \big\} \Big)\label{eq.lem2.1.1} \\ 
& \psi_{QU}(m_{1},m_{2},t) \leq (m_{1}D(Q||P) + m_{2}D(U||V)) \twocolbreak \includeonetwocol{}{\hspace{1in}\times} (-t + e^{5B''}t^{2})  \label{eq.lem2.1.2} \\
&E_{QU}\Big(m_{1},m_{2},-(1-\eta)(m_{1}D(Q||P) + m_{2}D(U||V))\Big)  \twocolbreak \includeonetwocol{}{\hspace{0.2in}} \geq 
 \frac{\eta^{2}}{4e^{5B''}} (m_{1}D(Q||P) + m_{2}D(U||V)) \label{eq.lem2.1.3} \\ 
&\psi_{PV}''(m_{1},m_{2},t) \leq 2e^{5B''}\Big(  \min \big\{ m_{1}D(Q||P)+  m_{2}D(U||V) , \twocolbreak
\includeonetwocol{}{\hspace{1in}} m_{1}D(P||Q)+  m_{2}D(V||U) \big\} \Big) \label{eq.lem2.1.4} \\ 
&\psi_{PV}(m_{1},m_{2},t)\leq (m_{1}D(P||Q) + m_{2}D(V||U)) \twocolbreak \includeonetwocol{}{\hspace{1in}\times} (t + e^{5B''}t^{2}) \label{eq.lem2.1.5} \\ 
&E_{PV}(m_{1},m_{2},(1-\eta)(m_{1}D(P||Q) + m_{2}D(V||U))) \twocolbreak \includeonetwocol{}{\hspace{0.2in}} 
\geq \frac{\eta^{2}}{4e^{5B''}} (m_{1}D(P||Q) + m_{2}D(V||U)) \label{eq.lem2.1.6}
\end{align}
where $\psi_{QU}''(m_{1},m_{2},t)$ and $\psi_{PV}''(m_{1},m_{2},t)$ denote the second derivatives with respect to $t$.
\end{lemma}

\begin{proof}
By direct computation of the second derivative,
\begin{align}
\label{eq.lem2.2}
\psi_{QU}''(m_{1},m_{2},t) & \leq m_{1} \frac{\mathbb{E}_Q[L_{G}^{2}e^{tL_{G}}]}{\mathbb{E}_Q[e^{tL_{G}}]} + m_{2} \frac{\mathbb{E}_U[L_{S}^{2}e^{tL_{S}}]}{\mathbb{E}_U[e^{tL_{S}}]} \twocolbreak
\overset{(a)}{\leq} m_{1} e^{2B} \mathbb{E}_{Q}[L_{G}^{2}] + m_{2} e^{2B'} \mathbb{E}_{U}[L_{S}^{2}]
\end{align}
where $(a)$ follows by the assumption that $|L_{G}| \leq B$, $|L_{S}| \leq B'$ and holds for all $t \in [-1,1]$. 

Now consider the following function: $\phi(x) = e^{x} - 1 - x$ restricted to $|x| \leq B$. It is easy to see that $\phi(x)$ is non-negative, convex with $\phi(0) = \phi'(0) = 0$ and $\phi''(x) = e^{x}$. Hence, $e^{-B} \leq \phi''(x) \leq e^{B}$. From Taylor's theorem with integral remainder~\cite{apostol1969calculus}, we get: $\frac{e^{-B}x^{2}}{2} \leq \phi(x) \leq \frac{e^{B}x^{2}}{2}$, which implies $x^{2} \leq 2e^{B}\phi(x)$. Using this result for $x = L_{G}$ and $x=L_{S}$:
\begin{align}
\mathbb{E}_{Q}[L_{G}^{2}] &\leq 2e^{B} \mathbb{E}_{Q}[\phi(L_{G})] = 2e^{B}D(Q||P) \label{eq.lem2.3.1} \\ 
\mathbb{E}_{U}[L_{S}^{2}] &\leq 2e^{B'} \mathbb{E}_{U}[\phi(L_{S})] = 2e^{B'}D(U||V) \label{eq.lem2.3.2}
\end{align}

Combining~\eqref{eq.lem2.2},~\eqref{eq.lem2.3.1},~\eqref{eq.lem2.3.2} lead to $\psi_{QU}''(m_{1},m_{2},t) \leq 2m_{1}e^{3B}D(Q||P) + 2m_{2}e^{3B'}D(U||V)$ for $t\in[-1,1]$. Similarly, it can shown for $t \in[0,2]$: $\psi_{QU}''(m_{1},m_{2},t) \leq 2m_{1}e^{5B}D(Q||P) + 2m_{2}e^{5B'}D(U||V)$.

On the other hand, using $\phi(x) = e^{-x} - 1 + x$ with $|x| \leq B$, it can be shown that $\psi_{PV}''(m_{1},m_{2},t) \leq 2m_{1}e^{5B}D(P||Q) + 2m_{2}e^{5B'}D(V||U)$, for $t \in[0,2]$. By definition, $\psi_{QU}(m_{1},m_{2},t) = \psi_{PV}(m_{1},m_{2},t-1)$, and hence,  $\psi_{QU}''(m_{1},m_{2},t) \leq 2m_{1}e^{5B}D(P||Q) + 2m_{2}e^{5B'}D(V||U)$, for $t \in[-1,1]$, which concludes the proof of~\eqref{eq.lem2.1.1}. The proof of~\eqref{eq.lem2.1.4} follows similarly.

Now since $\psi_{QU}(m_{1},m_{2},0) =0$ and $\psi_{QU}'(m_{1},m_{2},0) = -m_{1}D(Q||P) - m_{2}D(U||V)$, then using Taylor's theorem with integral remainder, we have for $t\in[-1,1]$:
\begin{align}
& \psi_{QU}(m_{1},m_{2},t) \nonumber \\
& = \psi_{QU}(m_{1},m_{2},0) + t\psi_{QU}'(m_{1},m_{2},0)  \twocolbreak \includeonetwocol{}{\hspace{0.2in}}
+ \int_{t}^{0} (\lambda - t) \psi_{QU}''(m_{1},m_{2},t) d\lambda \nonumber\\
&\overset{(a)}{\leq} -t(m_{1}D(Q||P) + m_{2}D(U||V)) \twocolbreak  \includeonetwocol{}{\hspace{0.2in}}
 + e^{5B''}(m_{1} D(Q||P) + m_{2}D(U||V))t^{2} \label{eq.lem2.4}
\end{align}
where $(a)$ follows using~\eqref{eq.lem2.1.1}. Similarly, it can be shown that:
\begin{align}
\label{eq.lem2.5}
& \psi_{PV}(m_{1},m_{2},t) \leq t(m_{1}D(P||Q) + m_{2}D(V||U)) \twocolbreak
  + e^{5B''}(m_{1} D(P||Q) + m_{2}D(V||U))t^{2}
\end{align}

Combining~\eqref{eq.lem2.4} and~\eqref{eq.lem2.5} concludes the proof of~\eqref{eq.lem2.1.2},~\eqref{eq.lem2.1.5}. Using~\eqref{eq.lem2.1.2} and~\eqref{eq.lem2.1.5}, we get:
\begin{align}
\label{eq.lem2.6} \nonumber
& E_{QU}\Big(m_{1},m_{2},-(1-\eta)(m_{1}D(Q||P) + m_{2}D(U||V))\Big)  \\ \nonumber
& \geq \sup_{t\in [0,1]} t(-(1-\eta)(m_{1}D(Q||P) + m_{2}D(U||V)))  \twocolbreak \includeonetwocol{}{\hspace{0.2in}} +
t(m_{1}D(Q||P) + m_{2}D(U||V)) \\ \nonumber
 & \includeonetwocol{}{\hspace{0.2in}} - e^{5B''}(m_{1} D(Q||P) + m_{2}D(U||V))t^{2} \\ 
&= \frac{\eta^{2}}{4e^{5B''}} (m_{1}D(Q||P) + m_{2}D(U||V))
\end{align}

Similarly,
\begin{align}
& E_{PV}\Big(m_{1},m_{2},(1-\eta)(m_{1}D(P||Q) + m_{2}D(V||U))\Big) \nonumber \\  
& \geq \frac{\eta^{2}}{4e^{5B''}} (m_{1} D(P||Q) + m_{2}D(V||U)) \label{eq.lem2.7}
\end{align}
Combining~\eqref{eq.lem2.6} and~\eqref{eq.lem2.7} concludes the proof of~\eqref{eq.lem2.1.3},~\eqref{eq.lem2.1.6}.
\end{proof}

\begin{lemma}
\label{rem.f1.f2}
$\eta_{3}(\rho, a,b,\beta) \geq \eta_{2}(\rho, a,b,\beta)$, for $0<\beta < \rho(a-b-bT)$.
\end{lemma}
\begin{proof}
It is easy to show that $\eta_{3}(\rho, a,b,\beta) - \beta$ is convex in $\beta>0$. Thus, the optimal $\beta$ can be calculated as $\beta^{*} = \rho(aT - a +b)$ at which $\eta_{3}(\rho, a,b,\beta^{*}) - \beta^{*} = 0$. Thus, $\eta_{3}(\rho, a,b,\beta) \geq \beta$ for all $a \geq b>0$.

Furthermore, note that $\eta_{2}(\rho, a,b,\beta)$ is convex and increasing in $0<\beta < \rho(a-b-bT)$. By direct substitution, it can be shown that at $\beta = \rho(a-b-bT)$: $\eta_{2}(\rho, a,b,\beta) = \beta$. This implies that at $\beta = \rho(a-b-bT)$:
\begin{align}
\eta_{3}(\rho, a,b,\beta) - \eta_{2}(\rho, a,b,\beta) & = \eta_{3}(\rho, a,b,\beta) - \beta \geq 0 \label{beta_inc}
\end{align}
Using~\eqref{beta_inc} together with the fact that $\eta_{3}(\rho, a,b,\beta) - \eta_{2}(\rho, a,b,\beta)$ is convex in $\beta > 0$, leads to the conclusion that $\eta_{3}(\rho, a,b,\beta) \geq \eta_{2}(\rho, a,b,\beta)$ for $0<\beta < \rho(a-b-bT)$.
\end{proof}

\begin{lemma}
\label{Le.10}
Let $X_{1},\cdots,X_{n}$ be a sequence of i.i.d random variables. Define $\Gamma(t) = \log(\mathbb{E}[e^{tX}])$. Define $S = \sum_{i=1}^{n} X_{i}$, then for any $\epsilon >0$ and $a \in \mathbb{R}$:
\begin{align}
&\mathbb{P}\big(S \geq a-\epsilon\big) \geq e^{-\big(t^{*}a-n\Gamma(t^{*})+|t^{*}|\epsilon\big)} \Big( 1 - \frac{n\sigma^{2}_{\hat{X}}}{\epsilon^{2}} \Big) \label{up.1} \\
&\mathbb{P}\big(S \leq a+\epsilon\big) \geq e^{-\big(t^{*}a-n\Gamma(t^{*})+|t^{*}|\epsilon\big)} \Big( 1 - \frac{n\sigma^{2}_{\hat{X}}}{\epsilon^{2}} \Big)\label{up.2}
\end{align}
where $t^{*} = \arg\sup_{t\in \mathbb{R}} ta - \Gamma(t)$, $\hat{X}$ is a random variable with the same alphabet as $X$ but distributed according to $\frac{e^{t^{*}x}\mathbb{P}(x)}{\mathbb{E}_{X}[e^{t^{*}x}]}$ and $\mu_{\hat{X}}, \sigma^{2}_{\hat{X}}$ are the mean and variance of $\hat{X}$, respectively.
\end{lemma}

\begin{proof}
\begin{align}
&\mathbb{P}\big(S \geq a-\epsilon\big) \geq \mathbb{P}\big(a-\epsilon \leq S \leq a+\epsilon\big) \nonumber \\
= &  \int_{a-\epsilon \leq S \leq a+\epsilon}  \mathbb{P}(x_{1}) \cdots \mathbb{P}(x_{n}) dx_{1}\cdots dx_{n} \nonumber \\ 
\overset{(a)}{\geq} &  e^{-(ta-n\Gamma(t))-|t|\epsilon}  \int_{a-\epsilon \leq S \leq a+\epsilon} \prod_{i=1}^{n} \bigg( \frac{e^{tx_{i}}\mathbb{P}(x_{i})}{\mathbb{E}_{X}[e^{tx}]} dx_{i}\bigg) \nonumber \\ 
\overset{(b)}{=} &  e^{-(ta-n\Gamma(t))-|t|\epsilon} \mathbb{P}_{\hat{X}_{n}}\Big( a-\epsilon \leq S \leq a+\epsilon \Big) \nonumber \\
\overset{(c)}{\geq} &  e^{-(ta-n\Gamma(t))-|t|\epsilon} \Big( 1 - \frac{ n\sigma^{2}_{\hat{X}} + (n\mu_{\hat{X}} - a)^{2} }{\epsilon^{2}} \Big) \label{lem.ap}
\end{align}
where, for all finite $\mathbb{E}[e^{tX}]$, $(a)$ is true because $e^{t\sum x_i} \le e^{n(ta+|t|\epsilon)}$ over the range of integration, $(b)$ holds because $\frac{e^{t x}\mathbb{P}_X(x)}{\mathbb{E}_{X}[e^{tX}]}$ is a valid distribution~\cite{large_dev}, and $(c)$ holds by Chebyshev inequality and by defining $\mu_{\hat{X}}, \sigma^{2}_{\hat{X}}$ to be the mean and variance of $\hat{X}$, respectively.
%
%
Since $ta - n\Gamma(t)$ is concave in $t$, to find $t^*=\arg \sup_t (ta - n\Gamma(t))$ we set the derivative to zero, finding $a = n\frac{\mathbb{E}_{X}[xe^{t^{*}x}]}{\mathbb{E}[e^{t^{*}x}]}$. Also, by direct computation of $\mu_{\hat{X}}$, it can be shown that $\mu_{\hat{X}} = \frac{\mathbb{E}_{X}[xe^{tx}]}{\mathbb{E}[e^{tx}]}$. This means that at $t = t^{*}$, $n\mu_{\hat{X}} = a$. Thus, substituting back in~\eqref{lem.ap} leads to:
\begin{align}
&\mathbb{P}\big(S \geq a-\epsilon\big) \geq e^{-(t^{*}a-n\Gamma(t^{*}))-|t^{*}|\epsilon} \big( 1 - \frac{n\sigma^{2}_{\hat{X}}}{\epsilon^{2}} \big) \nonumber
\end{align}
This concludes the proof of~\eqref{up.1}. The proof of~\eqref{up.2} follows similarly.

In our model $\epsilon = \log^{\frac{2}{3}}(n)$ and $n\sigma^{2}_{\hat{X}}$ is $O(\log(n))$, and hence,
\begin{align}
&\mathbb{P}\big(S \geq a-\epsilon\big) \geq e^{-(t^{*}a-n\Gamma(t^{*}))-|t^{*}|\epsilon} \big( 1 - o(1) \big) \nonumber
\end{align}
which concludes the proof.
\end{proof}



\section{Necessity of Theorem~\ref{The.1.new}}
\label{App.3}

Let $\boldsymbol{x}^{*}_{\backslash i,j}$ represent the vector $\boldsymbol{x}^*$ with two coordinates $i,j$ removed. We wish to determine $x_i^*$ via an observation of $\boldsymbol{G},\boldsymbol{Y}$, as well as a node index $J$ and the expurgated vector of labels $\boldsymbol{x}^{*}_{\backslash i,J}$, where node $J$ is randomly and uniformly chosen from inside (outside) the community if node $i$ is outside (inside) the community, i.e., $\{j \; : \; x_j^* \neq x_i^*   \}$. Then:
\begin{align} 
& \frac{\mathbb{P}(\boldsymbol{G},\boldsymbol{Y},J,\boldsymbol{x}^{*}_{\backslash i,J} | x^{*}_{i} = 0)}{\mathbb{P}(\boldsymbol{G},\boldsymbol{Y},J,\boldsymbol{x}^{*}_{\backslash i,J} | x^{*}_{i} = 1)} \nonumber \\
 = & \frac{\mathbb{P}(\boldsymbol{G}|\boldsymbol{Y},J,\boldsymbol{x}^{*}_{\backslash i,J}, x^{*}_{i} = 0)}{\mathbb{P}(\boldsymbol{G}|\boldsymbol{Y},J,\boldsymbol{x}^{*}_{\backslash i,J}, x^{*}_{i} = 1)}  \twocolbreak  \includeonetwocol{}{\times}
\frac{\mathbb{P}(\boldsymbol{x}^{*}_{\backslash i,J}|J,x^{*}_{i}=0,\boldsymbol{Y})}{\mathbb{P}(\boldsymbol{x}^{*}_{\backslash i,J}|J,x^{*}_{i}=1,\boldsymbol{Y})} \frac{\mathbb{P}(\boldsymbol{Y},J|x^{*}_{i}=0)}{\mathbb{P}(\boldsymbol{Y},J|x^{*}_{i}=1)}\nonumber \\ 
 \overset{(a)}{=} & \frac{\mathbb{P}(\boldsymbol{G}|J,\boldsymbol{x}^{*}_{\backslash i,J}, x^{*}_{i} = 0)}{\mathbb{P}(\boldsymbol{G}|J,\boldsymbol{x}^{*}_{\backslash i,J}, x^{*}_{i} = 1)}  \frac{\mathbb{P}(y_{i,1}\cdots,y_{i,M}|x^{*}_{i}=0)}{\mathbb{P}(y_{i,1}\cdots,y_{i,M}|x^{*}_{i}=1)}  \twocolbreak \includeonetwocol{}{\times}
\frac{\mathbb{P}(y_{J,1},\cdots,y_{J,M}|J,x^{*}_{i}=0)}{\mathbb{P}(y_{J,1},\cdots,y_{J,M}|J,x^{*}_{i}=1)} \nonumber \\ 
= & \bigg(  \prod_{ \substack{ k\neq i,J \\ x^{*}_{k}=1 }}  \frac{Q(G_{ik})P(G_{Jk})}{P(G_{ik})Q(G_{Jk})} \bigg) \bigg( \prod_{m=1}^{M} \frac{U(y_{i,m})V(y_{J,m})}{V(y_{i,m})U(y_{J,m})} \bigg) \label{hyp.}
\end{align}
where $(a)$ holds because $\boldsymbol{G}$ and $\boldsymbol{Y}$ are independent given the labels, $\mathbb{P}(J|x^{*}_{i}=0) = \mathbb{P}(J|x^{*}_{i}=1)$ and $\mathbb{P}(\boldsymbol{x}^{*}_{\backslash i,J}|J,x^{*}_{i}=0,\boldsymbol{Y}) = \mathbb{P}(\boldsymbol{x}^{*}_{\backslash i,J}|J,x^{*}_{i}=1,\boldsymbol{Y})$.

Denote the set of nodes inside the community, excluding $i,J$, with ${\cal K} = \{ k \neq i,J \; : \; x_k^*=1\}$, and construct a vector from four sets of random variables as follows:
\[
T \triangleq \Big[\{y_{i,m}\}_{m=1}^{M}, \{y_{J,m}\}_{m=1}^{M}, \{G_{ik}\}_{k\in\cal K},\{G_{Jk}\}_{k\in\cal K}\Big].
\]
where the members of each set appear in the vector in increasing order of their varying index. From~\eqref{hyp.}, $T$ is a sufficient statistic of $(G,\boldsymbol{Y},J,\boldsymbol{x}^{*}_{\backslash i,J})$ for testing $x^{*}_{i} \in \{0,1\}$. Moreover, conditioned on $x^{*}_{i} = 0$, $T$ is distributed according to $U^{\otimes M}V^{\otimes M}Q^{\otimes(K-1)}P^{\otimes(K-1)}$ and conditioned on $x^{*}_{i} = 1$, $T$ is distributed according to $V^{\otimes M}U^{\otimes M }P^{\otimes(K-1)}Q^{\otimes(K-1)}$. Then, for any estimator $\hat{\boldsymbol{x}}(\boldsymbol{G},\boldsymbol{Y})$ achieving weak recovery:
\begin{align}
\twocolAlignMarker \mathbb{E}[d(\hat{\boldsymbol{x}},\boldsymbol{x}^{*})] \twocolnewline & = \sum_{i=1}^{n} \mathbb{P}(x_{i}^{*} \neq \hat{x}_{i})\nonumber \\
 & \geq \sum_{i=1}^{n} \min_{\tilde{x}_{i}(\boldsymbol{G},\boldsymbol{Y})} \mathbb{P}(x_{i}^{*} \neq \tilde{x}_{i})\nonumber \\
& \geq \sum_{i=1}^{n} \min_{\tilde{x}_{i}(\boldsymbol{G},\boldsymbol{Y},J,\boldsymbol{x}^{*}_{\backslash i,J})} \mathbb{P}(x_{i}^{*} \neq \tilde{x}_{i}) \nonumber \\
& = n  \min_{\tilde{x}_{i}(\boldsymbol{G},\boldsymbol{Y},J,\boldsymbol{x}^{*}_{\backslash i,J})}   \mathbb{P}(x_{i}^{*} \neq \tilde{x}_{i}) \nonumber \\
& = n \min_{\tilde{x}_{i}(\boldsymbol{G},\boldsymbol{Y},J,\boldsymbol{x}^{*}_{\backslash i,J})} \bigg(\frac{K}{n} \mathbb{P}(x_{i}^{*} \neq \tilde{x}_{i} | x_{i}^{*} = 1 ) \twocolbreak \hspaceonetwocol{0in}{1.5in}+ \frac{n-K}{n} \mathbb{P}(x_{i}^{*} \neq \tilde{x}_{i} | x_{i}^{*} = 0 )  \bigg) \nonumber \\
& \geq n \min_{\tilde{x}_{i}(\boldsymbol{G},\boldsymbol{Y},J,\boldsymbol{x}^{*}_{\backslash i,J})} \bigg(\frac{K}{n} \mathbb{P}(x_{i}^{*} \neq \tilde{x}_{i} | x_{i}^{*} = 1 ) \twocolbreak \hspaceonetwocol{0in}{1.5in} + \frac{K}{n} \mathbb{P}(x_{i}^{*} \neq \tilde{x}_{i} | x_{i}^{*} = 0 ) \bigg) \nonumber \\
& = K \min_{\tilde{x}_{i}(\boldsymbol{G},\boldsymbol{Y},J,\boldsymbol{x}^{*}_{\backslash i,J})} \big(\mathbb{P}(x_{i}^{*} \neq \tilde{x}_{i} | x_{i}^{*} = 1 ) \twocolbreak \hspaceonetwocol{0in}{1.5in} +  \mathbb{P}(x_{i}^{*} \neq \tilde{x}_{i} | x_{i}^{*} = 0 ) \big) \label{type.1.2}
\end{align}
Since by assumption, $\mathbb{E}[d(\hat{\boldsymbol{x}},\boldsymbol{x}^{*})] = o(K)$, then by~\eqref{type.1.2}, the sum of Type-I and II probabilities of error is $o(1)$, which implies that as $n \to\infty$~\cite{polyanski}:
\begin{align}
&TV\Big( U^{\otimes M} V^{\otimes M} Q^{\otimes(K-1)} P^{\otimes(K-1)} , \twocolbreak \hspaceonetwocol{0in}{1in}V^{\otimes M} U^{\otimes M} P^{\otimes(K-1)} Q^{\otimes(K-1)}\Big) \to 1  \label{eq.1}
\end{align}
where $TV(\cdot,\cdot)$ is the total variational distance between probability distributions. By properties of the total variational distance and KL divergence~\cite{polyanski}, for any two distributions $\tilde{P}, \tilde{Q}$: $D(\tilde{P}|| \tilde{Q}) \geq \log(\frac{1}{2(1-TV(\tilde{P}|| \tilde{Q}))})$. Hence, using~\eqref{eq.1}:
\begin{align}
& D\Big(U^{\otimes M} V^{\otimes M} Q^{\otimes(K-1)} P^{\otimes(K-1)} \Big|\Big| \twocolbreak \hspaceonetwocol{0in}{1in}
V^{\otimes M} U^{\otimes M} P^{\otimes(K-1)} Q^{\otimes(K-1)} \Big) \nonumber \\
& = M\Big(D(U||V) + D(V||U)\Big) \twocolbreak \hspaceonetwocol{0in}{1in}
+ (K-1) \Big(D(P||Q) + D(Q||P)\Big) \to \infty \label{Con.1.gen}
\end{align}
Since the LLRs are bounded by assumption, using Lemma~\ref{Lem.2} in Appendix~\ref{App.1}, 
\begin{align}
\twocolAlignMarker (K-1)D(P||Q)+MD(V||U) \twocolnewline & = E_{QU}\Big((K-1)D(P||Q)+MD(V||U), K-1, M\Big) \nonumber \\ 
& \geq E_{QU}\Big(-\frac{(K-1)D(Q||P)+MD(U||V)}{2}, K-1, M\Big) \nonumber \\
& \geq C\Big((K-1)D(Q||P)+MD(U||V)\Big) 
\end{align}
for some positive constant $C$. Substituting in~\eqref{Con.1.gen} leads to:
\begin{equation}
MD(V||U) + (K-1)D(P||Q) \to\infty
\end{equation}
which proves the first condition in~\eqref{Eq.Hajek}. 

$\boldsymbol{x}^{*}$ is drawn uniformly from the set $\{ \boldsymbol{x} \in \{0,1\}^{n}: w(\boldsymbol{x}) = K \}$ and $w(\boldsymbol{x}) = \sum_{j=1}^{n} x_{j}$; therefore $x_{i}$'s are individually Bernoulli-$\frac{K}{n}$. Then, for any estimator $\hat{\boldsymbol{x}}(\boldsymbol{G},\boldsymbol{Y})$ achieving weak recovery we have the following, where $H(\cdot)$ and $I(\cdot;\cdot)$ are the entropy and mutual information of their respective arguments. 
\begin{align}
I(\boldsymbol{G},\boldsymbol{Y}; \boldsymbol{x}^{*}) & \overset{(a)}{\geq} I(\hat{\boldsymbol{x}}(\boldsymbol{G},\boldsymbol{Y}); \boldsymbol{x}^{*}) \twocolbreak 
\overset{(b)}{\geq}  \min_{\mathbb{E}[d(\tilde{\boldsymbol{x}},\boldsymbol{x}^{*})]\leq \epsilon_{n}K} I(\tilde{\boldsymbol{x}}(\boldsymbol{G},\boldsymbol{Y}); \boldsymbol{x}^{*})  \\ \nonumber
& \geq H(\boldsymbol{x}^{*}) - \max_{\mathbb{E}[d(\tilde{\boldsymbol{x}},\boldsymbol{x}^{*})]\leq \epsilon_{n}K} H(d(\tilde{\boldsymbol{x}},\boldsymbol{x}^{*})) \\ 
& \overset{(c)}{=} \log\big({{n}\choose{K}}\big) - nh(\frac{\epsilon_{n}K}{n}) \twocolbreak 
\overset{(d)}{\geq} K\log(\frac{n}{k})(1+o(1))\label{eq.2}
\end{align}
where $(a)$ is due to the data processing inequality~\cite{polyanski}, in $(b)$ we defined $\epsilon_{n} = o(1)$, $(c)$  is due to the fact that $\max_{\mathbb{E}(w(X)) \leq pn} H(X) = nh(p)$ for any $ p \leq \frac{1}{2}$~\cite{infor_limits}, where $h(p) \triangleq -p\log(p) - (1-p)\log(1-p)$, and $(d)$ holds because ${{n}\choose{K}} \geq (\frac{n}{K})^{K}$, the assumption $K = o(n)$ and the bound $h(p) \leq -p\log(p) + p$ for $p \in [0,1]$. Denoting by $P(\boldsymbol{G},\boldsymbol{Y},\boldsymbol{x}^{*})$ the joint distribution of the graph, side information, and node labels, and using~\cite{polyanski}:
\begin{align}
& I(\boldsymbol{G},\boldsymbol{Y}; \boldsymbol{x}^{*}) \twocolbreak
= \min_{\tilde{Q}} D\Big(\mathrm{P}(\boldsymbol{G},\boldsymbol{Y}|\boldsymbol{x}^{*}) \;\big|\big|\; \tilde{\mathrm{Q}} \; \big| \;\mathrm{P}(\boldsymbol{x}^{*})\Big) \nonumber\\ 
& \leq D\Big(\mathrm{P}(\boldsymbol{G}|\boldsymbol{x}^{*}) \prod_{m=1}^{M}(\mathrm{P}(\boldsymbol{y}_{m}|\boldsymbol{x}^{*})) \Big|\Big| Q^{\otimes{{n}\choose{2}}} \prod_{m=1}^{M}(U^{\otimes n}) \; \big| \mathrm{P}(\boldsymbol{x}^{*})\Big) \nonumber\\ 
& = {{K}\choose{2}} D(P||Q) + K M D(V||U) \label{eq.3}
\end{align}
Combining~\eqref{eq.2} and~\eqref{eq.3}:
\begin{equation}
\label{Con.2.gen}
\liminf_{n\to\infty} (K-1)D(P||Q)+ 2M D(V||U) \geq 2\log(\frac{n}{K})
\end{equation}
which proves the second condition in~\eqref{Eq.Hajek}. 


\section{Sufficiency of Theorem~\ref{The.1.new}}
\label{App.4}

The sufficient conditions for weak recovery is derived for the maximum likelihood (ML) detector. Define:
\begin{align} 
\label{def.e.1}
e_{1}(S,T) & \triangleq \sum_{i\in S} \sum_{j\in T}  L_{G}(i,j) \\ 
e_{2}(S) & \triangleq \sum_{i \in S} \sum_{m=1}^{M} L_{S}(i,m)
\end{align}
for any subsets $S,T \subset \{1,\cdots,n\}$. Using these definitions, the maximum likelihood detection can be characterized as follows:
\begin{equation}
\label{ML_rule.1}
\hat{C} = \hat{C}_{ML}= \underset{C \subset \{1,\cdots,n\}\atop |C|=K}{\arg\max} \big(e_{1}(C,C) + e_{2}(C)\big)\ 
\end{equation}

Let $R \triangleq |\hat{C} \cap C^{*}|$, then $|\hat{C} \triangle C^{*}| = 2(K-R)$, and hence, to show that maximum likelihood achieves weak recovery, it is sufficient to show that there exists positive $\epsilon = o(1)$, such that $\mathbb{P}\big(R \leq (1-\epsilon)K\big) = o(1)$. 

To bound the error probability of ML, we characterize the separation of its likelihood from the likelihood of the community $C^{*}$.
\begin{align}
& e_{1}(\hat{C},\hat{C}) + e_{2}(\hat{C}) - \big(e_{1}(C^{*},C^{*}) + e_{2}(C^{*})\big) \nonumber \\
& = e_{1}(\hat{C}\backslash C^{*},\hat{C}\backslash C^{*}) + e_{1}(\hat{C}\backslash C^{*},\hat{C} \cap C^{*}) - e_{1}(C^{*}\backslash \hat{C},C^{*}) + \twocolbreak  \hspaceonetwocol{0in}{0.15in}
e_{2}(\hat{C}\backslash C^{*}) - e_{2}(C^{*}\backslash \hat{C}) 
\end{align}
By definition $|C^{*}\backslash \hat{C}| = |\hat{C}\backslash C^{*}| = K - R$. Thus, for any $0\leq r \leq K-1$,
\begin{align}
&\mathbb{P}(R=r) \nonumber \\
& \le {\mathbb P}\Big(\big\{\hat{C}: |\hat{C}| = K, |\hat{C}\cap C^{*}| = r, \twocolbreak \hspaceonetwocol{0in}{0.15in}
e_{1}(\hat{C},\hat{C}) + e_{2}(\hat{C}) - e_{1}(C^{*},C^{*}) - e_{2}(C^{*}) \geq 0   \big\}\Big) \nonumber \\ 
& =  {\mathbb P}\Big(\big \{ S \subset C^{*}, T \subset (C^{*})^{c}: |S|=|T|= K-r, \twocolbreak \hspaceonetwocol{0in}{0.15in}
e_{1}(S,C^{*}) + e_{2}(S) \leq e_{1}(T,T) + e_{1}(T,C^{*}\backslash S) + e_{2}(T)   \big\}\Big) \nonumber \\ 
& \le  {\mathbb P}\Big( \big\{S \subset C^{*}: |S| = K-r,  e_{1}(S,C^{*}) + e_{2}(S) \leq \theta  \big\}  \nonumber \\ 
&  \qquad \cup \big \{ S \subset C^{*}, T \subset (C^{*})^{c}: |S|=|T|= K-r, \twocolbreak \hspaceonetwocol{0in}{0.75in}
 e_{1}(T,T) + e_{1}(T,C^{*}\backslash S) + e_{2}(T)  \geq \theta \big\}\Big) \label{subset}
\end{align}
where  $\theta = (1-\eta)(aD(P||Q)+(K-r)MD(V||U))$, for some $\eta \in (0,1)$ and $a = {{K}\choose{2}} - {{r}\choose{2}}$. We further assume random variables $L_{G,i}$ are drawn i.i.d. according to the distribution of $L_{G}$, and $L_{S,m,j}$ are similarly i.i.d. copies of $L_{S}$.
Then, using~\eqref{subset} and a union bound:
\begin{align} 
\twocolAlignMarker\mathbb{P}(R=r) \twocolnewline &\overset{}{\leq}  {{K}\choose{K-r}} \mathbb{P}\Big(\sum_{i=1}^{a} L_{G,i} + \sum_{j=1}^{K-r}\sum_{m=1}^{M} L_{S,m,j} \leq \theta \Big) \nonumber\\
& \hspace{0.2in}+  {{K}\choose{K-r}}{{n-K}\choose{K-r}} \mathbb{P}\Big(\sum_{i=1}^{a} L_{G,i} + \sum_{j=1}^{K-r} \sum_{m=1}^{M} L_{S,m,j} \geq \theta \Big) \nonumber \\ 
 &\overset{(a)}{\leq}   e^{(K-r) \log(\frac{Ke}{K-r})} \twocolbreak\includeonetwocol{}{\hspace{0.2in}\times}
e^{- \sup_{t \geq 0} -t \theta - a\log_{P}(\mathbb{E}[e^{-tL_{G}}]) - (K-r)M\log_{V}(\mathbb{E}[e^{-tL_{S}}])}  \nonumber \\ 
&   \hspace{0.2in} +e^{(K-r) \log(\frac{(n-K)Ke^{2}}{(K-r)^{2}})} \twocolbreak\includeonetwocol{}{\hspace{0.2in}\times}
 e^{-\sup_{t\geq 0} t\theta - a\log_{Q}(\mathbb{E}[e^{tL_{G}}]) - (K-r)M\log_{U}(\mathbb{E}[e^{tL_{S}}])}  \nonumber\\
 &\overset{(b)}{\leq}  e^{(K-r) \log(\frac{Ke}{K-r}) - E_{PV}(\theta,a,M(K-r))}  \twocolbreak\hspaceonetwocol{0in}{0.2in}
 + e^{(K-r) \log(\frac{(n-K)Ke^{2}}{(K-r)^{2}}) - E_{QU}(\theta,a,M(K-r))}  \nonumber\\ 
  &\overset{(c)}{=}  e^{(K-r) \log(\frac{Ke}{K-r}) - E_{PV}(\theta,a,M(K-r))}  \twocolbreak\hspaceonetwocol{0in}{0.2in}
+ e^{(K-r) \log(\frac{(n-K)Ke^{2}}{(K-r)^{2}}) - E_{PV}(\theta,a,M(K-r)) - \theta} \nonumber \\ 
& \overset{(d)}{\leq}  e^{(K-r) \log(\frac{Ke}{K-r}) - E_{PV}(\theta,a,M(K-r))} \nonumber \\
& \hspace{0.2in} + e^{-(K-r) \big( (1-\eta)((\frac{K-1}{2})D(P||Q)+MD(V||U)) - \log(\frac{n-K}{K}) \big)}  \twocolbreak\includeonetwocol{}{\hspace{0.2in}\times}
e^{2(K-r)\log(\frac{e}{\epsilon}) - E_{PV}(\theta,a,M(K-r))}  \nonumber\\ 
& \overset{(e)}{\leq}   2 e^{2(K-r)\log(\frac{e}{\epsilon}) - E_{PV}(\theta,a,M(K-r))} \label{ML.bound.1}
\end{align}
where $(a)$ holds by Chernoff bound and because ${{a}\choose{b}}\leq (\frac{ea}{b})^{b}$, $(b)$ holds from Lemma~\ref{Lem.1} in Appendix~\ref{App.1}, $(c)$ holds because $E_{PV}(\theta,a,M(K-r)) = E_{QU}(\theta,a,M(K-r)) - \theta$, $(d)$ holds because $ a \geq \frac{(K-r)(K-1)}{2}$, $r \leq (1-\epsilon)K$ and $(e)$ holds by assuming that $\liminf_{n\to\infty} (K-1)D(P||Q)+ 2MD(V||U) > 2\log(\frac{n}{K})$, which implies that 
\[
(1-\eta)((\frac{K-1}{2})D(P||Q)+MD(V||U)) - \log(\frac{n-K}{K}) \geq 0.
\]

Lemma~\ref{Lem.1} in Appendix~\ref{App.1} shows that 
\[E_{PV}(\theta,a,M(K-r)) \geq C (aD(P||Q)+(K-r)MD(V||U)]).
\]
 Using $ a \geq \frac{(K-r)(K-1)}{2}$ and substituting in~\eqref{ML.bound.1},
\begin{align}
\label{ML.bound.2}\nonumber
\mathbb{P}(R=r) \leq & 2 e^{-(K-r) \big( C ( \frac{K-1}{2}D(P||Q) + MD(V||U)) - 2 \log(\frac{e}{\epsilon}) \big) }  \\ 
\leq & 2 e^{-(K-r) \big( \frac{C}{2} ( (K-1)D(P||Q) + MD(V||U)) - 2 \log(\frac{e}{\epsilon}) \big) } 
\end{align}

Choose $\epsilon = \big((K-1)D(P||Q) + MD(V||U)\big)^{-\frac{1}{2}}$ and let $E = \big( \frac{C}{2} ( (K-1)D(P||Q) + MD(V||U)) - 2 \log(\frac{e}{\epsilon}) \big)$. Thus,
\begin{align}
\label{ML.bound.all}\nonumber
\mathbb{P}(R \leq (1-\epsilon)K) = & \sum_{r = 0}^{(1-\epsilon)K} \mathbb{P}(R=r) \leq \sum_{r = 0}^{(1-\epsilon)K} 2 e^{-(K-r) E } \\ 
\overset{(a)}{\leq} & 2 \sum_{r' = \epsilon K}^{\infty} e^{-r'E} \leq  2 \frac{e^{-\epsilon K E}}{1- e^{-E}} \overset{(b)}{\leq} o(1)
\end{align}
where $(a)$ holds by defining $r' = K-r$ and $(b)$ holds by assuming that $(K-1)D(P||Q) + MD(V||U) \to \infty$ and by the choice of $\epsilon$. This concludes the proof of Theorem~\ref{The.1}.


\section{Proof of Lemma~\ref{Suff.Random}}
\label{App.4.1}
Recall the definition of $\hat{C}$ from~\eqref{ML_rule.1}. Note that under the conditions of this Lemma, $\hat{C}$ may no longer be the maximum likelihood solution because $|C^{*}|$ need not be $K$. Let $|C^{*}| = K'$. Then, by assumption, with probability converging to one, $|K' - K| \leq \frac{K}{\log(K)}$. Let $R = |\hat{C} \cap C^{*}|$. Thus, $| \hat{C} \triangle C^{*}| = K + K' - 2R$. Hence, it is sufficient to show that $\mathbb{P}(R \leq (1-\epsilon)K - |K' - K|) = o(1)$, where $\epsilon$ is defined in the statement of the Lemma. Let $a = {{K}\choose{2}} - {{r}\choose{2}}$ and $a' = {{K'}\choose{2}} - {{r}\choose{2}}$, then for any $r \leq (1-\epsilon)K - |K' - K|$ and by the choice of $\epsilon$, the following holds as $n \to \infty$:
\begin{align}
& \frac{K}{K'} \to 1 \text{ , } \frac{K-r}{K'-r} \to 1  \text{ , }  \frac{a}{a'} \to 1  \label{same}
\end{align}

Following similar ideas as the proof of Theorem~\ref{The.1}: 
\begin{align}
&\mathbb{P}(R=r)  \nonumber \\
& \leq \mathbb{P}\Big( \big\{ C \in \{1,\cdots,n\}: |C| = K, |C\cap C^{*}| = r, e_{1}(\hat{C},\hat{C}) + \twocolbreak \hspaceonetwocol{0in}{0.15in}
e_{2}(\hat{C}) - e_{1}(C^{*},C^{*}) - e_{2}(C^{*}) \geq 0   \big\} \Big)  \nonumber \\ 
& = \mathbb{P} \Big( \big\{ S \subset C^{*}, T \subset (C^{*})^{c}: |S|=K'-r, =|T|= K-r, \nonumber \\
 & \hspace{0.15in} e_{1}(S,C^{*}) + e_{2}(S) \leq  e_{1}(T,T) + e_{1}(T,C^{*}\backslash S) + e_{2}(T)   \big\} \Big) \nonumber \\ 
& \leq \mathbb{P}\Big( \big\{ S \subset C^{*}: |S| = K'-r, e_{1}(S,C^{*}) + e_{2}(S) \leq \theta  \big\}  \nonumber \\ 
& \hspaceonetwocol{0in}{0.15in} \cup \big\{\exists S \subset C^{*}, T \subset (C^{*})^{c}: |S|=K'-r,|T|= K-r, \twocolbreak \hspaceonetwocol{0in}{0.15in}
e_{1}(T,T) + e_{1}(T,C^{*}\backslash S) + e_{2}(T)  \geq \theta \big\} \Big) \label{subset.random}
\end{align}
where $\theta = (1-\eta)(aD(P||Q)+(K-r)MD(V||U))$, for some $\eta \in (0,1)$. Using~\eqref{subset.random} and a union bound,
\begin{align} 
& \mathbb{P}(R=r)  \nonumber \\
& \overset{(a)}{\leq} {{K'}\choose{K'-r}} \mathbb{P}\big(\sum_{i=1}^{a'} L_{G,i} + \sum_{j=1}^{K'-r} \sum_{m=1}^{M} L_{S,m,j} \leq \theta \big) \nonumber \\ 
& \hspaceonetwocol{0in}{0.08in} + {{K'}\choose{K'-r}}{{n-K'}\choose{K-r}} \mathbb{P}\big(\sum_{i=1}^{a} L_{G,i} + \sum_{j=1}^{K-r} \sum_{m=1}^{M} L_{S,m,j} \geq \theta \big) \nonumber \\ 
& \overset{(b)}{\leq}  e^{(K'-r) \log(\frac{K'e}{K'-r})}  \twocolbreak \includeonetwocol{}{\hspace{0.2in}\times}
 e^{- \sup_{t \geq 0} -t \theta - a'\log_{P}(\mathbb{E}[e^{-tL_{G}}]) - M(K'-r)\log_{V}(\mathbb{E}[e^{-tL_{S}}])} \nonumber \\ 
& \hspaceonetwocol{0in}{0.15in} +  e^{(K'-r) \log(\frac{K'}{(K'-r)}) + (K-r)\log(\frac{(n-K)e}{(K-r)})}  \twocolbreak \includeonetwocol{}{\hspace{0.2in}\times}
e^{-\sup_{t\geq 0} t\theta - a\log_{Q}(\mathbb{E}[e^{tL_{G}}]) - M(K-r)\log_{U}(\mathbb{E}[e^{tL_{S}}])} \nonumber \\ 
 & \overset{(c)}{\leq} e^{(K'-r) \log(\frac{K'e}{K'-r}) - (1-o(1)) E_{PV}(\theta,a,M(K-r))} \twocolbreak 
\hspaceonetwocol{0in}{0.15in} + e^{(K'-r) \log(\frac{K'}{(K'-r)}) + (K-r)\log(\frac{(n-K)e}{(K-r)}) - E_{QU}(\theta,a,M(K-r))} \nonumber \\ 
& \overset{(d)}{=} e^{(K-r) \log(\frac{Ke}{K-r}) (1+o(1)) - E_{PV}(\theta,a,M(K-r))(1+o(1))}  \twocolbreak
\hspaceonetwocol{0in}{0.15in} + e^{(K-r) \log(\frac{(n-K)Ke^{2}}{(K-r)^{2}})(1+o(1)) - E_{PV}(\theta,a,M(K-r)) - \theta} \nonumber \\ 
& \overset{(e)}{\leq} e^{(K-r)(1+o(1)) \log(\frac{Ke}{K-r}) - (1+o(1))E_{PV}(\theta,a,M(K-r))}\nonumber  \\ 
& \hspaceonetwocol{0in}{0.06in} + e^{-(K-r)(1+o(1)) \big( (1-\eta)((\frac{K-1}{2})D(P||Q)+MD(V||U)) - \log(\frac{n-K}{K}) \big)}  \twocolbreak \includeonetwocol{}{\hspace{0.2in}\times}
e^{2(1+o(1))(K-r)\log(\frac{e}{\epsilon}) - E_{PV}(\theta,a,M(K-r))} \nonumber \\ 
& \overset{(f)}{\leq} 2 e^{2(K-r)(1+o(1))\log(\frac{e}{\epsilon}) - (1+o(1))E_{PV}(\theta,a,M(K-r))} \label{ML.bound.1.random}
\end{align}
where $(a)$ holds for $L_{G,i} (L_{S,m,j})$ be i.i.d copies of $L_{G}(L_{S})$, respectively, $(b)$ holds by Chernoff bound and because ${{a}\choose{b}}\leq (\frac{ea}{b})^{b}$, $(c)$ holds by using~\eqref{same} and by Lemma~\ref{Lem.1} in Appendix~\ref{App.1}, $(d)$ holds by using~\eqref{same} and because $E_{PV}(\theta,a,M(K-r)) = E_{QU}(\theta,a,M(K-r)) - \theta$, $(e)$ holds because $ a \geq \frac{(K-r)(K-1)}{2}$, $r \leq (1-\epsilon)K$ and $(f)$ holds by assuming that $\liminf_{n\to\infty} (K-1)D(P||Q)+ 2MD(V||U) \geq 2\log(\frac{n}{K})$, which implies that $(1-\eta)((\frac{K-1}{2})D(P||Q)+MD(V||U)) - \log(\frac{n-K}{K}) \geq 0$.

The remainder of the proof follows similarly to Appendix~\ref{App.4} following~\eqref{ML.bound.1}.

\section{Proof of Lemma~\ref{The.5.new}}
\label{App.5}


\begin{lemma}
\label{Lem.suff.exact.new}
Suppose that~\eqref{Cond.1.The.3.new} holds. Let $\{W_{\ell}\}$ and $\{\tilde{W}_{\ell}\}$ denote sequences of i.i.d. copies of $L_{G}$ under $P$ and $Q$, respectively. Also, for any node $i$, let $Z$ and $\tilde{Z}$ denote $\sum_{m=1}^{M} L_{S}(i,m)$ under $V$ and $U$, respectively. Then, for sufficiently small, but constant, $\delta$ and $\gamma = \frac{\log(\frac{n}{K})}{K}$:
\begin{align} 
\mathbb{P}\big( \sum_{\ell=1}^{K(1-\delta)}  \tilde{W}_{\ell} + \tilde{Z}  \geq K(1-\delta)\gamma \big) = o(\frac{1}{n}) \label{suff._exact.eq.2.new} \\
\mathbb{P}\big( \sum_{\ell=1}^{K(1-2\delta)} W_{\ell} + \sum_{\ell=1}^{\delta K} \tilde{W}_{\ell} + Z \leq K(1-\delta)\gamma \big) = o(\frac{1}{K}) \label{suff._exact.eq.3.new} 
\end{align}
\end{lemma}

\begin{proof}
By Chernoff bound:
\begin{align}
&\mathbb{P}\big( \sum_{\ell=1}^{K(1-\delta)}  \tilde{W}_{\ell} + \tilde{Z}  \geq K(1-\delta)\gamma \big) \twocolbreak
\leq e^{-(1-\delta) \sup_{t \geq 0} tK\gamma - K\log(\mathbb{E}_{Q}[e^{tL_{G}}]) - \frac{M}{1-\delta}\log(\mathbb{E}_{U}[e^{tL_{S}}])} \label{Lem.suff.exact.eq.1.new}
\end{align}

From~\eqref{Eq.Hajek} it follows that  for some positive $\epsilon_{o}$:  
\begin{align}
K\gamma & \leq \frac{KD(P||Q)}{2+\epsilon_{o}} + \frac{MD(V||U)}{1+\frac{\epsilon_{o}}{2}} \nonumber \\
& \leq KD(P||Q) + MD(V||U) \nonumber \\
& \leq KD(P||Q) + \frac{M}{1-\delta} D(V||U)
\end{align}
Hence, using Lemma~\ref{Lem.1} in Appendix~\ref{App.1}, $\sup_{t \geq 0}$ is replaced by $\sup_{t \in [0,1]}$. Also, $ \log(\mathbb{E}_{U}[e^{tL_{S}}]) = (t-1)D_{t}(V||U) \leq 0$ where the first equality holds by the definition of the R\'{e}nyi-divergence between distributions $V$ and $U$~\cite{polyanski} and the second inequality because $t \in [0,1]$. This implies that $\frac{M}{1-\delta}\log(\mathbb{E}_{U}[e^{tL_{S}}]) \leq M\log(\mathbb{E}_{U}[e^{tL_{S}}])$. Substituting in~\eqref{Lem.suff.exact.eq.1.new}:
\begin{align}
\mathbb{P}\big( \sum_{\ell=1}^{K(1-\delta)}  \tilde{W}_{\ell} + \tilde{Z}  \geq K(1-\delta)\gamma \big) 
&\leq e^{-(1-\delta) E_{QU}(K\gamma,K,M)} \nonumber \\ 
& \leq e^{-(1-\delta)(1+\epsilon)\log(n)}\label{Lem.suff.exact.eq.2.new}
\end{align}
where~\eqref{Lem.suff.exact.eq.2.new} follows since~\eqref{Cond.1.The.3.new} holds by assumption, i.e., there exists $\epsilon \in (0,1): E_{QU}(K\gamma,K,M) \geq (1+\epsilon)\log(n)$. Equation~\eqref{Lem.suff.exact.eq.2.new} implies that~\eqref{suff._exact.eq.2.new} holds for sufficiently small $\delta$. 

To show~\eqref{suff._exact.eq.3.new}, Chernoff bound is used:
\begin{align}
& \mathbb{P}\big( \sum_{\ell=1}^{K(1-2\delta)} W_{\ell} + \sum_{\ell=1}^{\delta K} \tilde{W}_{\ell} + Z \leq K(1-\delta)\gamma \big) \nonumber \\ 
& \overset{(a)}{\leq} e^{tK\gamma(1-\delta) + K(1-2\delta)\log(\mathbb{E}_{P}[e^{-tL_{G}}]) + K\delta\log(\mathbb{E}_{Q}[e^{-tL_{G}}])}  \twocolbreak \includeonetwocol{}{\hspace{0.2in}\times}
e^{M(1-\delta)\log(\mathbb{E}_{V}[e^{-tL_{S}}]) + M\delta\log(\mathbb{E}_{U}[e^{-tL_{S}}])} \nonumber \\ 
& = e^{(1-2\delta)(tK\gamma  + K\log(\mathbb{E}_{P}[e^{-tL_{G}}]) + M\frac{1-\delta}{1-2\delta}\log(\mathbb{E}_{V}[e^{-tL_{S}}]))} \twocolbreak \includeonetwocol{}{\hspace{0.2in}\times}
e^{\delta(tK\gamma+ K\log(\mathbb{E}_{Q}[e^{-tL_{G}}]) + M\log(\mathbb{E}_{U}[e^{-tL_{S}}]))} \nonumber \\ 
& \overset{(b)}{\leq} e^{(1-2\delta)(tK\gamma  + K\log(\mathbb{E}_{P}[e^{-tL_{G}}]) + M\log(\mathbb{E}_{V}[e^{-tL_{S}}]))} \twocolbreak \includeonetwocol{}{\hspace{0.2in}\times}
e^{\delta(tK\gamma+ K\log(\mathbb{E}_{Q}[e^{-tL_{G}}]) + M\log(\mathbb{E}_{U}[e^{-tL_{S}}]))} \label{Lem.suff.exact.eq.3.new} 
\end{align}
where $(a)$ and $(b)$ hold because $\frac{1-\delta}{1-2\delta} \geq 1$ for sufficiently small $\delta$ and $\log(\mathbb{E}_{V}[e^{-tL_{S}}]) = (t-1)D_{t}(U||V) \leq t D_{t+1}(U||V) = \log(\mathbb{E}_{U}[e^{-tL_{S}}])$, where $D_{t}(V||U)$ is the R\'{e}nyi-divergence between distributions $V$ and $U$, which is non-decreasing in $t \geq 0$~\cite{polyanski}.

By definition $-E_{PV}(K\gamma,K,M) = -\sup_{\lambda \in [-1,0]} \lambda K\gamma - K\log(\mathbb{E}_{P}[e^{\lambda L_{G}}]) - M\log(\mathbb{E}_{V}[e^{\lambda L_{S}}]) = -\lambda^{*}K\gamma + K\log(\mathbb{E}_{P}[e^{\lambda^{*}L_{G}}]) + M\log(\mathbb{E}_{V}[e^{\lambda^{*}L_{S}}])$. Hence, by choosing $t = -\lambda^{*} \in [0,1]$ and substituting in~\eqref{Lem.suff.exact.eq.3.new},
\begin{align}
& \mathbb{P}\big( \sum_{\ell=1}^{K(1-2\delta)} W_{\ell} + \sum_{\ell=1}^{\delta K} \tilde{W}_{\ell} + Z \leq K(1-\delta)\gamma \big) \twocolbreak
\leq e^{-(1-2\delta) E_{PV}(K\gamma,K,M)} \twocolbreak \includeonetwocol{}{\hspace{0.2in}\times}
e^{\delta(K\gamma+ K\log(\mathbb{E}_{Q}[e^{-tL_{G}}]) + M\log(\mathbb{E}_{U}[e^{-tL_{S}}]))} \label{Lem.suff.exact.eq.4.new}
\end{align}

By Lemma~\ref{Lem.2} and convexity of $\psi_{QU}(t,m_{1},m_{2})$: 
\begin{align}
\psi_{QU}(-t,K,M) & \leq \psi_{QU}(-1,K,M) \twocolbreak \leq A (KD(Q||P) + MD(U||V))\label{Lem.suff.exact.eq.5.new}
\end{align}
for some positive constant $A$. Moreover, by Lemma~\ref{Lem.2}, $E_{QU}(K\gamma,K,M) \geq E_{QU}(0,K,M) \geq A_{1}(KD(Q||P) + MD(U||V))$, for some positive constant $A_{1}$. Hence, by substituting in~\eqref{Lem.suff.exact.eq.4.new}, for some positive constant $A_{2}$:
\begin{align} 
& \mathbb{P}\big( \sum_{\ell=1}^{K(1-2\delta)} W_{\ell} + \sum_{\ell=1}^{\delta K} \tilde{W}_{\ell} + Z \leq K(1-\delta)\gamma \big) \twocolbreak
\leq e^{-(1-2\delta) E_{PV}(K\gamma,K,M)+ \delta(K\gamma+ A_{2} E_{QU}(K\gamma,K,M))} \nonumber\\ 
& \overset{(a)}{\leq} e^{-E_{QU}(K\gamma,K,M) (1 - 2\delta - \delta A_{2}) + (1-\delta)K\gamma} \nonumber \\ 
& \overset{(b)}{=} e^{- \log(n)( (1+\epsilon) (1 - 2\delta - \delta A_{2}) + \delta -1) - \log(K) (1-\delta)} \nonumber \\
& \overset{(c)}{=} o(\frac{1}{K}) \label{Lem.suff.exact.eq.6.new}
\end{align} 
where $(a)$ holds because $E_{PV}(K\gamma,K,M) = E_{QU}(K\gamma,K,M) - K\gamma$ from Lemma~\ref{Lem.2}, $(b)$ holds by the assumption that~\eqref{Cond.1.The.3.new} holds, which implies that there exists $\epsilon \in (0,1): E_{QU}(K\gamma,K,M) \geq (1+\epsilon)\log(n)$ and $(c)$ holds for sufficiently small $\delta$. 

Equations \eqref{Lem.suff.exact.eq.2.new} and \eqref{Lem.suff.exact.eq.6.new} concludes the proof of Lemma~\ref{Lem.suff.exact.new}.
\end{proof}

Define the event $E \triangleq \{(\hat{C}_k,C^*_k)  :  |\hat{C}_{k} \triangle C^{*}_{k}| \leq \delta K \quad \forall k\}$; then conditioned on $E$ we have:
\begin{align*}
|\hat{C}_{k} \cap C^{*}_{k}| &\geq |\hat{C}_{k}| - |\hat{C}_{k} \triangle C^{*}_{k}| \twocolbreak = \lceil K(1-\delta) \rceil - |\hat{C}_{k} \triangle C^{*}_{k}| \twocolbreak \geq K(1-2 \delta)
\end{align*}
Thus, in Algorithm~\ref{Alg.1}, for nodes $i$ within the community $C^*$, $r_{i}$ is stochastically greater than or equal to $(\sum_{\ell=1}^{K(1-2\delta)} W_{\ell}) + (\sum_{\ell=1}^{K\delta} \tilde{W_{\ell}}) + Z$ by Lemma~\ref{Lem.suff.exact.new} and~\eqref{Lem.suff.exact.eq.3.new}. For $i \notin C^{*}$, $r_{i}$ has the same distribution as $(\sum_{\ell=1}^{K(1-\delta)} \tilde{W_{\ell}}) + \tilde{Z}$. Thus, by Lemma~\ref{Lem.suff.exact.new}, with probability converging to 1, 
\begin{align*}
r_{i} > K(1-\delta)\gamma, &\qquad  i \in C^{*}\\
r_{i} < K(1-\delta)\gamma, & \qquad i \notin C^{*}
\end{align*}
Hence, $\mathbb{P}(\tilde{C} = C^{*}) \to 1$ as $n \to \infty$.


\section{Sufficiency of Theorem~\ref{The.3.new}}
\label{App.6}

The cardinality $|C^{*}_{k}|$ is a random variable that corresponds to sampling, without replacement, from the nodes of the original graph. Let $Z$ be a binomial random variable $\text{Bin}(n(1-\delta),\frac{K}{n})$.  The Chernoff bound for $Z$:
\begin{align}
\mathbb{P}\bigg(\Big |Z - (1-\delta)K \Big| \geq \frac{K}{\log(K)} \bigg) &\leq e^{-\Omega(\frac{K}{\log^{2}(K)})} 
\label{eq:Chernoff-Z}
\end{align}
A result of Hoeffding~\cite[Theorem 4]{hoeffding1963probability} for sampling with and without replacement indicates that  $\mathbb{E}[\phi(|C^{*}_{k}|)] \leq \mathbb{E}[\phi(Z)]$ for any convex $\phi$. This can be applied to~\eqref{eq:Chernoff-Z} on the negative and positive side, individually. Putting them back together, we get a bound on the tails of $|C^{*}_{k}|$:
\begin{align}
\mathbb{P}\bigg( \Big| |C^{*}_{k}| - (1-\delta)K \Big| \geq \frac{K}{\log(K)} \bigg) 
&\leq e^{-\Omega(\frac{K}{\log^{2}(K)})} \nonumber \\ 
& \overset{}{\leq} o(1)\label{Suff.exact.eq.1.new}
\end{align}

Since~\eqref{Eq.Hajek} holds, for sufficiently small $\delta$,
\[
\liminf_{n \to \infty} \lceil (1-\delta)K \rceil D(P||Q) + 2MD(V||U) > 2 \log(\frac{n}{K})
\]
which together with~\eqref{Suff.exact.eq.1.new} indicates, via Lemma~\ref{Suff.Random}, that ML achieves weak recovery. Thus, for any $1 \leq k \leq \frac{1}{\delta}$:
\begin{equation}
\mathbb{P}\Big(\frac{ |\hat{C}_{k}\triangle C^{*}_{k} |}{K}  \leq 2\epsilon + \frac{1}{\log(K)}\Big) \geq 1 - o(1)
\end{equation}
with $\epsilon = o(1)$. Since $\delta$ is constant, by the union bound 
\begin{equation}
\mathbb{P}\Big(\frac{ |\hat{C}_{k}\triangle C^{*}_{k} |}{K}  \leq 2\epsilon + \frac{1}{\log(K)} ,\quad\forall k\Big) \geq 1 - o(1)
\end{equation}
Since $\epsilon = o(1)$, the desired~\eqref{suff._exact.eq.1.new} holds.


\section{Necessity of Theorem~\ref{The.3.new}}
\label{App.7} 
The following Lemma characterizes necessary conditions that are weaker than needed for Theorem~\ref{The.3.new}, i.e., the Lemma is stronger than needed at this point, but will subsequently be used for unbounded LLR as well.

\begin{lemma}
\label{Lem.necc.exact}
Let $\{W_{\ell}\}$ and $\{\tilde{W}_{\ell}\}$ denote sequences of i.i.d. copies of $L_{G}$ under $P$ and $Q$, respectively. 
For any node $i$ inside the community, let $Z$ denote a random variable drawn according to the distribution of $\sum_{m=1}^{M} L_{S}(i,m)$. Let $\tilde{Z}$ be the corresponding random variable when $i$ is {\em outside} the community.
Let $K_{o} \to \infty$ such that $K_{o} = o(K)$. 
Then, for any estimator $\hat{C}$ achieving exact recovery, there exists a sequence $\theta_{n}$ such that for sufficiently large $n$:
\begin{align}
\mathbb{P}\Big(\sum_{\ell=1}^{K-K_{o}} W_{\ell} + Z \leq (K-1)\theta_{n} - \tilde{\theta}_n\Big) &\leq \frac{2}{K_{o}} \label{necc.general.1.1}\\
\mathbb{P}\Big(\sum_{\ell=1}^{K-1} \tilde{W}_{\ell} + \tilde{Z} \geq (K-1)\theta_{n}\Big) &\leq \frac{1}{n-K} \label{necc.general.1.2}
\end{align}
where 
\begin{equation}
\tilde{\theta}_n \triangleq (K_{o}-1)D(P||Q) + 6\sqrt{K_o}\sigma  \label{sigma.1}
\end{equation} 
and $\sigma^{2}$ is the variance of $L_{G}$ under $P$.
\end{lemma}

\begin{proof}

Recall that ML is optimal for exact recovery since $C^{*}$ is chosen uniformly. Assume $\mathbb{P}(\text{ML fails}) =o(1)$. Define
\begin{align}
i_{o} & \triangleq  \arg\min_{i \in C^{*}} e_{1}(i,C^{*}) + \sum_{m=1}^{M} L_{S}(i,m) \nonumber \\
\tilde{C} & \triangleq C^{*}\backslash\{i_{o}\}\cup{\{j\}} \text{ for } j \notin C^{*}
\end{align}
Also, define the following event:
\begin{align}
F_{M} &\triangleq \Big\{ (\boldsymbol{G,Y}) : \min_{i \in C^{*}}e_{1}(i,C^{*}) + \sum_{m=1}^{M} L_{S}(i,m) \twocolbreak \hspaceonetwocol{0in}{0.3in} \leq \max_{j \notin C^{*}} e(j,C^{*}\backslash\{i_{o}\}) + \sum_{m=1}^{M} L_{S}(j,m)\Big \}
\end{align}
Since $\mathbb{P}(\text{ML fails}) = o(1)$,  using~\eqref{ML_rule.1}:
\begin{align}
& e_{1}(\tilde{C},\tilde{C}) + e_{2}(\tilde{C}) - e_{1}(C^{*},C^{*}) - e_{2}(C^{*}) \nonumber \\
& = \Big(e(j,C^{*}\backslash\{i_{o}\}) + \sum_{m=1}^{M} L_{S}(j,m) \Big) \twocolbreak \hspaceonetwocol{0in}{0.15in} - \Big(e_{1}(i,C^{*}) + \sum_{m=1}^{M} L_{S}(i,m)\Big) 
\label{eq:BadGraphWins}
\end{align}
For observations belonging to $F_{M}$, the expression~\eqref{eq:BadGraphWins} is non-negative, implying ML fails with non-zero probability. Then,
\begin{equation}
\mathbb{P}(F_{M}) \leq \mathbb{P}(\text{ML fails}) = o(1)
\end{equation}
since ML achieves exact recovery. 

Define $\theta_{n}'$, $\theta_{n}''$  and the events $E_{1}$ and $E_{2}$ as follows:
\begin{align}
\twocolAlignMarker \theta_{n}'  \triangleq \twocolnewline& \inf \bigg\{ x \in \mathbb{R}: \mathbb{P}\Big(\sum_{\ell=1}^{K-K_{o}} W_{\ell} + Z \leq (K-1)x - \tilde{\theta}_n\Big) \geq \frac{2}{K_{o}}  \bigg\} \label{theta.1} \\
\twocolAlignMarker\theta_{n}'' \triangleq \twocolnewline&\sup \bigg\{ x \in \mathbb{R}: \mathbb{P}
\Big (\sum_{\ell=1}^{K-1} \tilde{W}_{\ell} + \tilde{Z} \geq (K-1)x\Big) \geq \frac{1}{n-K}  \bigg\} \label{theta.2}\\
\twocolAlignMarker E_{1} \triangleq \twocolnewline& \Big\{ (\boldsymbol{G, Y}) :\max_{j \notin C^{*}} \Big(e(j,C^{*}\backslash\{i_{o}\}) + \sum_{m=1}^{M} L_{S}(j,m)\Big) \geq (K-1)\theta_{n}'' \Big\}
 \label{E.1}\\ 
\twocolAlignMarker E_{2} \triangleq \twocolnewline& 
 \Big\{ (\boldsymbol{G, Y}) : \min_{i \in C^{*}}\Big(e_{1}(i,C^{*}) + \sum_{m=1}^{M} L_{S}(i,m)\Big) \leq (K-1)\theta_{n}' \Big\}
\label{E.2}
\end{align}
where $\tilde{\theta}_{n}$ is defined in~\eqref{sigma.1}.

\begin{lemma}
$\mathbb{P}(E_{1}) = \Omega(1)$ and $\mathbb{P}(E_{2}) = \Omega(1)$.
\end{lemma}
\begin{proof}
\begin{align}
& \mathbb{P}(E_{1}) \twocolbreak
\overset{(a)}{=} 1- \prod_{j \notin C^{*}} \mathbb{P}\Big(e(j,C^{*}\backslash\{i_{o}\}) + \sum_{m=1}^{M} L_{S}(j,m) < (K-1)\theta_{n}'' \Big) \nonumber \\ 
& = 1 -  \Big( 1 
 - \mathbb{P}\big(e(j,C^{*}\backslash\{i_{o}\})  \twocolbreak \hspaceonetwocol{0in}{1.25in} +  \sum_{m=1}^{M} L_{S}(j,m) \geq (K-1)\theta_{n}'' \big) \Big)^{n-K} \nonumber\\ S
& \overset{(b)}{\geq} 1 - e^{\Big(-(n-K) \mathbb{P}\big(e(j,C^{*}\backslash\{i_{o}\}) + \sum_{m=1}^{M} L_{S}(j,m) \geq (K-1)\theta_{n}'' \big)\Big)} \nonumber \\
& \overset{(c)}{\geq} 1 - e^{-1}
\end{align}
where $(a)$ holds because $e(j,C^{*}\backslash\{i_{o}\}) + \sum_{m=1}^{M} L_{S}(j,m)$ are i.i.d. for all $j \notin C^{*}$, $(b)$ holds because $1-x \leq e^{-x}$ $\forall x \in \mathbb{R}$ and $(c)$ holds by definition of $\theta_{n}''$. Thus, $\mathbb{P}(E_{1}) = \Omega(1)$.

To show $\mathbb{P}(E_{2})=\Omega(1)$, we are confronted with the difficulty that $e_{1}(i,C^{*})$ are not independent. Let $T$ be the set of the first $K_{o}$ indices in $C^{*}$, where $K_{o} \to\infty$ such that $K_{o} = o(K)$. Also, let $T' = \{ i\in T: e_{1}(i,T) \leq \tilde{\theta}_{n} \}$. Then,
%
\begin{align}
& \min_{i \in C^{*}} e_{1}(i,C^{*}) + \sum_{m=1}^{M} L_{S}(i,m) \twocolbreak
 \leq \min_{i \in T'} e_{1}(i,C^{*}) + \sum_{m=1}^{M} L_{S}(i,m) \nonumber \\
&  \leq \min_{i \in T'} e_{1}(i,C^{*}\backslash T) + \sum_{m=1}^{M} L_{S}(i,m) + \tilde{\theta}_{n}
\end{align}
It follows that:
\begin{align}
& \mathbb{P}(E_{2}) \nonumber \\ &\geq \mathbb{P}\Big(\min_{i \in T'}
  e_{1}(i,C^{*}\backslash T) + \sum_{m=1}^{M} L_{S}(i,m) \leq
  (K-1)\theta_{n}' \! - \! \tilde{\theta}_{n}\Big) \nonumber\\ &
  \overset{(a)}{=} 1 - \mathbb{P}\Big(\bigcap_{i \in T'} \Big\{
  e_{1}(i,C^{*}\backslash T) + \sum_{m=1}^{M} L_{S}(i,m) \twocolbreak
  \hspaceonetwocol{0in}{0.75in} > (K-1)\theta_{n}' -
  \tilde{\theta}_{n} \Big\} \Big) \nonumber\\ & = 1 -
  \mathbb{P}\bigg(\bigcap_{i \in T'} \Big\{ e_{1}(i,C^{*}\backslash
  T) + \sum_{m=1}^{M} L_{S}(i,m) \twocolbreak
  \hspaceonetwocol{0in}{0.5in} > (K-1)\theta_{n}' -
  \tilde{\theta}_{n} \Big\} \bigg| \;|T'| \geq \frac{K_{o}}{2}\bigg )
  \mathbb{P}\big(|T'| \geq \frac{K_{o}}{2}\big) \nonumber \\ &
  \hspaceonetwocol{0in}{0.05in} - \mathbb{P}\bigg(\bigcap_{i \in T'}
  \Big\{ e_{1}(i,C^{*}\backslash T) + \sum_{m=1}^{M} L_{S}(i,m)
  \twocolbreak \hspaceonetwocol{0in}{0.25in} > (K-1)\theta_{n}' -
  \tilde{\theta}_{n} \Big\} \bigg| \;|T'| < \frac{K_{o}}{2}\bigg)
  \times \mathbb{P}\big(|T'| < \frac{K_{o}}{2}\big)\nonumber \\ &
  \geq 1- \mathbb{P}\bigg(\bigcap_{i \in T'} \Big\{
  e_{1}(i,C^{*}\backslash T) + \sum_{m=1}^{M} L_{S}(i,m) \twocolbreak
  \hspaceonetwocol{0in}{0.25in} > (K-1)\theta_{n}' -
  \tilde{\theta}_{n} \Big\} \bigg| \;|T'| \geq \frac{K_{o}}{2}\bigg) -
  \mathbb{P}\big(|T'| < \frac{K_{o}}{2}\big)\nonumber \\ & \geq 1-
  \Big(1- \mathbb{P}\big(\sum_{\ell=1}^{K-K_{o}} W_{\ell} + Z \leq
  (K-1)\theta_{n}' - \tilde{\theta_{n}}\big)\Big)^{\frac{K_{o}}{2}}
  \twocolbreak \hspaceonetwocol{0in}{0.25in} - \mathbb{P}\big(|T'| <
  \frac{K_{o}}{2}\big)\nonumber \\ &\overset{(b)}{\geq} 1-
  e^{\Big(-(\frac{K_{o}}{2})\mathbb{P}\big(\sum_{\ell=1}^{K-K_{o}}
    W_{\ell} + Z \leq (K-1)\theta_{n}' -
    \tilde{\theta_{n}}\big)\Big)} \twocolbreak
  \hspaceonetwocol{0in}{0.25in} - \mathbb{P}\big(|T'| <
  \frac{K_{o}}{2}\big) \nonumber \\ &\overset{(c)}{\geq} 1- e^{-1} -
  \mathbb{P}\big(|T'| < \frac{K_{o}}{2}\big)\nonumber
\end{align}
 where $(a)$ holds because $e_{1}(i,C^{*}\backslash T) + \sum_{m=1}^{M} L_{S}(i,m)$ are i.i.d. for all $i \in T'$, $(b)$ holds because $1-x \leq e^{-x}$  $\forall x \in \mathbb{R}$, $(c)$ holds by definition of $\theta_{n}'$.

To conclude the proof, it remains to show that $\mathbb{P}(|T'| < \frac{K_{o}}{2}) = o(1)$. Recall $T' = \{ i\in T: e_{1}(i,T) \leq \tilde{\theta}_{n} \}$. For $i \in T$, $e_{1}(i,T) = G_{i} + H_{i}$, where $G_{i} = e_{1}(i,\{1\cdots,i-1\})$ and $H_{i} = e_{1}(i,\{i+1\cdots,K_{o}\})$. Thus, by Chebyshev inequality:
\[
\mathbb{P}\Big(G_{i} \geq (i-1)D(P||Q) + 3\sqrt{K_o}\sigma \Big) \leq \frac{1}{9}
\]
for all $i \in T$. Therefore, $|\{i: G_{i} \leq (i-1)D(P||Q) + 3\sqrt{K_o}\sigma \}|$ is stochastically at least as large as a $\text{Bin}(K_{o},\frac{8}{9})$ random variable. Thus, 
\begin{equation}
\mathbb{P}\Big(\big|\{i: G_{i} \leq (i-1)D(P||Q) + 3\sqrt{K_o}\sigma \}\big| \geq \frac{3K_{o}}{4}\Big) \to 1 \label{Cheb.1}
\end{equation}
as $K_{o} \to \infty$. Similarly,
\begin{equation}
\mathbb{P}\Big(\big|\{i: H_{i} \leq (K_o-i)D(P||Q) + 3\sqrt{K_o}\sigma \}\big| \geq \frac{3K_{o}}{4}\Big) \to 1 \label{Cheb.2}
\end{equation}
as $K_{o} \to \infty$. Combining~\eqref{Cheb.1} and~\eqref{Cheb.2} and using the definition of $e_{1}(i,T)$:
\[
\mathbb{P}(|T'| \geq \frac{K_{o}}{2}) \xrightarrow{K_{o} \to \infty} 1
\]
which concludes the proof of the lemma. 
\end{proof}

By definition, $E_{1}$ and $E_{2}$ are independent. Since $\mathbb{P}(\text{ML fails}) =o(1)$ implies that $\mathbb{P}(F_M) = o(1)$:
\begin{align}
\mathbb{P}(E_{1}\cap E_{2} \cap F_{M}^{c}) & \geq \mathbb{P}(E_{1}\cap E_{2}) - \mathbb{P}(F_M) \twocolbreak
 = \mathbb{P}(E_{1})\mathbb{P}(E_{2}) - o(1) \nonumber \\
 & = \Omega(1) \label{last.1}
\end{align}
where~\eqref{last.1} holds since $\mathbb{P}(E_{1}) = \Omega(1)$ and $\mathbb{P}(E_{2}) = \Omega(1)$. 

It is easy to see that $E_{1}\cap E_{2} \cap F_{M}^{c} \subset \{\theta_{n}' > \theta_{n}'' \}$. It follows $\mathbb{P}(\theta_{n}' > \theta_{n}'') = \Omega(1)$ for sufficiently large $n$. Let $\theta_{n} = \frac{\theta_{n}' + \theta_{n}''}{2}$. For sufficiently large $n$, $\theta_{n} < \theta_{n}'$ and $\theta_{n} > \theta_{n}'$. Combining this with the definitions of $\theta_{n}'$ and $\theta_{n}''$, implies that~\eqref{necc.general.1.1} and~\eqref{necc.general.1.2} hold simultaneously.


\end{proof}

The necessity of Theorem~\ref{The.3.new} expresses the following: subject to conditions~\eqref{Eq.Hajek}, exact recovery implies~\eqref{Cond.1.The.3.new}. Lemma~\ref{Lem.necc.exact} shows that exact recovery implies~\eqref{necc.general.1.1} and~\eqref{necc.general.1.2}. It remains to be shown that ~\eqref{necc.general.1.1} and~\eqref{necc.general.1.2} imply~\eqref{Cond.1.The.3.new}. We show that by contraposition.

Assume~\eqref{Cond.1.The.3.new} does not hold, then for arbitrarily small $\epsilon > 0$ and sufficiently large $n$
\begin{align}
 E_{QU}\big(\log(\frac{n}{K}),K,M\big) &\leq (1-\epsilon)\log(n)
\end{align}
Let 
\[
\gamma \triangleq \frac{\log(\frac{n}{K})}{K}
\]
and define $S \triangleq \sum_{\ell=1}^{K-1} \tilde{W}_{\ell} + \tilde{Z}$ and $a \triangleq (K-1)\gamma + \delta$, for some $\delta > 0$. 
Since~\eqref{Eq.Hajek} holds, for sufficiently large $n$ and arbitrary small $\epsilon_{o} > 0$:
\begin{align}
K\gamma &\leq \frac{KD(P||Q)}{2+\epsilon_{o}} + \frac{MD(V||U)}{(1+\frac{\epsilon_{o}}{2})} \nonumber \\
& \leq \frac{1}{1+\frac{\epsilon_{o}}{2}} (KD(P||Q) + MD(V||U)) \nonumber \\
& \leq KD(P||Q) + MD(V||U) \label{necc.holds}
\end{align}
At $\theta_{n} = \gamma$:
\begin{align}
&\mathbb{P}\big(\sum_{\ell=1}^{K-1} \tilde{W}_{\ell} + \tilde{Z} \geq (K-1)\gamma\big) \twocolbreak
= \int\limits_{S \geq (K-1)\gamma}  \mathbb{P}\big(\tilde{w}_{1},\cdots,\tilde{w}_{K-1},\tilde{z}\big) \nonumber \\ 
& \overset{(a)}{\geq} \int\limits_{a-\delta \leq S \leq a+\delta} \big(\prod_{\ell=1}^{K-1}\mathbb{P}(\tilde{w}_{\ell})\big)\big(\mathbb{P}(\tilde{z})\big) \twocolbreak
\overset{(b)}{=} \int\limits_{a-\delta \leq S \leq a+\delta} \bigg(\frac{\mathbb{E}[e^{tS}] e^{tS}}{\mathbb{E}[e^{tS}] e^{tS}}\bigg) \big(\prod_{\ell=1}^{K-1}\mathbb{P}(\tilde{w}_{\ell})\big)\big(\mathbb{P}(\tilde{z})\big) \nonumber\\ 
&\overset{(c)}{\geq} e^{-ta - |t|\delta + \psi_{QU}(K-1,M,t)} \twocolbreak \includeonetwocol{}{\hspace{0.2in}\times}
\int\limits_{a-\delta \leq S \leq a+\delta} \bigg(\prod_{\ell=1}^{K-1} \frac{\mathbb{P}(\tilde{w}_{\ell})e^{t\tilde{w}_{\ell}}}{\mathbb{E}[e^{t\tilde{w}_{\ell}]}}\bigg) \bigg(\frac{\mathbb{P}(\tilde{z})e^{t\tilde{z}}}{\mathbb{E}[e^{t\tilde{z}}]}\bigg) \nonumber\\ 
& \overset{(d)}{=} e^{-ta - |t|\delta + \psi_{QU}(K-1,M,t)} \mathbb{P}_{\tilde{Q}\tilde{U}}\big(a-\delta \leq S \leq a+\delta\big)\nonumber \\ 
& \overset{(e)}{\geq} e^{-\big(ta-\psi_{QU}(K-1,M,t)\big) - |t|\delta }  \twocolbreak \includeonetwocol{}{\hspace{0.05in}\times}
\Bigg(1 - \twocolbreak \includeonetwocol{}{\hspace{0.05in}}   \frac{\big((K-1)\tilde{\sigma}^{2}_{L_{G}} + M\tilde{\sigma}^{2}_{L_{S}}\big) + \big((K-1)\tilde{\mu}_{L_{G}} + M\tilde{\mu}_{L_{S}} - a\big)^{2}}{\delta^{2}}\Bigg) \label{necc.proof.bounded.1}
\end{align}
where $(a)$ holds because $\tilde{W}_{\ell}$ are i.i.d. and independent of $\tilde{Z}$, $(b)$ holds for any $t \in \mathbb{R}$ such that $\mathbb{E}[e^{tS}]$ is finite, $(c)$ holds by the definition of $\psi_{QU}$ and because $a-\delta \leq S \leq a+\delta$, $(d)$ holds because $\frac{\mathbb{P}(\tilde{W}_{\ell})e^{t\tilde{W}_{\ell}}}{\mathbb{E}[e^{t\tilde{W}_{\ell}]}}$ and $\frac{\mathbb{P}(\tilde{Z})e^{t\tilde{Z}}}{\mathbb{E}[e^{t\tilde{Z}}]}$ define two new probability distributions $\tilde{Q}$ and $\tilde{U}$ over the same support of $Q$ and $U$, respectively and $(e)$ holds from Chebyshev's inequality and by defining $\tilde{\sigma}^{2}_{L_{G}}$, $\tilde{\mu}_{L_{G}}$, $\tilde{\sigma}^{2}_{L_{S}}$ and $\tilde{\mu}_{L_{S}}$ to be the variances and means of $L_{G}$ and $L_{S}$ under $\tilde{Q}$ and $\tilde{U}$, respectively.

Since $ta-\psi_{QU}(K-1,M,t)$ is concave in $t$, to find $t^{*} = \arg\sup_{t \in \mathbb{R}} ta-\psi_{QU}(K-1,M,t)$ we set the derivative to zero, finding 
\[
a = \psi'_{QU} = \frac{(K-1)\mathbb{E}_{Q}[L_{G}e^{tL_{G}}]}{\mathbb{E}_{Q}[e^{tL_{G}}]} + M\frac{\mathbb{E}_{U}[L_{S}e^{tL_{S}}]}{\mathbb{E}_{U}[e^{tL_{S}}]}.
\]
Also, by the definition of $\tilde{Q}$ and $\tilde{U}$, 
\begin{align*}
(K-1)\tilde{\mu}_{L_{G}} + M\tilde{\mu}_{L_{S}} &=  \frac{(K-1)\mathbb{E}_{Q}[L_{G}e^{tL_{G}}]}{\mathbb{E}_{Q}[e^{tL_{G}}]} + M\frac{\mathbb{E}_{U}[L_{S}e^{tL_{S}}]}{\mathbb{E}_{U}[e^{tL_{S}}]} \\
&= a.
\end{align*}
Thus, by substituting in~\eqref{necc.proof.bounded.1}:
\begin{align}
& \mathbb{P}_{QU}\big(\sum_{\ell=1}^{K-1} \tilde{W}_{\ell} + \tilde{Z} \geq (K-1)\gamma\big) \twocolbreak
\geq e^{-\big(t^{*}a-\psi_{QU}(K-1,M,t^{*})\big) - |t^{*}|\delta } \big(1 - \frac{(K-1)\tilde{\sigma}^{2}_{L_{G}} + M\tilde{\sigma}^{2}_{L_{S}}}{\delta^{2}}\big) \label{necc.proof.bounded.2}
\end{align}
By direct computation, and Lemma~\ref{Lem.2},
\begin{align}
(K-1)\tilde{\sigma}^{2}_{L_{G}} + M\tilde{\sigma}^{2}_{L_{S}} & = \psi_{QU}''(K-1,M,t) \nonumber \\
& \leq B\big((K-1) D(P||Q) +  MD(V||U)\big)
\end{align}
for some positive constant $B$. This allows us to eliminate the Chebyshev term (asymptotically) by setting
\[
\delta = \big((K-1) D(P||Q) +  MD(V||U)\big)^{\frac{2}{3}}.
\]
Moreover, for sufficiently large $n$:
\begin{align}
a &= (K-1)\gamma + \delta \nonumber \\
&\leq K\gamma + \delta \nonumber\\ 
& \overset{(a)}{\leq} \frac{KD(P||Q)}{2+\epsilon_{o}} + \frac{MD(V||U)}{1+\frac{\epsilon_{o}}{2}} \twocolbreak \includeonetwocol{}{\hspace{0.2in}} + (KD(P||Q) +  MD(V||U))^{\frac{2}{3}} \nonumber\\ 
& \overset{}{\leq} KD(P||Q)+MD(V||U) (\frac{1}{1+\frac{\epsilon_{o}}{2}} +o(1)) \label{necc.proof.bounded.3}
\end{align}
where $(a)$ holds from~\eqref{necc.holds}. Thus, for sufficiently large $n$,
\begin{align*}
-(K-1)D(Q||P)-\twocolAlignMarker MD(U||V) \le a \twocolbreak \le (K-1)D(P||Q)+MD(V||U)
\end{align*}
Hence, by Lemma~\ref{Lem.1}, 
\begin{align*}
t^{*} &= \arg\sup_{t \in \mathbb{R}} ta-\psi_{QU}(K-1,M,t) \\
&= \arg\sup_{t \in [0,1]} ta-\psi_{QU}(K-1,M,t).
\end{align*}
Using this result and substituting in~\eqref{necc.proof.bounded.2}:
\begin{align}
\mathbb{P}_{QU}\big(\sum_{\ell=1}^{K-1} \tilde{W}_{\ell} + \tilde{Z} \geq (K-1)\gamma \big) & \geq e^{-E_{QU}(a,K-1,M) - \delta} \nonumber \\ 
& \geq e^{-E_{QU}(a,K,M) - \delta} \label{necc.proof.bounded.4}
\end{align}
where~\eqref{necc.proof.bounded.4} holds because $t^{*} \in [0,1]$ and  $\log(\mathbb{E}_{Q}[e^{tL_{G}}]) = (t-1)D_{t}(P||Q) \leq 0$, where $D_{t}(V||U) \geq 0$ is the R\'{e}nyi-divergence between distributions $P$ and $Q$~\cite{polyanski}.  Moreover,
\begin{align}
 E_{QU}(a,K,M)
& = E_{QU}((K-1)\gamma + \delta,K,M) \nonumber \\
& \leq E_{QU}(K\gamma + \delta,K,M) \nonumber \\
&\leq E_{QU}(K\gamma,K,M) + \delta \label{new.1.1}
\end{align} 
where~\eqref{new.1.1} holds because $t \in[0,1]$ and, by~\eqref{necc.proof.bounded.3},
\[a \in \big[-KD(Q||P)-MD(U||V), KD(P||Q)+MD(V||U)\big]
\]
 Also, by Lemma~\ref{Lem.2}, for some positive constant $B$:
\begin{align}
E_{QU}(0,K-1,M) 
& \geq B((K-1)D(Q||P) + MD(U||V)) \nonumber \\
& \geq B'((K-1)D(P||Q) + MD(V||U)) \label{new.1.2}
\end{align}
where~\eqref{new.1.2} holds for some positive constant $B'$ because for bounded LLR $D(Q||P) \approx D(P||Q)$ and $D(U||V) \approx D(V||U)$. Thus, for sufficiently large $n$, and for some positive constant $B''$:
\begin{align}
\delta &= ((K-1)D(P||Q) + MD(V||U))^{\frac{2}{3}} \twocolbreak \leq (B''E_{QU}(0,K-1,M))^{\frac{2}{3}} \nonumber \\
& \leq (B''E_{QU}(K\gamma,K,M))^{\frac{2}{3}} \label{new.1.3}
\end{align}
Combining Equations~\eqref{new.1.1},~\eqref{new.1.2},~\eqref{new.1.3}: 
\begin{align}
\twocolAlignMarker E_{QU}(a,K,M) + \delta \twocolnewline
 &\leq E_{QU}(K\gamma,K,M) + 2\delta \nonumber \\
 &\leq E_{QU}(K\gamma,K,M) + 2(B''E_{QU}(K\gamma,K,M))^{\frac{2}{3}}
\end{align}
Substituting in~\eqref{necc.proof.bounded.4}:
\begin{align}
\mathbb{P}_{QU}\big(\sum_{\ell=1}^{K-1} \tilde{W}_{\ell} \twocolAlignMarker+ \tilde{Z} \geq (K-1)\gamma \big) \twocolnewline
&\geq e^{-(E_{QU}(K\gamma,K,M) + 2(B''E_{QU}(K\gamma,K,M))^{\frac{2}{3}})} \nonumber \\  
& \overset{(a)}{\geq} e^{-( (1-\epsilon)\log(n) + 2(B''(1-\epsilon)\log(n))^{\frac{2}{3}})} \nonumber\\  
& \geq e^{-(1-\epsilon)\log(n)(1 + o(1))} \label{necc.proof.bounded.5}
\end{align}
where $(a)$ comes from the contraposition assumption that~\eqref{Cond.1.The.3.new} does not hold, i.e., $E_{QU}(K\gamma,K,M) \leq (1-\epsilon)\log(n)$ for arbitrary small $\epsilon > 0$. Equation~\eqref{necc.proof.bounded.5} shows that
\[
n\mathbb{P}_{QU}(\sum_{\ell=1}^{K-1} \tilde{W}_{\ell} + \tilde{Z} \geq (K-1)\gamma) \geq n^{\epsilon(1+o(1))}
\]
which implies that~\eqref{necc.general.1.2} does not hold for $\theta_{n} = \gamma$. 

Similarly, we will show that~\eqref{necc.general.1.1} does not hold for $\theta_{n} = \gamma$. Define
\begin{align}
K_{o} &= \frac{K}{\log(K)} = o(K) \nonumber \\
\delta'  &= \frac{(K_{o}-1)(D(P||Q)-\gamma) + 6\sqrt{K_o}\sigma}{(K-K_{o})D(P||Q) + MD(V||U)}
\end{align}
Note that $\delta' = o(1)$, which holds because $K\gamma \leq KD(P||Q) + MD(V||U)$,  $K_{o} = o(K)$ and $ K_{o}\sigma^{2} = K_{o} \frac{d^{2}(\log(\mathbb{E}_{Q}[e^{tL_{G}}]))}{dt^{2}}|_{t=1} \leq BK_{o} D(P||Q)$ by Lemma~\ref{Lem.2} for some positive constant $B$. Let $a = (K-K_{o}) (\gamma -\delta' D(P||Q) - \frac{\delta'}{K-K_{o}} MD(V||U)) - \delta$, for some $\delta > 0$. Then, by a similar analysis as in~\eqref{necc.proof.bounded.1}:
\begin{align}
&\mathbb{P}_{PV}\big(\sum_{\ell=1}^{K-K_{o}} W_{\ell} + Z \leq (K-1)\gamma + \tilde{\theta}_{n}\big) \nonumber\\ 
& = \mathbb{P}_{PV}\bigg(\sum_{\ell=1}^{K-K_{o}} W_{\ell} + Z  \twocolbreak \includeonetwocol{}{\hspace{0.2in}} \leq
(K-K_{o}) (\gamma -\delta' D(P||Q) - \frac{\delta'}{K-K_{o}} MD(V||U))\bigg)\nonumber \\ 
& \overset{(a)}{\geq} e^{-\big(t^{*}a-\psi_{PV}(K-K_{o},M,t^{*})\big) - |t^{*}|\delta } \twocolbreak \includeonetwocol{}{\hspace{0.2in}\times}
\big(1 - \frac{(K-K_{o})\tilde{\sigma}^{2}_{L_{G}} + M\tilde{\sigma}^{2}_{L_{S}}}{\delta^{2}}\big) \nonumber\\ 
& \overset{(b)}{\geq} e^{-(t^{*}a-\psi_{PV}(K-K_{o},M,t^{*})) - |t^{*}|\delta }(1 - o(1)) \label{necc.proof.bounded.6}
\end{align}
where $(a)$ holds for $t^{*} = \arg\sup_{t \in \mathbb{R}} ta-\psi_{PV}(K-K_{o},M,t)$ and by defining two new probability distributions $\tilde{P}$ and $\tilde{V}$ over the same support of $P$ and $V$, respectively and $\tilde{\sigma}^{2}_{L_{G}}$, $\tilde{\mu}_{L_{G}}$, $\tilde{\sigma}^{2}_{L_{S}}$ and $\tilde{\mu}_{L_{S}}$ to be the variances and means of $L_{G}$ and $L_{S}$ under $\tilde{P}$ and $\tilde{V}$, respectively. $(b)$ holds by choosing 
\[
\delta = ((K-K_{o}) D(P||Q) +  MD(V||U))^{\frac{2}{3}}
\]
and noticing that for bounded LLR, 
\begin{align*}
(K-K_{o})\tilde{\sigma}^{2}_{L_{G}} + M\tilde{\sigma}^{2}_{L_{S}} &= \psi''(K-K_{o},M,t) \\
&\leq B((K-K_{o}) D(P||Q) +  MD(V||U)),
\end{align*}
by Lemma~\ref{Lem.2} for some positive constant $B$. 

Moreover, for sufficiently large $n$:
\begin{align}
a &= (K-K_{o}) (\gamma -\delta' D(P||Q) - \frac{\delta'}{K-K_{o}} MD(V||U)) - \delta \nonumber \\
& = (1-o(1)) (K\gamma - K\delta'D(P||Q) - \delta'MD(V||U)) -\delta  \nonumber\\ \nonumber
& \overset{(a)}{\leq} (KD(P||Q)+MD(V||U)) ( \frac{1}{1+\frac{\epsilon_{o}}{2}} - \delta' - o(1)) \nonumber\\ 
& \overset{(b)}{\leq} KD(P||Q)+MD(V||U) \label{necc.proof.bounded.7}
\end{align}
where $(a)$ holds from~\eqref{necc.holds} and $(b)$ holds because $(K-1)D(P||Q) + MD(V||U) \to \infty$ and $\delta' = o(1)$. Thus 
\[
a \in [-KD(Q||P)-MD(U||V), KD(P||Q)+MD(V||U)].
\]
By Lemma~\ref{Lem.1}, 
\begin{align*}
    t^{*} &= \arg\sup_{t \in \mathbb{R}} ta-\psi_{PV}(K-K_{o},M,t) \\
    &= \arg\sup_{t \in [-1,0]} ta-\psi_{PV}(K,M,t).
    \end{align*}
    Substituting in~\eqref{necc.proof.bounded.6}:
\begin{align} 
\mathbb{P}_{PV}\big(\sum_{\ell=1}^{K-K_{o}} W_{\ell} \twocolAlignMarker + Z \leq (K-1)\gamma + \tilde{\theta}_{n}\big)  \twocolbreak
\geq e^{-(E_{PV}(a,K,M) + \delta)} (1-o(1)) \label{necc.proof.bounded.8}
\end{align}
Moreover,
\begin{align}
& E_{PV}(a,K,M) \twocolbreak \leq E_{PV}(K\gamma,K,M) +  \delta'(KD(P||Q)+MD(V||U)) + \delta \label{new.2.x.1}
\end{align}
which holds because $t \in [-1,0]$ and $a \in [-KD(Q||P)-MD(U||V), KD(P||Q)+MD(V||U)]$ by~\eqref{necc.proof.bounded.7}. Also, by Lemma~\ref{Lem.2}, for some positive constant $B$
\begin{align}
E_{PV}(K\gamma,K,M) &\geq E_{PV}(\frac{KD(P||Q)+MD(V||U)}{1+\frac{\epsilon_{o}}{2}},K,M) \nonumber \\
& \geq B(KD(P||Q) + MD(V||U)) \label{new.2.x.2}
\end{align}
Thus, for sufficiently large $n$ and for some positive constant $B'$:
\begin{align} 
\delta & = (KD(P||Q) + MD(V||U))^{\frac{2}{3}} \nonumber \\
& \leq (B'E_{PV}(K\gamma,K,M))^{\frac{2}{3}} \label{new.2.x.3}
\end{align}
Combining equations~\eqref{new.2.x.1},~\eqref{new.2.x.2},~\eqref{new.2.x.3}:
\begin{align}
& E_{PV}(a,K,M) + \delta \twocolbreak \leq E_{PV}(K\gamma,K,M) +  \delta'(KD(P||Q)+MD(V||U)) + 2\delta \nonumber \\
& \leq E_{PV}(K\gamma,K,M) + \delta'B''E_{PV}(K\gamma,K,M) \twocolbreak \includeonetwocol{}{\hspace{0.2in}} + 2(B'E_{PV}(K\gamma,K,M))^{\frac{2}{3}}
\end{align}
for some positive constants $B'$ and $B''$. Substituting in~\eqref{necc.proof.bounded.8}:
\begin{align} 
& \mathbb{P}_{PV}\big(\sum_{\ell=1}^{K-K_{o}} W_{\ell} + Z \leq (K-1)\gamma + \tilde{\theta}_{n}\big)  \twocolbreak
\geq e^{-E_{PV}(K\gamma,K,M)( 1 + \delta'B'' + \frac{2(B')^{\frac{2}{3}}}{(E_{PV}(K\gamma,K,M))^{\frac{1}{3}} }) } (1-o(1)) \nonumber\\  
& \overset{(a)}{=} e^{-E_{PV}(K\gamma,K,M)( 1 + o(1)) } \nonumber \\  
& \overset{(b)}{=} e^{(K\gamma - E_{QU}(K\gamma,K,M)) ( 1 + o(1) )} \nonumber\\  
& \overset{(c)}{\geq} e^{(\epsilon\log(n) - \log(K))(1 + o(1))}  \nonumber\\ 
&\geq e^{(\epsilon\log(K) - \log(K))(1 + o(1))}  \nonumber\\ 
& = e^{-\log(K) (1-\epsilon+ o(1))} \label{necc.proof.bounded.9}
\end{align}
where $(a)$ holds because $\delta' = o(1)$, $(b)$ holds because from Lemma~\ref{Lem.1} $E_{PV}(K\gamma,K,M) = E_{QU}(K\gamma,K,M) - K\gamma$ and $(c)$ is due to the contraposition assumption that~\eqref{Cond.1.The.3.new} does not hold, i.e., $E_{QU}(K\gamma,K,M) \leq (1-\epsilon)\log(n)$ for arbitrary small $\epsilon > 0$.

Equation~\eqref{necc.proof.bounded.9} shows:
\[
K_{o}\mathbb{P}_{PV}(\sum_{\ell=1}^{K-K_{o}} W_{\ell} + Z \leq (K-1)\gamma + \tilde{\theta}_{n}) \geq K^{\epsilon(1+o(1))}
\]
which implies that~\eqref{necc.general.1.1} does not hold for $\theta_{n} = \gamma$. 

Thus, if \eqref{Cond.1.The.3.new} does not hold, both ~\eqref{necc.proof.bounded.5} and~\eqref{necc.proof.bounded.9} show that ~\eqref{necc.general.1.1} and~\eqref{necc.general.1.2} does not hold simultaneously at $\theta_{n} = \gamma$. Thus, for any $\theta_{n} > \gamma$,~\eqref{necc.general.1.1} will not hold and for any $\theta_{n} < \gamma$,~\eqref{necc.general.1.2} will not hold, and hence, if \eqref{Cond.1.The.3.new} does not hold, then there does not exist $\theta_{n}$ such that ~\eqref{necc.general.1.1} and~\eqref{necc.general.1.2} hold simultaneously. This concludes the proof.


\section{Necessity of Theorem~\ref{The.3.5}}
\label{App.8}

Recall that Definition~\ref{def.4} introduced Chernoff-information-type functions for the LLR of the graph plus side information; for convenience we now introduce a narrowed version of the same functions that focus on graph information only.

\begin{definition}
\label{def.5}
\begin{align}
\psi_{Q}(t,m_{1})  & \triangleq m_{1}\log(\mathbb{E}_{Q}[e^{tL_{G}}]) \\
\psi_{P}(t,m_{1})  & \triangleq m_{1}\log(\mathbb{E}_{P}[e^{tL_{G}}]) \\
E_{Q}(\theta,m_{1}) &\triangleq \sup_{t\in[0,1]}  t\theta - \psi_{Q}(t,m_{1}) \\ 
E_{P}(\theta,m_{1}) &\triangleq \sup_{t\in[-1,0]} t\theta - \psi_{P}(t,m_{1})
\end{align}
\end{definition}

The quantities introduced in Definition~\ref{def.4} reduce to Definition~\ref{def.5} by setting $m_2=0$, therefore Lemmas~\ref{Lem.1} and~\ref{Lem.2} continue to hold.

In view of Lemma~\ref{Lem.necc.exact}, it suffices to test whether there exists $\theta_{n}$ such that both~\eqref{necc.general.1.1} and~\eqref{necc.general.1.2} hold. We will show that if one of the conditions $(1)$-$(6)$ of Theorem~\ref{The.3.5} is not satisfied, then there does not exist $\theta_{n}$ such that ~\eqref{necc.general.1.1} and~\eqref{necc.general.1.2} hold simultaneously.

Let $\theta_{n} = \gamma = \frac{\log(\frac{n}{K})}{K}$, and $a = (K-1)\gamma - \sum_{m=1}^{M} \hfeat + \delta$ for $\delta = \log(n)^{\frac{2}{3}}$. 
\begin{align}
\mathbb{P}&\big(\sum_{\ell=1}^{K-1} \tilde{W}_{\ell} + \tilde{Z} \geq (K-1)\gamma \big)\nonumber \\
&= \sum_{\ell_{1}=1}^{L_{1}}\cdots\sum_{\ell_{M}=1}^{L_{M}} \bigg[ (\prod_{m=1}^{M}\alpha_{-,\msubK}^{m}) \twocolbreak \includeonetwocol{}{\hspace{0.2in}\times} 
 \mathbb{P}_{Q}\big(\sum_{\ell=1}^{K-1} \tilde{W}_{\ell}  \geq (K-1)\gamma - \sum_{m=1}^{M} \hfeat \big) \bigg] \nonumber \\ 
&\overset{(a)}{\geq} \sum_{\ell_{1}=1}^{L_{1}}\cdots\sum_{\ell_{M}=1}^{L_{M}} \bigg[ (\prod_{m=1}^{M}\alpha_{-,\msubK}^{m})  \twocolbreak \includeonetwocol{}{\hspace{0.2in}\times}
e^{-(t^{*}a - (K-1)\log(\mathbb{E}_{Q}[e^{t^{*}L_{G}}]) ) - |t^{*}|\delta} (1 - o(1)) \bigg] \label{necc.proof.M_ary.1}
\end{align}
where $(a)$ holds by Lemma~\ref{Le.10}, where $t^{*} = \arg\sup_{t \in \mathbb{R}} (ta - (K-1)\log(\mathbb{E}_{Q}[e^{tL_{G}}]))$\footnote{For ease of notation, we omit any subscript for both $a$ and $t^{*}$. However, both depend on the outcomes of the features as shown in their definitions.}. 

Under~\eqref{reg.1}:
\begin{align}
KD(Q||P) &= \rho(a-b-bT)(1+o(1))\log(n) \nonumber\\
KD(P||Q) & = \rho(aT+b-a)(1+o(1))\log(n) \nonumber
\end{align}
Thus, according to conditions of Theorem~\ref{The.3.5}, 
\[
a \in \big[-KD(Q||P), KD(P||Q)\big].
\]
So, by Lemma~\ref{Lem.1}, 
\begin{align*}
    t^{*} &= \arg\sup_{t \in \mathbb{R}} (ta - (K-1)\log(\mathbb{E}_{Q}[e^{tL_{G}}])) \\
    & = \arg\sup_{t \in [0,1]} (ta - (K-1)\log(\mathbb{E}_{Q}[e^{tL_{G}}]))
\end{align*}

Without loss of generality, we focus on one term of the nested sum in~\eqref{necc.proof.M_ary.1}. Then,

\begin{itemize}
\itemsep 5pt
\item If $\sum_{m=1}^{M} \hfeat = o(\log(n))$ and both $\sum_{m=1}^{M} \log(\alpha_{+,\msubK}^{m})$ and $\sum_{m=1}^{M} \log(\alpha_{-,\msubK}^{m})$ are $ o(\log(n))$, then by evaluating the supremum and by substituting in~\eqref{necc.proof.M_ary.1},
\begin{align}
\mathbb{P}\big(\sum_{\ell=1}^{K-1} \tilde{W}_{\ell} + \tilde{Z} \geq (K-1)\gamma \big) & \geq n^{-\eta_{1}(\rho,a,b) + o(1)} \nonumber
\end{align}
Thus, if $\eta_{1}(\rho,a,b) \leq 1 - \varepsilon$ for some $0 < \varepsilon < 1$, then $(n-K)\mathbb{P}(\sum_{\ell=1}^{K-1} \tilde{W}_{\ell} + \tilde{Z} \geq (K-1)\gamma) \geq n^{\varepsilon +o(1)}$ which shows that~\eqref{necc.general.1.2} does not hold for $\theta_{n} = \gamma$.

\item If $\sum_{m=1}^{M} \hfeat = o(\log(n))$ and $\sum_{m=1}^{M} \log(\alpha_{+,\msubK}^{m}) = \sum_{m=1}^{M} \log(\alpha_{-,\msubK}^{m})=-\beta\log(n)+ o(\log(n)), \beta > 0$, then by evaluating the supremum and by substituting in~\eqref{necc.proof.M_ary.1},
\begin{align}
\mathbb{P}\big(\sum_{\ell=1}^{K-1} \tilde{W}_{\ell} + \tilde{Z} \geq (K-1)\gamma\big) &\geq n^{-\eta_{1}(\rho,a,b)-\beta + o(1)}\nonumber
\end{align}
Thus, if $\eta_{1}(\rho,a,b)+\beta \leq 1 - \varepsilon$ for some $0 < \varepsilon < 1$, then $(n-K)\mathbb{P}(\sum_{\ell=1}^{K-1} \tilde{W}_{\ell} + \tilde{Z} \geq (K-1)\gamma) \geq n^{\varepsilon +o(1)}$ which shows that~\eqref{necc.general.1.2} does not hold for $\theta_{n} = \gamma$.

\item If $\sum_{m=1}^{M} \hfeat = \beta\log(n)+ o(\log(n)), 0 < \beta < \rho(a-b-bT), \sum_{m=1}^{M}\log(\alpha_{+,\msubK}^{m}) = o(\log(n))$, then by evaluating the supremum and by substituting in~\eqref{necc.proof.M_ary.1},
\begin{align}
\mathbb{P}\big(\sum_{\ell=1}^{K-1} \tilde{W}_{\ell} + \tilde{Z} \geq (K-1)\gamma\big) &\geq n^{-\eta_{2}(\rho,a,b,\beta) + o(1)}\nonumber
\end{align}
Thus, if $\eta_{2}(\rho,a,b,\beta)  \leq 1 - \varepsilon$ for some $0 < \varepsilon < 1$, then $(n-K)\mathbb{P}(\sum_{\ell=1}^{K-1} \tilde{W}_{\ell} + \tilde{Z} \geq (K-1)\gamma) \geq n^{\varepsilon +o(1)}$ which shows that~\eqref{necc.general.1.2} does not hold for $\theta_{n} = \gamma$.

\item If $\sum_{m=1}^{M} \hfeat = -\beta\log(n)+ o(\log(n)), 0 < \beta < \rho(a-b-bT), \sum_{m=1}^{M}\log(\alpha_{-,\msubK}^{m}) = o(\log(n))$, then by evaluating the supremum and by substituting in~\eqref{necc.proof.M_ary.1},
\begin{align}
\mathbb{P}\big(\sum_{\ell=1}^{K-1} \tilde{W}_{\ell} + \tilde{Z} \geq (K-1)\gamma\big) &\geq n^{-\eta_{3}(\rho,a,b,\beta) + o(1)}\nonumber
\end{align}
Thus, if $\eta_{3}(\rho,a,b,\beta)  \leq 1 - \varepsilon$ for some $0 < \varepsilon < 1$, then $(n-K)\mathbb{P}(\sum_{\ell=1}^{K-1} \tilde{W}_{\ell} + \tilde{Z} \geq (K-1)\gamma) \geq n^{\varepsilon +o(1)}$ which shows that~\eqref{necc.general.1.2} does not hold for $\theta_{n} = \gamma$.

\item If $\sum_{m=1}^{M} \hfeat = \beta\log(n)+ o(\log(n)), 0 < \beta < \rho(a-b-bT), \sum_{m=1}^{M}\log(\alpha_{+,\msubK}^{m}) = -\beta'\log(n) + o(\log(n))$, then by evaluating the supremum and by substituting in~\eqref{necc.proof.M_ary.1},:
\begin{align}
\mathbb{P}\big(\sum_{\ell=1}^{K-1} \tilde{W}_{\ell} + \tilde{Z} \geq (K-1)\gamma\big) &\geq n^{-\eta_{2}(\rho,a,b,\beta) -\beta' + o(1)}\nonumber
\end{align}
Thus, if $\eta_{2}(\rho,a,b,\beta) + \beta'  \leq 1 - \varepsilon$ for some $0 < \varepsilon < 1$, then $(n-K)\mathbb{P}(\sum_{\ell=1}^{K-1} \tilde{W}_{\ell} + \tilde{Z} \geq (K-1)\gamma) \geq n^{\varepsilon +o(1)}$ which shows that~\eqref{necc.general.1.2} does not hold for $\theta_{n} = \gamma$.

\item If $\sum_{m=1}^{M} \hfeat = -\beta\log(n)+ o(\log(n)), 0 < \beta < \rho(a-b-bT), \sum_{m=1}^{M}\log(\alpha_{-,\msubK}^{m}) = -\beta'\log(n) + o(\log(n))$, then by evaluating the supremum and by substituting in~\eqref{necc.proof.M_ary.1},
\begin{align}
\mathbb{P}\big(\sum_{\ell=1}^{K-1} \tilde{W}_{\ell} + \tilde{Z} \geq (K-1)\gamma\big) &\geq n^{-\eta_{3}(\rho,a,b,\beta) -\beta' + o(1)}\nonumber
\end{align}
Thus, if $\eta_{3}(\rho,a,b,\beta) + \beta'  \leq 1 - \varepsilon$ for some $0 < \varepsilon < 1$, then $(n-K)\mathbb{P}(\sum_{\ell=1}^{K-1} \tilde{W}_{\ell} + \tilde{Z} \geq (K-1)\gamma) \geq n^{\varepsilon +o(1)}$ which shows that~\eqref{necc.general.1.2} does not hold for $\theta_{n} = \gamma$.
\end{itemize}

Now we show that~\eqref{necc.general.1.1} does not hold for $\theta_{n} = \gamma$. Let $K_{o} = \frac{K}{\log(K)} = o(K)$. Also, let $a = (K-1)\gamma + \tilde{\theta}_{n} - \sum_{m=1}^{M} \hfeat - \delta$ for $\delta = \log(n)^{\frac{2}{3}}$. Then,
\begin{align}
& \mathbb{P}\big(\sum_{\ell=1}^{K-K_{o}} W_{\ell} + Z \leq (K-1)\gamma + \tilde{\theta}_{n}\big) \nonumber \\
& \overset{(a)}{\geq} \sum_{\ell_{1}=1}^{L_{1}}\cdots\sum_{\ell_{L}=1}^{L_M} (\prod_{m=1}^{M}\alpha_{+,\msubK}^{m})  \twocolbreak \includeonetwocol{}{\hspace{0.2in}\times}
e^{-(t^{*}a - (K-K_{o})\log(\mathbb{E}_{P}[e^{t^{*}L_{G}}]) ) - |t^{*}|\delta} (1-o(1)) \nonumber \\
& \overset{(b)}{=} \sum_{\ell_{1}=1}^{L_{1}}\cdots\sum_{\ell_{L}=1}^{L_M} (\prod_{m=1}^{M}\alpha_{+,\msubK}^{m})  \twocolbreak \includeonetwocol{}{\hspace{0.2in}\times}
  e^{-(\lambda^{*}a - (K-K_{o})\log(\mathbb{E}_{Q}[e^{\lambda ^{*}L_{G}}]) )  + a - |\lambda^{*}-1|\delta } (1-o(1)) \label{necc.proof.M_ary.2}
\end{align}
where $(a)$ holds by Lemma~\ref{Le.10}, where $t^{*} = \arg\sup_{t \in \mathbb{R}} (ta - (K-K_{o})\log(\mathbb{E}_{P}[e^{tL_{G}}]) )$ and $(b)$ holds for $\lambda^{*} = 1 + t^{*}$ and by Lemma~\ref{Lem.1}.

Thus, according to conditions of Theorem~\ref{The.3.5}, 
\begin{equation}
a \in \big[-KD(Q||P), KD(P||Q)\big]. \label{key_step}
\end{equation}
Thus, by Lemma~\ref{Lem.1}, $\arg\sup_{t \in \mathbb{R}}$ is replaced by $\arg\sup_{t \in [-1,0]}$. 

\begin{itemize}
\itemsep 5pt
\item If $\sum_{m=1}^{M} \hfeat = o(\log(n))$ and both $\sum_{m=1}^{M} \log(\alpha_{+,\msubK}^{m})$ and $\sum_{m=1}^{M} \log(\alpha_{-,\msubK}^{m})$ are $ o(\log(n))$, then by evaluating the supremum and by substituting in~\eqref{necc.proof.M_ary.2},
\begin{align}
\mathbb{P}\big(\sum_{\ell=1}^{K-K_{o}} W_{\ell} + Z \leq (K-1)\gamma + \tilde{\theta}_{n}\big) & \geq n^{-\eta_{1}(\rho,a,b) + o(1)} \nonumber
\end{align}
Thus, if $\eta_{1}(\rho,a,b) \leq 1 - \varepsilon$ for some $0 < \varepsilon < 1$, then $K\mathbb{P}(\sum_{\ell=1}^{K-K_{o}} W_{\ell} + Z \leq (K-1)\gamma + \tilde{\theta}_{n}) \geq n^{\varepsilon +o(1)}$ which shows that~\eqref{necc.general.1.1} does not hold for $\theta_{n} = \gamma$.

\item If $\sum_{m=1}^{M} \hfeat = o(\log(n))$ and $\sum_{m=1}^{M} \log(\alpha_{+,\msubK}^{m}) = \sum_{m=1}^{M} \log(\alpha_{-,\msubK}^{m})=-\beta\log(n)+ o(\log(n)), \beta > 0$, then by evaluating the supremum and by substituting in~\eqref{necc.proof.M_ary.2},
\begin{align}
\mathbb{P}\big(\sum_{\ell=1}^{K-K_{o}} W_{\ell} + Z \leq (K-1)\gamma + \tilde{\theta}_{n}\big) &\geq n^{-\eta_{1}(\rho,a,b)-\beta + o(1)} \nonumber
\end{align}
Thus, if $\eta_{1}(\rho,a,b)+\beta \leq 1 - \varepsilon$ for some $0 < \varepsilon < 1$, then $K\mathbb{P}(\sum_{\ell=1}^{K-K_{o}} W_{\ell} + Z \leq (K-1)\gamma + \tilde{\theta}_{n}) \geq n^{\varepsilon +o(1)}$ which shows that~\eqref{necc.general.1.1} does not hold for $\theta_{n} = \gamma$.

\item If $\sum_{m=1}^{M} \hfeat = \beta\log(n)+ o(\log(n)), 0 < \beta < \rho(a-b-bT), \sum_{m=1}^{M}\log(\alpha_{+,\msubK}^{m}) = o(\log(n))$, then by evaluating the supremum and by substituting in~\eqref{necc.proof.M_ary.2},
\begin{align}
\mathbb{P}\big(\sum_{\ell=1}^{K-K_{o}} W_{\ell} + Z \leq (K-1)\gamma + \tilde{\theta}_{n}\big) &\geq n^{-\eta_{2}(\rho,a,b,\beta) + o(1)}\nonumber
\end{align}
Thus, if $\eta_{2}(\rho,a,b,\beta)  \leq 1 - \varepsilon$ for some $0 < \varepsilon < 1$, then $K\mathbb{P}(\sum_{\ell=1}^{K-K_{o}} W_{\ell} + Z \leq (K-1)\gamma + \tilde{\theta}_{n}) \geq n^{\varepsilon +o(1)}$ which shows that~\eqref{necc.general.1.1} does not hold for $\theta_{n} = \gamma$.

\item If $\sum_{m=1}^{M} \hfeat = -\beta\log(n)+ o(\log(n)), 0 < \beta < \rho(a-b-bT), \sum_{m=1}^{M}\log(\alpha_{-,\msubK}^{m}) = o(\log(n))$, then by evaluating the supremum and by substituting in~\eqref{necc.proof.M_ary.2},
\begin{align}
\mathbb{P}\big(\sum_{\ell=1}^{K-K_{o}} W_{\ell} + Z \leq (K-1)\gamma + \tilde{\theta}_{n}\big) &\geq n^{-\eta_{3}(\rho,a,b,\beta) + o(1)}\nonumber
\end{align}
Thus, if $\eta_{3}(\rho,a,b,\beta)  \leq 1 - \varepsilon$ for some $0 < \varepsilon < 1$, then $K\mathbb{P}(\sum_{\ell=1}^{K-K_{o}} W_{\ell} + Z \leq (K-1)\gamma + \tilde{\theta}_{n}) \geq n^{\varepsilon +o(1)}$ which shows that~\eqref{necc.general.1.1} does not hold for $\theta_{n} = \gamma$.

\item If $\sum_{m=1}^{M} \hfeat = \beta\log(n)+ o(\log(n)), 0 < \beta < \rho(a-b-bT), \sum_{m=1}^{M}\log(\alpha_{+,\msubK}^{m}) = -\beta'\log(n) + o(\log(n))$, then by evaluating the supremum and by substituting in~\eqref{necc.proof.M_ary.2},
\begin{align}
\mathbb{P}\big(\sum_{\ell=1}^{K-K_{o}} W_{\ell} + Z \leq (K-1)\gamma + \tilde{\theta}_{n}\big) &\geq n^{-\eta_{2}(\rho,a,b,\beta) -\beta' + o(1)} \nonumber
\end{align}
Thus, if $\eta_{2}(\rho,a,b,\beta) + \beta'  \leq 1 - \varepsilon$ for some $0 < \varepsilon < 1$, then $K\mathbb{P}(\sum_{\ell=1}^{K-K_{o}} W_{\ell} + Z \leq (K-1)\gamma + \tilde{\theta}_{n}) \geq n^{\varepsilon +o(1)}$ which shows that~\eqref{necc.general.1.1} does not hold for $\theta_{n} = \gamma$.

\item If $\sum_{m=1}^{M} \hfeat = -\beta\log(n)+ o(\log(n)), 0 < \beta < \rho(a-b-bT), \sum_{m=1}^{M}\log(\alpha_{-,\msubK}^{m}) = -\beta'\log(n) + o(\log(n))$, then by evaluating the supremum and by substituting in~\eqref{necc.proof.M_ary.2},
\begin{align}
\mathbb{P}\big(\sum_{\ell=1}^{K-K_{o}} W_{\ell} + Z \leq (K-1)\gamma + \tilde{\theta}_{n}\big) &\geq n^{-\eta_{3}(\rho,a,b,\beta) -\beta' + o(1)}\nonumber
\end{align}
Thus, if $\eta_{3}(\rho,a,b,\beta) + \beta'  \leq 1 - \varepsilon$ for some $0 < \varepsilon < 1$, then $K\mathbb{P}(\sum_{\ell=1}^{K-K_{o}} W_{\ell} + Z \leq (K-1)\gamma + \tilde{\theta}_{n}) \geq n^{\varepsilon +o(1)}$ which shows that~\eqref{necc.general.1.1} does not hold for $\theta_{n} = \gamma$.
\end{itemize}

To summarize, when $\theta_{n} = \gamma$, if one of the conditions $(1)$-$(6)$ of Theorem~\ref{The.3.5} does not hold, then~\eqref{necc.general.1.1} and~\eqref{necc.general.1.2} cannot hold simultaneously. Thus, for any $\theta_{n} > \gamma$,~\eqref{necc.general.1.1} will not hold and for any $\theta_{n} < \gamma$,~\eqref{necc.general.1.2} will not hold, and hence, if one of the conditions $(1)$-$(6)$ of Theorem~\ref{The.3.5} does not hold, then there does not exist $\theta_{n}$ such that ~\eqref{necc.general.1.1} and~\eqref{necc.general.1.2} hold simultaneously. This concludes the proof of the necessary conditions.

Finally, we comment on how the proof would change if instead of the regime~\eqref{reg.1}, $K$ was chosen such that for all large $n$, $\log(\frac{n}{K}) = (C-o(1)) \log(n)$ for some constant $C \in (0,1]$. A key step in the proof was to ensure that $\theta$ in definition~\ref{def.5} is between $[-KD(Q||P),\; KD(P||Q)]$, e.g, see~\eqref{key_step}. Hence, the only modification needed is to take $C$ into account. For example, when $\sum_{m=1}^{M} \hfeat = \beta\log(n)+ o(\log(n))$ for some positive $\beta$, then a condition on $\beta$ would be $-\rho(a-b-bT) < C \pm \beta < \rho(a-b-bT) $. The proofs for the modified regime would then follow a similar strategy as the proofs in this section. Similar modifications are needed for  the sufficiency proofs as well.


\section{Sufficiency of Theorem~\ref{The.3.5}}
\label{App.10}

The sufficient conditions are derived via Algorithm~\ref{Alg.1} provided in Section~\ref{exact.lim.1} with only one modification in the weak recovery step. Since the LLRs of the side information may not be bounded, the maximum likelihood detector with side information presented in Lemma~\ref{Suff.Random} cannot be used for the weak recovery step. Instead the maximum likelihood detector without side information provided in~\cite{infor_limits} will be used. 

The following lemma gives sufficient conditions for Algorithm~\ref{Alg.1} to achieve exact recovery.
\begin{lemma}
\label{The.3.6}
Define $C^*_k = C^* \cap {S_k}^c$ and assume $\hat{C}_k$ achieves weak recovery, i.e.
\begin{equation}
\label{M_ary.suff.eq.1}
\mathbb{P}\big(|\hat{C}_{k} \triangle C^{*}_{k}| \leq \delta K \text{ for } 1\leq k \leq \frac{1}{\delta}\big) \to 1
\end{equation}
Under conditions~\eqref{reg.1}, if conditions $(1)$-$(6)$ of Theorem~\ref{The.3.5} hold, then $\mathbb{P}(\tilde{C} = C^{*}) \to 1$.
\end{lemma}

\begin{proof}
Please see Appendix~\ref{App.9}
\end{proof}

In view of Lemma~\ref{The.3.6}, it suffices to show that there exists an estimator that achieves weak recovery for a random cluster size and satisfies~\eqref{M_ary.suff.eq.1}. We use the estimator presented in~\cite[Lemma 4]{infor_limits}, where it was shown that the maximum likelihood estimator can achieve weak recovery for a random cluster size upon observing only the graph if:
\begin{align}
& KD(P||Q) \to \infty \label{suff.M.1} \\
& \liminf_{n\to\infty} (K-1)D(P||Q) \geq 2\log(\frac{n}{K}) \label{suff.M.2} \\
& \mathbb{P}\bigg( \bigg| |C^{*}_{k}| - (1-\delta)K \bigg| \geq \frac{K}{\log(K)} \bigg) \leq o(1) \label{suff.M.3}
\end{align}
It is obvious that in the regime~\eqref{reg.1}, both~\eqref{suff.M.1} and~\eqref{suff.M.2} are satisfied. Thus, it remains to show that~\eqref{suff.M.3} holds too. Let $\hat{C}_{k}$ be the ML estimator for $C^{*}_{k}$ based on observing $\boldsymbol{G}_{k}$ defined in Algorithm~\ref{Alg.1}. The distribution of $|C^{*}_{k}|$ is obtained by sampling the indices of the original graph without replacement. Hence, for any convex function $\phi$: $\mathbb{E}[\phi(|C^{*}_{k}|)] \leq \mathbb{E}[\phi(Z)]$, where $Z$ is a binomial random variable $\text{Bin}(n(1-\delta),\frac{K}{n})$. Therefore, the Chernoff bound for $Z$ also holds for $|C^{*}_{k}|$. Thus,
\begin{align}
& \mathbb{P}\Big( \Big| |C^{*}_{k}| - (1-\delta)K \Big| \geq \frac{K}{\log(K)} \Big) \overset{}{\leq} o(1)\label{Suff.exact.eq.1}
\end{align}
Thus,~\eqref{suff.M.3} holds, which implies that ML achieves weak recovery with $K$ replaced with $\lceil (1-\delta)K \rceil$ in~\cite[Lemma 4]{infor_limits}. Thus, from~\cite[Lemma 4]{infor_limits}, for any $1 \leq k \leq \frac{1}{\delta}$:
\begin{equation}
\mathbb{P}\Big(\frac{ |\hat{C}_{k}\triangle C^{*}_{k} |}{K}  \leq 2\epsilon + \frac{1}{\log(K)}\Big) \geq 1 - o(1)
\end{equation}
with $\epsilon = o(1)$. Since $\delta$ is constant,  by the union bound over all $1 \leq k \leq \frac{1}{\delta}$, we have:
\begin{equation}
\mathbb{P}\Big(\frac{ |\hat{C}_{k}\triangle C^{*}_{k} |}{K}  \leq 2\epsilon + \frac{1}{\log(K)} \quad\forall 1 \leq k \leq \frac{1}{\delta}\Big) \geq 1 - o(1)
\end{equation}
Since $\epsilon = o(1)$, the desired~\eqref{M_ary.suff.eq.1} holds.


\section{Proof of Lemma~\ref{The.3.6}}
\label{App.9}

To prove Lemma~\ref{The.3.6}, we follow essentially the same strategy used for Lemma~\ref{The.5.new} in Appendix~\ref{App.5}. Namely, we intend to show that the total LLR for nodes inside and outside the community are, asymptotically, stochastically dominated by a certain constant. Since the strategy is essentially similar to an earlier result, we only provide a sketch in this appendix.

\begin{lemma}
\label{Lem.suff.M_ary}
In the regime~\eqref{reg.1}, suppose conditions $(1)$-$(6)$ of Theorem~\ref{The.3.5} hold. Let $\{W_{\ell}\}$ and $\{\tilde{W}_{\ell}\}$ denote two sequences of i.i.d copies of $L_{G}$ under $P$ and $Q$, respectively. Also, let $Z$ be a random variable whose distribution is identical to $\sum_{m=1}^{M} \hnod$ conditioned on $i\in C^{*}$, and $\tilde{Z}$ drawn according to the same distribution conditioned on $i\notin C^{*}$. Then, for sufficiently small constant $\delta$ and $\gamma = \frac{\log(\frac{n}{K})}{K}$:
\begin{align} 
\mathbb{P}\Big( \sum_{\ell=1}^{K(1-\delta)}  \tilde{W}_{\ell} + \tilde{Z}  \geq K(1-\delta)\gamma \Big) = o(\frac{1}{n}) \label{suff.M_ary.eq.2} \\
\mathbb{P}\Big( \sum_{\ell=1}^{K(1-2\delta)} W_{\ell} +\sum_{\ell=1}^{\delta K}  \tilde{W}_{\ell} + Z \leq K(1-\delta)\gamma \Big) = o(\frac{1}{K}) \label{suff.M_ary.eq.3} 
\end{align}
\end{lemma}

\begin{proof}
Using the Chernoff bound:
\begin{align}
&\mathbb{P}\Big( \sum_{\ell=1}^{K(1-\delta)} \tilde{W}_{\ell} + \tilde{Z}  \geq K(1-\delta)\gamma \Big) \nonumber \\
& \leq \mathbb{P}\Big( \sum_{\ell=1}^{K}  \tilde{W}_{\ell} + \tilde{Z} \geq K(1-\delta)\gamma \Big) \nonumber \\
& \leq \sum_{\ell_{1}=1}^{L_{1}}\cdots\sum_{\ell_{M}=1}^{L_{M}} \Big(\prod_{m=1}^{M}\alpha_{-,\msubK}^{m}\Big)  \twocolbreak \includeonetwocol{}{\hspace{0.2in}\times}
 e^{-\sup_{t \geq 0} t(K(1-\delta)\gamma - \sum_{m=1}^{M} \hfeat) - K\log(\mathbb{E}_{Q}[e^{tL_{G}}]) } \label{suff.proof.M_ary.4}
\end{align}
The terms inside the nested sum in~\eqref{suff.proof.M_ary.4} are upper bounded by:
\begin{itemize}
\itemsep 5pt
\item  $n^{-\eta_{1}(\rho,a,b) + o(1)}$, if $\sum_{m=1}^{M} \hfeat = o(\log(n))$ and both $\sum_{m=1}^{M} \log(\alpha_{+,\msubK}^{m})$ and $\sum_{m=1}^{M} \log(\alpha_{-,\msubK}^{m})$ are $ o(\log(n))$.

\item $n^{-\eta_{1}(\rho,a,b) - \beta + o(1)}$, if $\sum_{m=1}^{M} \hfeat = o(\log(n))$ and $\sum_{m=1}^{M} \log(\alpha_{+,\msubK}^{m}) = \sum_{m=1}^{M} \log(\alpha_{-,\msubK}^{m})=-\beta\log(n)+ o(\log(n)), \beta > 0$.

\item $n^{-\eta_{2}(\rho,a,b,\beta) + o(1)}$, if $\sum_{m=1}^{M} \hfeat = \beta\log(n)+ o(\log(n)), 0 < \beta < \rho(a-b-bT), \sum_{m=1}^{M}\log(\alpha_{+,\msubK}^{m}) = o(\log(n))$.

\item $n^{-\eta_{3}(\rho,a,b,\beta) + o(1)}$, if $\sum_{m=1}^{M} \hfeat = -\beta\log(n)+ o(\log(n)), 0 < \beta < \rho(a-b-bT), \sum_{m=1}^{M}\log(\alpha_{-,\msubK}^{m}) = o(\log(n))$.

\item $n^{-\eta_{2}(\rho,a,b,\beta) - \beta' + o(1)}$, if $\sum_{m=1}^{M} \hfeat = \beta\log(n)+ o(\log(n)), 0 < \beta < \rho(a-b-bT), \sum_{m=1}^{M}\log(\alpha_{+,\msubK}^{m}) = -\beta'\log(n) + o(\log(n))$.

\item $n^{-\eta_{3}(\rho,a,b,\beta) - \beta' + o(1)}$, if $\sum_{m=1}^{M} \hfeat = -\beta\log(n)+ o(\log(n)), 0 < \beta < \rho(a-b-bT), \sum_{m=1}^{M}\log(\alpha_{-,\msubK}^{m}) = -\beta'\log(n) + o(\log(n))$.
\end{itemize}
Since $M$ and $L_m$ are independent of $n$ and finite, it follows that if items $(1)$-$(6)$ of Theorem~\ref{The.3.5} are satisfied, then Equation~\eqref{suff.M_ary.eq.2} holds. 

To show~\eqref{suff.M_ary.eq.3}, Chernoff bound is used.
\begin{align}
& \mathbb{P}\Big( \sum_{\ell=1}^{K(1-2\delta)} W_{\ell} +\sum_{\ell=1}^{\delta K} \tilde{W}_{\ell} + Z \leq K(1-\delta)\gamma \Big) \nonumber \\ 
& \leq \sum_{\ell_{1}=1}^{L_{1}}\cdots\sum_{\ell_{M}=1}^{L_{M}} (\prod_{m=1}^{M}\alpha_{+,\msubK}^{m})  \twocolbreak \includeonetwocol{}{\hspace{0.2in}\times}
 e^{t(K(1-2\delta)\gamma- \sum_{m=1}^{M} \hfeat)+ K(1-2\delta) \log(\mathbb{E}_{P}[e^{-tL_{G}}])}  \twocolbreak \includeonetwocol{}{\hspace{0.2in}\times}
e^{K\delta\log(\mathbb{E}_{Q}[e^{-tL_{G}}]) + tK\gamma\delta} \label{suff.proof.M_ary.5}
\end{align}
Without loss of generality, we focus on one term inside the nested sum in\eqref{suff.proof.M_ary.5}:

\begin{itemize}
\itemsep 5pt
\item  If $\sum_{m=1}^{M} \hfeat = o(\log(n))$ and both $\sum_{m=1}^{M} \log(\alpha_{+,\msubK}^{m})$ and $\sum_{m=1}^{M} \log(\alpha_{-,\msubK}^{m})$ are $ o(\log(n))$, then:
\begin{align}
& (\prod_{m=1}^{M}\alpha_{+,\msubK}^{m})  e^{t(K(1-2\delta)\gamma- \sum_{m=1}^{M} \hfeat)+ K(1-2\delta) \log(\mathbb{E}_{P}[e^{-tL_{G}}])} \twocolbreak \includeonetwocol{}{\hspace{0.2in}\times}
e^{ K\delta\log(\mathbb{E}_{Q}[e^{-tL_{G}}]) + tK\gamma\delta} \nonumber \\
&\leq (\prod_{m=1}^{M}\alpha_{+,\msubK}^{m})  e^{(1-2\delta) \big( t(k\gamma -\frac{\sum_{m=1}^{M} \hfeat}{1-2\delta}) + K \log(\mathbb{E}_{P}[e^{-tL_{G}}]) \big)} \twocolbreak \includeonetwocol{}{\hspace{0.2in}\times}
e^{ \delta \big( tK\gamma +  K\log(\mathbb{E}_{Q}[e^{-tL_{G}}])\big)}\label{suff.proof.M_ary.6}
\end{align}
Since $\sum_{m=1}^{M} \hfeat = o(\log(n))$, it is easy to show that \
\[
K\gamma-\frac{\sum_{m=1}^{M} \hfeat}{1-2\delta} \in [-KD(Q||P)\;,\;KD(P||Q)].
\]
Define $\theta \triangleq K\gamma-\frac{\sum_{m=1}^{M} \hfeat}{1-2\delta}$ and choose $t^{*}\in [0,1]$, such that $t^{*}\theta + K\log(\mathbb{E}[e^{-t^{*}L_{G}}]) = -E_{P}(\theta,K)$. Substituting in~\eqref{suff.proof.M_ary.6}:
\begin{align}
& (\prod_{m=1}^{M}\alpha_{+,\msubK}^{m})  e^{t(K(1-2\delta)\gamma- \sum_{m=1}^{M} \hfeat)+ K(1-2\delta) \log(\mathbb{E}_{P}[e^{-tL_{G}}])}  \twocolbreak \includeonetwocol{}{\hspace{0.2in}\times}
e^{K\delta\log(\mathbb{E}_{Q}[e^{-tL_{G}}]) + tK\gamma\delta} \nonumber \\ 
& \leq (\prod_{m=1}^{M}\alpha_{+,\msubK}^{m}) e^{-(1-2\delta)E_{P}(\theta,K) + \delta \big( t^{*}K\gamma +  K\log(\mathbb{E}_{Q}[e^{-t^{*}L_{G}}])\big)}  \nonumber \\
&\leq  (\prod_{m=1}^{M}\alpha_{+,\msubK}^{m}) e^{-(1-2\delta)E_{P}(\theta,K) + \delta \big( K\gamma +  K\log(\mathbb{E}_{Q}[e^{-t^{*}L_{G}}])\big)} \label{suff.proof.M_ary.7}
\end{align}
where the last inequality holds because $t^{*} \in [0,1]$. Also, by Lemma~\ref{Lem.2} and convexity of $\log(\mathbb{E}_{Q}[e^{-tL_{G}}])$, the following holds for some positive constant $A$: 
\begin{equation}
K \log(\mathbb{E}_{Q}[e^{-t^{*}L_{G}}]) \leq K \log(\mathbb{E}_{Q}[e^{-L_{G}}]) \leq A KD(Q||P) \label{suff.proof.M_ary.8}
\end{equation}
Moreover, by Lemma~\ref{Lem.2}, $E_{P}[\theta,K] = E_{Q}[\theta,K] - \theta$ and $E_{Q}[\theta,K] \geq E_{Q}[0,K] \geq A_{1}KD(Q||P)$. Combining the last observation with~\eqref{suff.proof.M_ary.8}, for some positive constant $A_{2}$,
\begin{align}
& (\prod_{m=1}^{M}\alpha_{+,\msubK}^{m})  e^{t(K(1-2\delta)\gamma- \sum_{m=1}^{M} \hfeat)+ K(1-2\delta) \log(\mathbb{E}_{P}[e^{-tL_{G}}])} \twocolbreak \includeonetwocol{}{\hspace{0.2in}\times}
e^{K\delta\log(\mathbb{E}_{Q}[e^{-tL_{G}}]) + tK\gamma\delta} \nonumber \\ 
& \leq (\prod_{m=1}^{M}\alpha_{+,\msubK}^{m}) e^{-(1-2\delta)(E_{Q}(\theta,K)-\theta) + \delta K\gamma +  \delta A_{2}E_{Q}(\theta,K)} \nonumber \\
& = (\prod_{m=1}^{M}\alpha_{+,\msubK}^{m}) e^{-E_{Q}(\theta,K)(1-2\delta -\delta A_{2}) +(1-2\delta)\theta + \delta K\gamma} \label{suff.proof.M_ary.9}
\end{align}
Since $\sum_{m=1}^{M} \log(\alpha_{+,\msubK}^{m}) = o(\log(n))$, evaluating the supremum in $E_{Q}[\theta,K]$ and substituting in~\eqref{suff.proof.M_ary.9} leads to:
\begin{align}
& (\prod_{m=1}^{M}\alpha_{+,\msubK}^{m})  e^{t(K(1-2\delta)\gamma- \sum_{m=1}^{M} \hfeat)+ K(1-2\delta) \log(\mathbb{E}_{P}[e^{-tL_{G}}])} \twocolbreak \includeonetwocol{}{\hspace{0.2in}\times}
e^{K\delta\log(\mathbb{E}_{Q}[e^{-tL_{G}}]) + tK\gamma\delta} \nonumber \\ 
& \leq e^{-\log(n)(1-2\delta-\delta A_{2})(\eta_{1} + o(1))} \nonumber \\
& \leq n^{-(1+\varepsilon)(1-2\delta-\delta A_{2}) + o(1))} \label{suff.proof.M_ary.10}
\end{align}
where~\eqref{suff.proof.M_ary.10} holds by assuming $\eta_{1} \geq 1+ \varepsilon$ for some $\varepsilon >0$. Multiplying~\eqref{suff.proof.M_ary.10}  by $K$:
\begin{align}
&K (\prod_{m=1}^{M}\alpha_{+,\msubK}^{m}) \twocolbreak \includeonetwocol{}{\hspace{0.2in}\times}  e^{t(K(1-2\delta)\gamma- \sum_{m=1}^{M} \hfeat)+ K(1-2\delta) \log(\mathbb{E}_{P}[e^{-tL_{G}}])}  \twocolbreak \includeonetwocol{}{\hspace{0.2in}\times}
e^{K\delta\log(\mathbb{E}_{Q}[e^{-tL_{G}}]) + tK\gamma\delta} \nonumber \\
& \leq n^{1 - (1+\varepsilon)(1-2\delta-\delta A_{2}) + o(1))}
\end{align}
Thus, for any $\varepsilon >0$, there exists a sufficiently small $\delta$ such that $(1+\varepsilon)(1-2\delta-\delta A_{2}) > 1$. This concludes the proof of the first case of Lemma~\ref{Lem.suff.M_ary}.

\item If $\sum_{m=1}^{M} \hfeat = o(\log(n))$ and $\sum_{m=1}^{M} \log(\alpha_{+,\msubK}^{m}) = \sum_{m=1}^{M} \log(\alpha_{-,\msubK}^{m})=-\beta\log(n)+ o(\log(n)), \beta > 0$, then:
\begin{align}
& (\prod_{m=1}^{M}\alpha_{+,\msubK}^{m}) e^{t(K(1-2\delta)\gamma- \sum_{m=1}^{M} \hfeat)+ K(1-2\delta) \log(\mathbb{E}_{P}[e^{-tL_{G}}])} \twocolbreak \includeonetwocol{}{\hspace{0.2in}\times}
e^{K\delta\log(\mathbb{E}_{Q}[e^{-tL_{G}}]) + tK\gamma\delta} \nonumber \\ 
& \leq \prod_{m=1}^{M}(\alpha_{+,\msubK}^{m})e^{-\log(n)(1-2\delta-\delta A_{2}) (\eta_{1} + o(1))} \label{suff.proof.M_ary.11}
\end{align}
Since $\sum_{m=1}^{M} \log(\alpha_{+,\msubK}^{m}) = -\beta\log(n)+ o(\log(n)), \beta > 0$:
\begin{align}
& (\prod_{m=1}^{M}\alpha_{+,\msubK}^{m})  e^{t(K(1-2\delta)\gamma- \sum_{m=1}^{M} \hfeat)+ K(1-2\delta) \log(\mathbb{E}_{P}[e^{-tL_{G}}])} \twocolbreak \includeonetwocol{}{\hspace{0.2in}\times}
e^{K\delta\log(\mathbb{E}_{Q}[e^{-tL_{G}}]) + tK\gamma\delta} \nonumber \\ 
& \leq e^{-\log(n)(1-2\delta-\delta A_{2}) (\eta_{1} + \frac{\beta}{1-2\delta-\delta A_{2}} + o(1))} \nonumber \\
& \leq e^{-\log(n)(1-2\delta-\delta A_{2}) (\eta_{1} + \beta + o(1))} \label{suff.proof.M_ary.12}
\end{align}
where the last inequality holds because $0<1-2\delta-\delta A_{2}<1$ for sufficiently small $\delta$. Thus: 
\begin{align}
& K (\prod_{m=1}^{M}\alpha_{+,\msubK}^{m}) \twocolbreak \includeonetwocol{}{\hspace{0.2in}\times} e^{t(K(1-2\delta)\gamma- \sum_{m=1}^{M} \hfeat)+ K(1-2\delta) \log(\mathbb{E}_{P}[e^{-tL_{G}}])} \twocolbreak \includeonetwocol{}{\hspace{0.2in}\times}
e^{K\delta\log(\mathbb{E}_{Q}[e^{-tL_{G}}]) + tK\gamma\delta} \nonumber \\ 
& \leq  n^{1 - (\eta_{1} + \beta) (1-2\delta-\delta A_{2}) + o(1)} \nonumber \\
& \leq n^{1 - (1+\varepsilon)(1-2\delta-\delta A_{2}) + o(1))}
\end{align}
where the last inequality holds by assuming $\eta_{1} + \beta \geq 1+ \varepsilon$ for some $\varepsilon >0$. Thus, for any $\varepsilon >0$, there exists a sufficiently small $\delta$ such that $(1+\varepsilon)(1-2\delta-\delta A_{2}) > 1$. This concludes the proof of the second case of Lemma~\ref{Lem.suff.M_ary}.

\item If $\sum_{m=1}^{M} \hfeat = \beta\log(n)+ o(\log(n)), 0 < \beta < \rho(a-b-bT), \sum_{m=1}^{M}\log(\alpha_{+,\msubK}^{m}) = o(\log(n))$, then:
\begin{align}
& (\prod_{m=1}^{M}\alpha_{+,\msubK}^{m})  e^{t(K(1-2\delta)\gamma- \sum_{m=1}^{M} \hfeat)+ K(1-2\delta) \log(\mathbb{E}_{P}[e^{-tL_{G}}])} \twocolbreak \includeonetwocol{}{\hspace{0.2in}\times}
e^{K\delta\log(\mathbb{E}_{Q}[e^{-tL_{G}}]) + tK\gamma\delta} \nonumber \\
&\leq (\prod_{m=1}^{M}\alpha_{+,\msubK}^{m})  e^{(1-2\delta) \big( t(k\gamma -\sum_{m=1}^{M} \hfeat) + K \log(\mathbb{E}_{P}[e^{-tL_{G}}]) \big)} \twocolbreak  \includeonetwocol{}{\hspace{0.2in}\times}
e^{\delta \big( t(k\gamma - \sum_{m=1}^{M} \hfeat) +  K\log(\mathbb{E}_{Q}[e^{-tL_{G}}])\big)}\label{suff.proof.M_ary.13}
\end{align}
Since $\sum_{m=1}^{M} \hfeat = \beta\log(n)+ o(\log(n)), 0 < \beta < \rho(a-b-bT)$, it is easy to show that 
\[
K\gamma-\sum_{m=1}^{M} \hfeat \in [-KD(Q||P)\;,\;KD(P||Q)]
\]
Define $\theta \triangleq K\gamma-\sum_{m=1}^{M} \hfeat$ and choose $t^{*}\in [0,1]$, such that $t^{*}\theta + K\log(\mathbb{E}[e^{-t^{*}L_{G}}]) = -E_{P}(\theta,K)$. Substituting in~\eqref{suff.proof.M_ary.13}:
\begin{align}
& (\prod_{m=1}^{M}\alpha_{+,\msubK}^{m})  e^{t(K(1-2\delta)\gamma- \sum_{m=1}^{M} \hfeat)+ K(1-2\delta) \log(\mathbb{E}_{P}[e^{-tL_{G}}])}  \twocolbreak \includeonetwocol{}{\hspace{0.2in}\times}
e^{K\delta\log(\mathbb{E}_{Q}[e^{-tL_{G}}]) + tK\gamma\delta} \nonumber \\ 
& \leq (\prod_{m=1}^{M} \alpha_{+,\msubK}^{m})  \twocolbreak \includeonetwocol{}{\hspace{0.2in}\times}
e^{-(1-2\delta) E_{P}[\theta,K] + \delta \big( t^{*}(K\gamma - \sum_{m=1}^{M} \hfeat)+  K\log(\mathbb{E}_{Q}[e^{-t^{*}L_{G}}])\big)} \label{suff.proof.M_ary.14}
\end{align}
By Lemma~\ref{Lem.2} and convexity of $\log(\mathbb{E}_{Q}[e^{-t^{*}L_{G}}])$, the following holds for some positive constant $A$: 
\begin{equation}
K \log(\mathbb{E}_{Q}[e^{-t^{*}L_{G}}]) \leq K \log(\mathbb{E}_{Q}[e^{-L_{G}}]) \leq A KD(Q||P) \label{suff.proof.M_ary.15}
\end{equation}
Moreover, since 
\[
-KD(Q||P) < K\gamma-\sum_{m=1}^{M} \hfeat <0   \quad ,
\]
it follows that  $\theta = -(1-\tilde{\eta})KD(Q||P)$ for some $\tilde{\eta} \in (0,1)$. Thus, by Lemma~\ref{Lem.2}, for some positive constant $A_{1}$: 
\begin{align*}
E_{Q}[\theta,K] &= E_{Q}[-(1-\tilde{\eta})KD(Q||P),K] \\
& \geq A_{1}KD(Q||P) \\
&\geq \frac{A_{1}}{A} K \log(\mathbb{E}_{Q}[e^{-t^{*}L_{G}}])
\end{align*}
where the last inequality holds because of~\eqref{suff.proof.M_ary.15}. Substituting in~\eqref{suff.proof.M_ary.14}, for some positive constant $A_{2}$,
\begin{align}
& (\prod_{m=1}^{M}\alpha_{+,\msubK}^{m})  e^{t(K(1-2\delta)\gamma- \sum_{m=1}^{M} \hfeat)+ K(1-2\delta) \log(\mathbb{E}_{P}[e^{-tL_{G}}])} \twocolbreak  \includeonetwocol{}{\hspace{0.2in}\times}
e^{K\delta\log(\mathbb{E}_{Q}[e^{-tL_{G}}]) + tK\gamma\delta} \nonumber \\ 
& \leq (\prod_{m=1}^{M}\alpha_{+,\msubK}^{m}) 
e^{-(1-2\delta) (E_{Q}[\theta,K]  - \theta)  + \delta A_{2}E_{Q}[\theta,K] } \nonumber \\
&\leq  (\prod_{m=1}^{M}\alpha_{+,\msubK}^{m})  
e^{-E_{Q}[\theta,K] (1-2\delta-\delta A_{2})  + (1-2\delta) \theta } \label{suff.proof.M_ary.16}
\end{align}
Since $\sum_{m=1}^{M} \log(\alpha_{+,\msubK}^{m}) = o(\log(n))$, by evaluating the supremum in $E_{Q}[\theta,K]$, multiplying by $K$ and substituting in~\eqref{suff.proof.M_ary.16}:
\begin{align}
& K (\prod_{m=1}^{M}\alpha_{+,\msubK}^{m})  e^{t(K(1-2\delta)\gamma- \sum_{m=1}^{M} \hfeat)+ K(1-2\delta) \log(\mathbb{E}_{P}[e^{-tL_{G}}])}  \twocolbreak \includeonetwocol{}{\hspace{0.2in}\times}
e^{K\delta\log(\mathbb{E}_{Q}[e^{-tL_{G}}]) + tK\gamma\delta} \nonumber \\ 
& \leq K e^{-\log(n)(1-2\delta-\delta A_{2})(\eta_{2} -\beta + \frac{(1-2\delta)\beta}{1-2\delta-\delta A_{2}} + o(1))} \nonumber \\
& \overset{(a)}{\leq} K e^{-\log(n)(1-2\delta-\delta A_{2})(\eta_{2} + o(1))} \nonumber \\
& \leq n^{1-(1+\varepsilon)(1-2\delta-\delta A_{2}) + o(1))} \label{suff.proof.M_ary.17}
\end{align}
where $(a)$ holds for sufficiently small $\delta$. Thus, for any $\varepsilon >0$, there exists a sufficiently small $\delta$ such that $(1+\varepsilon)(1-2\delta-\delta A_{2}) > 1$. This concludes the proof of the third case of Lemma~\ref{Lem.suff.M_ary}.

\item If $\sum_{m=1}^{M} \hfeat = -\beta\log(n)+ o(\log(n)), 0 < \beta < \rho(a-b-bT), \sum_{m=1}^{M}\log(\alpha_{-,\msubK}^{m}) = o(\log(n))$, then:
\begin{align}
& (\prod_{m=1}^{M}\alpha_{+,\msubK}^{m}) e^{t(K(1-2\delta)\gamma- \sum_{m=1}^{M} \hfeat)+ K(1-2\delta) \log(\mathbb{E}_{P}[e^{-tL_{G}}])}  \twocolbreak \includeonetwocol{}{\hspace{0.2in}\times}
e^{K\delta\log(\mathbb{E}_{Q}[e^{-tL_{G}}]) + tK\gamma\delta} \nonumber \\
&\leq (\prod_{m=1}^{M}\alpha_{+,\msubK}^{m})  e^{(1-2\delta) \big( t(k\gamma -\sum_{m=1}^{M} \hfeat) + K \log(\mathbb{E}_{P}[e^{-tL_{G}}]) \big)}  \twocolbreak \includeonetwocol{}{\hspace{0.2in}\times}
e^{\delta \big( t(k\gamma - 2\sum_{m=1}^{M} \hfeat) +  K\log(\mathbb{E}_{Q}[e^{-tL_{G}}])\big)}\label{suff.proof.M_ary.18}
\end{align}
Following similar analysis as in~\eqref{suff.proof.M_ary.16}:
\begin{align}
& (\prod_{m=1}^{M}\alpha_{+,\msubK}^{m})  e^{t(K(1-2\delta)\gamma- \sum_{m=1}^{M} \hfeat)+ K(1-2\delta) \log(\mathbb{E}_{P}[e^{-tL_{G}}])}  \twocolbreak \includeonetwocol{}{\hspace{0.2in}\times}
e^{K\delta\log(\mathbb{E}_{Q}[e^{-tL_{G}}]) + tK\gamma\delta} \nonumber \\ 
&\leq (\prod_{m=1}^{M}\alpha_{+,\msubK}^{m})  e^{-\log(n)(1-2\delta-\delta A_{2})(\eta_{3}+ o(1))} e^{-\sum_{m=1}^{M} \hfeat}\label{suff.proof.M_ary.19}
\end{align}
Since $\sum_{m=1}^{M}\log(\alpha_{-,\msubK}^{m}) = o(\log(n))$, by multiplying by $K$:
\begin{align}
& K (\prod_{m=1}^{M}\alpha_{+,\msubK}^{m})  e^{t(K(1-2\delta)\gamma- \sum_{m=1}^{M} \hfeat)+ K(1-2\delta) \log(\mathbb{E}_{P}[e^{-tL_{G}}])} \twocolbreak \includeonetwocol{}{\hspace{0.2in}\times}
e^{K\delta\log(\mathbb{E}_{Q}[e^{-tL_{G}}]) + tK\gamma\delta} \nonumber \\
&\leq K e^{-\log(n)(1-2\delta-\delta A_{2})(\eta_{3}+ o(1))} \nonumber \\
&\leq K n^{-(\eta_{3}+o(1))(1-2\delta-\delta A_{2})} \nonumber \\ 
&\leq n^{1 -(1+\varepsilon)(1-2\delta-\delta A_{2}) + o(1))} \label{suff.proof.M_ary.20}
\end{align}
where the last inequality holds by  assuming $\eta_{3} \geq 1+ \varepsilon$ for some $\varepsilon >0$. Thus, for any $\varepsilon >0$, there exists a sufficiently small $\delta$ such that $(1+\varepsilon)(1-2\delta-\delta A_{2}) > 1$. This concludes the proof of the fourth case of Lemma~\ref{Lem.suff.M_ary}.

\item If $\sum_{m=1}^{M} \hfeat = \beta\log(n)+ o(\log(n)), 0 < \beta < \rho(a-b-bT), \sum_{m=1}^{M}\log(\alpha_{+,\msubK}^{m}) = -\beta'\log(n) + o(\log(n))$, then following similar analysis as in~\eqref{suff.proof.M_ary.16}:
\begin{align}
& K(\prod_{m=1}^{M}\alpha_{+,\msubK}^{m})  e^{t(K(1-2\delta)\gamma- \sum_{m=1}^{M} \hfeat)+ K(1-2\delta) \log(\mathbb{E}_{P}[e^{-tL_{G}}])}  \twocolbreak \includeonetwocol{}{\hspace{0.2in}\times}
e^{K\delta\log(\mathbb{E}_{Q}[e^{-tL_{G}}]) + tK\gamma\delta} \nonumber \\ 
& \leq K n^{-(1-2\delta-\delta A_{2})(\eta_{2} + \beta' +o(1))} \nonumber \\
& \leq n^{1-(1+\varepsilon)(1-2\delta-\delta A_{2}) + o(1))} \label{suff.proof.M_ary.21}
\end{align}
where the last inequality holds by  assuming $\eta_{2} + \beta' \geq 1+ \varepsilon$ for some $\varepsilon >0$. Thus, for any $\varepsilon >0$, there exists a sufficiently small $\delta$ such that $(1+\varepsilon)(1-2\delta-\delta A_{2}) > 1$. This concludes the proof of the fifth case of Lemma~\ref{Lem.suff.M_ary}.

\item If $\sum_{m=1}^{M} \hfeat = -\beta\log(n)+ o(\log(n)), 0 < \beta < \rho(a-b-bT), \sum_{m=1}^{M}\log(\alpha_{-,\msubK}^{m}) = -\beta'\log(n) + o(\log(n))$, then following similar analysis as in~\eqref{suff.proof.M_ary.16}:
\begin{align}
& K (\prod_{m=1}^{M}\alpha_{+,\msubK}^{m})  e^{t(K(1-2\delta)\gamma- \sum_{m=1}^{M} \hfeat)+ K(1-2\delta) \log(\mathbb{E}_{P}[e^{-tL_{G}}])}  \twocolbreak  \includeonetwocol{}{\hspace{0.2in}\times}
e^{K\delta\log(\mathbb{E}_{Q}[e^{-tL_{G}}]) + tK\gamma\delta} \nonumber \\ 
& \leq K n^{-(1-2\delta-\delta A_{2})(\eta_{3} + \beta' +o(1))} \nonumber \\
& \leq n^{1-(1+\varepsilon)(1-2\delta-\delta A_{2}) + o(1))} \label{suff.proof.M_ary.22}
\end{align}
where the last inequality holds by  assuming $\eta_{3} + \beta' \geq 1+ \varepsilon$ for some $\varepsilon >0$. Thus, for any $\varepsilon >0$, there exists a sufficiently small $\delta$ such that $(1+\varepsilon)(1-2\delta-\delta A_{2}) > 1$. This concludes the proof of the last case of Lemma~\ref{Lem.suff.M_ary}.
\end{itemize}

\end{proof}

The proof of Lemma~\ref{The.3.6} then follows similarly as the proof of Lemma~\ref{The.5.new}.




\section{Auxiliary Lemmas For Belief Propagation} 
\label{App.2.BP}

\begin{lemma}
\label{Lem.3}
Recall the definition of $\LLRCroppedTree$ from~\eqref{likeli}. For any measurable function $g(.)$:
\begin{equation}
\label{measurable}
\mathbb{E}[g(\LLRCroppedTree) | \RootLabel = 0] = \mathbb{E}[g(\LLRCroppedTree) e^{-\LLRCroppedTree} | \RootLabel = 1]
\end{equation}
\end{lemma}

\begin{proof}
Let $ Y = (\CroppedTree,\TreeSideInfo^t)$ denote the observed tree and side information. Then,
\begin{align}\nonumber
\mathbb{E}[g(\LLRCroppedTree) | \RootLabel = 0] &= \mathbb{E}_{Y|\RootLabel =0}[g(\LLRCroppedTree)] \\ \nonumber 
& = \int_{Y} g(\LLRCroppedTree) \frac{\mathbb{P}(Y|\RootLabel=0)}{\mathbb{P}(Y|\RootLabel=1)} \mathbb{P}(Y|\RootLabel=1) \\ \nonumber
& = \int_{Y} g(\LLRCroppedTree) e^{-\LLRCroppedTree} \mathbb{P}(Y|\RootLabel=1) \\ \nonumber
&= \mathbb{E}_{Y|\RootLabel =1}[g(\LLRCroppedTree)e^{-\LLRCroppedTree}] \nonumber \\
&= \mathbb{E}[g(\LLRCroppedTree) e^{-\LLRCroppedTree}| \RootLabel = 1]
\end{align}
\end{proof}


\begin{lemma}
\label{Lem.4}

Let $b_{t} = \mathbb{E}[\frac{e^{Z_{1}^{t}+U_{1}}}{1+ e^{Z_{1}^{t}+U_{1}-\nu} }]$ and $a_{t} = \mathbb{E}[e^{2(Z_{0}^{t}+U_{0})}]$. Let $\ChiS = \mathbb{E}[e^{U_{1}}] = \mathbb{E}[e^{2U_{0}}]$. Then, for any $t \geq 0$ 
\begin{align}
&a_{t+1}  = \mathbb{E}[e^{Z_{1}^{t}+U_{1}}] = \ChiS e^{\lambda b_{t}} \label{a_t_b_t} \\ 
&\mathbb{E}[e^{3(Z_{0}^{t}+U_{0})}]  = \mathbb{E}[e^{2(Z_{1}^{t}+U_{1})}] \twocolbreak \includeonetwocol{}{\hspace{0.75in}}
 = \mathbb{E}[e^{3U_{0}}] e^{3\lambda b_{t} + \frac{\lambda^{2}}{K(p-q)} \mathbb{E}[ ( \frac{ e^{Z_{1}^{t} + U_{1}} }{1+e^{Z_{1}^{t} + U_{1}-\nu}} )^{2} ] } \label{a_t_b_t_2}
\end{align}
\end{lemma}

\begin{proof}

The first equality in~\eqref{a_t_b_t} holds by Lemma~\ref{Lem.3} for $g(x) = e^{2x}$. Similarly, the first equality in ~\eqref{a_t_b_t_2} holds by Lemma~\ref{Lem.3} for $g(x) = e^{3x}$. 

Let $f(x) = \frac{1+ \frac{p}{q} x}{1+x} = 1 + \frac{\frac{p}{q}-1}{1+x^{-1}}$. Then:
\begin{align}
a_{t+1} & = \mathbb{E}[e^{2(Z_{0}^{t}+U_{0})}] \nonumber  \\ 
& \overset{(a)}{=} e^{-2K(p-q)} \mathbb{E}[e^{2U_{0}}] \mathbb{E}[ (\mathbb{E}[f^{2}(e^{Z_{1}^{t}+U_{1}-\nu})])^{H_{u}} ] \twocolbreak \includeonetwocol{}{\hspace{0.5in}\times}
 \mathbb{E}[ (\mathbb{E}[f^{2}(e^{Z_{0}^{t}+U_{0}-\nu})])^{F_{u}} ] \nonumber \\ 
&\overset{(b)}{=} \ChiS e^{-2K(p-q)} e^{Kq(\mathbb{E}[f^{2}(e^{Z_{1}^{t}+U_{1}-\nu})]-1)} \twocolbreak
\includeonetwocol{}{\hspace{0.5in}\times}
e^{(n-K)q(\mathbb{E}[f^{2}(e^{Z_{0}^{t}+U_{0}-\nu})]-1)} \label{lem4.1}
\end{align}
where $(a)$ holds by the definition of $Z_{0}^{t}$ and $U_{0}$, $(b)$ holds by the definition of $\ChiS$ and by using the fact that $\mathbb{E}[c^{X}] = e^{\lambda(c-1)}$ for $X \sim \text{Poi}(\lambda)$ and $c > 0$. By the definition of $f(x)$:
\begin{align}
& Kq\big(\mathbb{E}\big[f^{2}(e^{Z_{1}^{t}+U_{1}-\nu})\big]-1\big) + (n-K)q\big(\mathbb{E}\big[f^{2}(e^{Z_{0}^{t}+U_{0}-\nu})\big]-1\big) \nonumber \\ 
& = Kq \mathbb{E}\bigg[\frac{2(\frac{p}{q}-1)}{1+e^{-(Z_{1}^{t}+U_{1}-\nu)}} + \frac{(\frac{p}{q}-1)^{2}}{(1+e^{-(Z_{1}^{t}+U_{1}-\nu)})^{2}} \bigg] \twocolbreak \includeonetwocol{}{\hspace{0.1in}}
+ (n-K)q  \mathbb{E}\bigg[\frac{2(\frac{p}{q}-1)}{1+e^{-(Z_{0}^{t}+U_{0}-\nu)}} + \frac{(\frac{p}{q}-1)^{2}}{(1+e^{-(Z_{0}^{t}+U_{0}-\nu)})^{2}} \bigg] \nonumber \\ 
& \overset{(a)}{=} 2K(p-q) + Kq (\frac{p}{q}-1)^{2} \mathbb{E}\bigg[\frac{1}{1+e^{-(Z_{1}^{t}+U_{1}-\nu)}}\bigg] \nonumber \\ 
& \overset{(b)}{=} 2K(p-q) + \lambda b_{t} \label{lem4.2}
\end{align}
where $(a)$ holds by Lemma~\ref{Lem.3} and $(b)$ holds by the definition of $\lambda$ and $b_{t}$.

Using~\eqref{lem4.2} and substituting in~\eqref{lem4.1} concludes the proof of~\eqref{a_t_b_t}. The proof of~\eqref{a_t_b_t_2} follows similarly using $f^{3}(x)$ instead of $f^{2}(x)$.
\end{proof}


\section{Proof of Lemma~\ref{recursive}} 
\label{BP-App.1}
The independent splitting property of the Poisson distribution is used to give an equivalent description of the numbers of children having a given label for any vertex in the tree. An equivalent description of the generation of the tree is as follows: for each node $i$, generate  a set ${\mathcal N}_i$ of children with $N_i=|{\mathcal N}_i|$. If $\tau_{i} = 1$, we generate $N_{i} \sim \text{Poi}(Kp + (n-K)q)$ children. Then for each child $j$, independent from everything else, let $\tau_{j} =1$ with probability $\frac{Kp}{Kp + (n-K)q}$  and  $\tau_{j} =0$ with probability $\frac{(n-K)q}{Kp + (n-K)q}$. If $\tau_{i} = 0$ generate $N_{i} \sim \text{Poi}(nq)$, then for each child $j$, independent from everything else, let $\tau_{j} =1$ with probability $\frac{K}{n}$ and $\tau_{j} =0$ with probability $\frac{(n-K)}{n}$. Finally, for each node $i$ in the tree, $\tilde{\tau}_{i}$ is observed according to $\alpha_{+,\ell}, \alpha_{-,\ell}$. Then:
\begin{align} 
\Gamma_{0}^{t+1} &= \log\Bigg( \frac{\mathbb{P}(T^{t+1},\tilde{\tau}^{t+1} | \RootLabel = 1)}{\mathbb{P}(T^{t+1},\tilde{\tau}^{t+1} | \RootLabel = 0)} \Bigg) \nonumber \\ 
& = \log\Bigg( \frac{\mathbb{P}(N_{0} , \tilde{\tau}_{0}, \{T_{k}^{t}\}_{k \in {\mathcal N}_0},\{\tilde{\tau}_{k}^{t}\}_{k \in {\mathcal N}_0} | \RootLabel = 1)}{\mathbb{P}(N_{0} , \tilde{\tau}_{0}, \{T_{k}^{t}\}_{k \in {\mathcal N}_0},\{\tilde{\tau}_{k}^{t}\}_{k \in {\mathcal N}_0} | \RootLabel = 0)} \Bigg) \nonumber\\ 
& \overset{(a)}{=} \log\Bigg( \frac{\mathbb{P}\big( N_{0},\tilde{\tau}_{0} | \RootLabel =1\big)}{\mathbb{P}\big( N_{0},\tilde{\tau}_{0} | \RootLabel = 0\big)} \Bigg) \twocolbreak \includeonetwocol{}{\hspace{0.2in}} +  \log\Bigg( \frac{\prod_{k \in {\mathcal N}_0} \mathbb{P}\big(T_{k}^{t},\tilde{\tau}_{k}^{t} | \RootLabel = 1\big)}{\prod_{k \in {\mathcal N}_0} \mathbb{P}\big(T_{k}^{t},\tilde{\tau}_{k}^{t} | \RootLabel = 0\big)} \Bigg) \nonumber \\ 
& \overset{(b)}{=} \log\Bigg( \frac{\mathbb{P}\big( N_{0} | \RootLabel =1\big)}{\mathbb{P}\big( N_{0} | \RootLabel =0\big)} \Bigg) + \log\Bigg( \frac{\mathbb{P}\big( \tilde{\tau}_{0} | \RootLabel =1\big)}{\mathbb{P}\big( \tilde{\tau}_{0} | \RootLabel =0\big)} \Bigg) \nonumber \\
& \includeonetwocol{}{\hspace{0.1in}} + \sum_{k \in {\mathcal N}_0} \log\Bigg( \frac{ \sum_{\tau_{k} \in \{0,1\} } \mathbb{P}\big( T_{k}^{t},\tilde{\tau}_{k}^{t} | \tau_{k} \big) \mathbb{P}\big( \tau_{k}|\RootLabel=1 \big) }{\sum_{\tau_{k} \in \{0,1\} }  \mathbb{P}\big( T_{k}^{t},\tilde{\tau}_{k}^{t} | \tau_{k} \big) \mathbb{P}\big( \tau_{k}|\RootLabel=0 \big)} \Bigg) \nonumber  \\ 
& \overset{(c)}{=} -K(p-q) + h_{0} + \sum_{k \in {\mathcal N}_0} \log(\frac{\frac{p}{q}e^{\Gamma_{k}^{t} - \nu}+1}{e^{\Gamma_{k}^{t} - \nu}+1})\label{rec.eq}
\end{align} 
where $(a)$ holds because conditioned on $\RootLabel$: 1) $(N_{0},\tilde{\tau}_{0})$ are independent of the rest of the tree and 2) $(T_{k}^{t},\tilde{\tau}_{k}^{t})$ are independent random variables $\forall k \in {\mathcal N}_0$,  $(b)$ holds because conditioned on $\RootLabel$, $N_{0}$ and $\tilde{\tau}_{0}$ are independent, $(c)$ holds by the definition of $N_{0}$ and $h_{0}$ and because $\tau_{k}$ is Bernoulli-$\frac{Kp}{Kp + (n-K)q}$ if $\RootLabel =1$ and is Bernoulli-$\frac{K}{n}$ if $\RootLabel =0$.
%


\section{Proof of Lemma~\ref{lower.upper}} 
\label{BP-App.2}
Let $f(x) \triangleq \frac{1+\frac{p}{q}x}{1+x}$, then:
\begin{align}
\mathbb{E}\big[e^{\frac{Z_{0}^{t}}{2}}\big] 
& = e^{\frac{-K(p-q)}{2}} \mathbb{E}_{H_0}\big[ (\mathbb{E}_{Z_1U_1}[f^{\frac{1}{2}}(e^{Z_{1}^{t}+U_{1}-\nu})])^{H_{0}} \big] \twocolbreak\includeonetwocol{}{\hspace{1in}\times}\mathbb{E}_{F_0}\big[ (\mathbb{E}_{Z_0U_0}[f^{\frac{1}{2}}(e^{Z_{0}^{t}+U_{0}-\nu})])^{F_{0}} \big] \nonumber \\ 
&\overset{(a)}{=} e^{\frac{-K(p-q)}{2}} e^{Kq(\mathbb{E}[f^{\frac{1}{2}}(e^{Z_{1}^{t}+U_{1}-\nu})]-1)} \twocolbreak\includeonetwocol{}{\hspace{1in}\times}e^{(n-K)q(\mathbb{E}[f^{\frac{1}{2}}(e^{Z_{0}^{t}+U_{0}-\nu})]-1)} \label{low.up.1}
\end{align}
where $(a)$ holds using  $\mathbb{E}[c^{X}] = e^{\lambda(c-1)}$ for $X \sim \text{Poi}(\lambda)$ and $c > 0$. 

By the intermediate value form of Taylor's theorem, for any $x \geq 0$ there exists $y$ with $1 \leq y \leq x$ such that $\sqrt{1+x} = 1 + \frac{x}{2} - \frac{x^{2}}{8(1+y)^{1.5}}$. Therefore,
\begin{equation}
\label{bound.f}
\sqrt{1+x} \leq 1 + \frac{x}{2} - \frac{x^{2}}{8(1+A)^{1.5}}, \qquad 0\leq x \leq A
\end{equation}
Let $A = \frac{p}{q} - 1$ and $B = (1+A)^{1.5}$. By assumption, $B$ is bounded. Then,	
\begin{align}
& \bigg(\frac{1 + \frac{p}{q} e^{Z_{0}^{t} + U_{0}-\nu}}{1 + e^{Z_{0}^{t} + U_{0}-\nu} }\bigg)^{\frac{1}{2}} \nonumber \\
&= \Big(1+ \frac{\frac{p}{q}-1}{1+ e^{-(Z_{0}^{t} + U_{0} -\nu)}}\Big)^{\frac{1}{2}} \nonumber \\
& \leq 1 + \frac{1}{2} \frac{\frac{p}{q}-1}{1 + e^{-(Z_{0}^{t} + U_{0} -\nu)} } - \frac{1}{8B} \frac{(\frac{p}{q}-1)^{2}}{(1 + e^{-(Z_{0}^{t} + U_{0} -\nu)} )^{2}}   \label{bound.f1}
\end{align}
It follows that:
\begin{align}
& Kq\big(\mathbb{E}[f^{\frac{1}{2}}(e^{Z_{1}^{t}+U_{1}-\nu})]-1 \big) \twocolbreak \includeonetwocol{}{\hspace{0.2in}} 
+ (n-K)q \big(\mathbb{E}[f^{\frac{1}{2}}(e^{Z_{0}^{t}+U_{0}-\nu})]-1\big) \nonumber \\ 
& \leq  \frac{Kq(\frac{p}{q}-1)}{2}\bigg( \mathbb{E}\bigg[\frac{1}{1 + e^{-(Z_{1}^{t}+U_{1}-\nu)}}\bigg] \twocolbreak \includeonetwocol{}{\hspace{0.2in}} 
+ e^{\nu} \mathbb{E}\bigg[\frac{1}{1 + e^{-(Z_{0}^{t}+U_{0}-\nu)}}\bigg]  \bigg) \nonumber \\  
& \includeonetwocol{}{\hspace{0.2in}} - \frac{Kq(\frac{p}{q}-1)^{2}}{8B} \bigg( \mathbb{E}\bigg[\frac{1}{(1 + e^{-(Z_{1}^{t}+U_{1}-\nu)})^{2}}\bigg] \twocolbreak \includeonetwocol{}{\hspace{0.2in}}
+ e^{\nu} \mathbb{E}\bigg[\frac{1}{(1 + e^{-(Z_{0}^{t}+U_{0}-\nu)})^{2}}\bigg]  \bigg) \nonumber \\ 
& \overset{(a)}{=} \frac{K(p-q)}{2} - \frac{K(p-q)^{2}}{8Bq} \mathbb{E}\bigg[\frac{1}{1+e^{-(Z_{1}^{t}+U_{1}-\nu)}}\bigg] \\
& = \frac{K(p-q)}{2} - \frac{\lambda}{8B} b_{t} \label{low.up.2}
\end{align}
where $(a)$ holds by the following consequence of Lemma~\ref{Lem.3} (from Appendix~\ref{App.2.BP}):
\begin{align}
& \mathbb{E}\bigg[\frac{1}{1 + e^{-(Z_{1}^{t}+U_{1}-\nu)}}\bigg] + e^{\nu} \mathbb{E}\bigg[\frac{1}{1 + e^{-(Z_{0}^{t}+U_{0}-\nu)}}\bigg]  = 1 \nonumber \\ 
& \mathbb{E}\bigg[\frac{1}{(1 + e^{-(Z_{1}^{t}+U_{1}-\nu)})^{2}}\bigg] + e^{\nu} \mathbb{E}\bigg[\frac{1}{(1 + e^{-(Z_{0}^{t}+U_{0}-\nu)})^{2}}\bigg]  \twocolbreak  \includeonetwocol{}{\hspace{0.2in}}
 = \mathbb{E}\bigg[\frac{1}{1 + e^{-(Z_{1}^{t}+U_{1}-\nu)}}\bigg]
\end{align}  
Using~\eqref{low.up.1} and~\eqref{low.up.2}:
\begin{align}
\label{low.up.3}
\mathbb{E}\Big[e^{\frac{Z_{0}^{t}+U_{0}}{2}}\Big] & \leq \mathbb{E}\big[e^{\frac{U_{0}}{2}}\big] \; e^{\frac{-\lambda}{8B}b_{t}}  
\end{align}

Similarly, using the fact that $\sqrt{1+x} \geq 1 + \frac{x}{2} - \frac{x^{2}}{8}$ for all $x \geq 0$:
\begin{align}
\label{low.up.4}
\mathbb{E}\Big[e^{\frac{Z_{0}^{t}+U_{0}}{2}}\Big] & \geq \mathbb{E}\big[e^{\frac{U_{0}}{2}}\big]\, e^{\frac{-\lambda}{8}b_{t}}  
\end{align}


\section{Proof of Lemma~\ref{upper_b}}
\label{BP-App.3}

Fix $\lambda >0$ and define $(v_{t}: t \geq 0)$ recursively by $v_{0} = 0$ and $v_{t+1} = \lambda \ChiS e^{v_{t}}$.  From Lemma~\ref{Lem.4} in Appendix~\ref{App.2.BP}, $a_{t+1} = \ChiS e^{\lambda b_{t}}$.

We first prove by induction that $\lambda b_{t} \leq \lambda a_{t} \leq v_{t+1}$ for all $t \geq 0$. $a_{0} = \mathbb{E}[e^{U_{1}}] = \ChiS$ and $\lambda b_{0} = \lambda \mathbb{E}[\frac{e^{U_{1}}}{1+ e^{U_{1} - \nu}}] \leq \lambda \mathbb{E}[e^{U_{1}}] = \lambda a_{0}$. Thus, $\lambda b_{0} \leq \lambda a_{0} = \lambda \ChiS = v_{1}$. Assume that $\lambda b_{t-1} \leq \lambda a_{t-1} \leq v_{t}$. Then, $\lambda b_{t} \leq \lambda a_{t} = \lambda \ChiS e^{\lambda b_{t-1}} \leq \lambda \ChiS e^{v_{t}} = v_{t+1}$, where the first inequality holds by the definition of $a_{t}$ and $b_{t}$ and the second inequality holds by the induction assumption. Thus, $\lambda b_{t} \leq \lambda a_{t} \leq v_{t+1}$ for all $t \geq 0$.

Next we prove by induction that $\frac{v_{t}}{\lambda}$ is increasing in $t \geq 0$. We have $\frac{v_{t+1}}{\lambda} = \ChiS e^{v_{t}}$. Then, $\frac{v_{1}}{\lambda} = \ChiS \geq 0 = \frac{v_{o}}{\lambda}$. Now assume that $\frac{v_{t}}{\lambda} > \frac{v_{t-1}}{\lambda}$. Then,  $\frac{v_{t+1}}{\lambda} = \ChiS e^{v_{t}} = \ChiS e^{\lambda (\frac{v_{t}}{\lambda})} > \ChiS e^{v_{t-1}} = \frac{v_{t}}{\lambda}$. Thus, we have: $\frac{v_{t+1}}{\lambda} > \frac{v_{t}}{\lambda}$ for all $t \geq 0$.

Note that $\frac{v_{t+1}}{\lambda} = \ChiS e^{\lambda (\frac{v_{t}}{\lambda})}$ has the form of $x = \ChiS e^{\lambda x}$, which has no solutions for $\lambda > \frac{1}{\ChiS e}$ and has two solutions for $\lambda \leq\frac{1}{\ChiS e}$, where the largest solution is $\ChiS e$. Thus, for $\lambda\leq\frac{1}{\ChiS e}$, $b_{t} \leq \frac{v_{t+1}}{\lambda} \leq \ChiS e$.


\section{Proof of Lemma~\ref{Lem.lower.b.1}}
\label{BP-App.4}
By definition of $a_{t}$, we have:
\begin{align} 
a_{t+1} - \mathbb{E}\big[e^{-\nu + 2(Z_{1}^{t+1} + U_{1})}\big] & = \mathbb{E}\big[e^{Z_{1}^{t+1} + U_{1}} (1- e^{Z_{1}^{t+1} + U_{1}-\nu})\big] \nonumber\\ 
& \leq \mathbb{E}\bigg[\frac{e^{Z_{1}^{t}+U_{1}}}{1+ e^{Z_{1}^{t}+U_{1}-\nu} }\bigg] \nonumber\\
&= b_{t+1}\nonumber
\end{align}
where the first inequality holds because $ 1- x \leq \frac{1}{1+x} $. Then,
\begin{align} 
b_{t+1} &\geq a_{t+1} - \mathbb{E}[e^{-\nu + 2(Z_{1}^{t+1} + U_{1})}] \nonumber \\
&  \overset{(a)}{=} \ChiS e^{\lambda b_{t}} - e^{-\nu} \ChiS' e^{3\lambda b_{t} + \frac{\lambda^{2}}{K(p-q)} \mathbb{E}\big[ ( \frac{ e^{Z_{1}^{t} + U_{1}} }{1+e^{Z_{1}^{t} + U_{1}-\nu}} )^{2} \big] } \nonumber\\ \nonumber
& \overset{(b)}{\geq} \ChiS e^{\lambda b_{t}} - \ChiS' e^{Cb_{t} - \nu} \nonumber \\
&= \ChiS e^{\lambda b_{t}} \big(1- \frac{\ChiS'}{\ChiS} e^{-\nu + (C-\lambda)b_{t}}\big) \nonumber\\ 
&\overset{(c)}{\geq} \ChiS e^{\lambda b_{t}} \big(1- \frac{\ChiS'}{\ChiS} e^{\frac{-\nu}{2}}\big)
\end{align}
where $(a)$ holds from Lemma~\ref{Lem.4}, $(b)$ holds because $(\frac{e^{x}}{1+ e^{x-\nu}})^{2} \leq e^{\nu} (\frac{e^{x}}{1+e^{x-\nu}})$, which holds because $e^{\nu} \geq \frac{e^{x}}{1+ e^{x-\nu}}$ for all $x$, and $(c)$ holds by the assumption that $ b_{t} \leq \frac{\nu}{2(C-\lambda)}$.

\section{Proof of Lemma~\ref{Lem.lower.b.main}}
\label{BP-App.5}

Given $\lambda$ with $\lambda > \frac{1}{\ChiS e}$, assume $\nu \geq \nu_{o}$ and $\nu \geq 2\ChiS (C-\lambda)$ for some positive $\nu_{o}$. Moreover, select the following constants depending only on $\lambda$ and the LLR of side information:
\begin{itemize}
\item $D$ and $\nu_{o}$ large enough such that $\lambda \ChiS e^{\lambda D} (1- \frac{\ChiS'}{\ChiS}e^{-\nu_{o}}) > 1$ and $\ChiS\lambda e (1- \frac{\ChiS'}{\ChiS}e^{-\nu_{o}}) \geq \sqrt{\lambda \ChiS e}$.

\item $w_{o} > 0$ so large that 
\begin{equation}
w_{o}\lambda \ChiS e^{\lambda D} (1- \frac{\ChiS'}{\ChiS}e^{-\nu_{o}}) - \lambda D \geq w_{o}.
\label{eq:condition-w0}
\end{equation}

\item A positive integer $\bar{t}_{o}$ large enough such that $\lambda (\ChiS (\lambda \ChiS e)^{\frac{\bar{t}_{o}}{2}-1} -D) \geq w_{o}$
\end{itemize}

The goal is to show that there exists some $\tilde{t}$ after which $\nu = o(b_t)$. 

Let $t^{*} = \max\{ t > 0: b_{t} < \frac{\nu}{2(C-\lambda)}\}$ and $\bar{t}_{1} = \log^{*}(\nu)$. The first step is to show that $t^{*} \leq \bar{t}_{o} + \bar{t}_{1}$. 

By the definition of $b_t$,

\begin{align*}
b_{0} &= \mathbb{E}\big[\frac{e^{U_{1}}}{1+ e^{U_{1}-\nu}}\big] \\
& < \mathbb{E}[e^{U_{1}}] = \ChiS
\end{align*}

Since $\nu \geq 2\ChiS(C-\lambda)$, we get $b_{0} < \frac{\nu}{2(C-\lambda)}$. 


Since for all $t \leq t^{*}$, $b_{t} < \frac{\nu}{2(C-\lambda)}$, then by Lemma~\ref{Lem.lower.b.1}:
\begin{align}
b_{t+1} &\geq \ChiS e^{\lambda b_{t}} (1 - \frac{\ChiS'}{\ChiS} e^{\frac{-\nu}{2}}) \nonumber \\
& \geq \ChiS e^{\lambda b_{t}} (1 - \frac{\ChiS'}{\ChiS} e^{\frac{-\nu_{o}}{2}})
\end{align}
where the last inequality holds since $\nu \geq \nu_{o}$. Thus, 
\begin{align}
b_{1} &\geq \ChiS e^{\lambda b_{0}}(1 - \frac{\ChiS'}{\ChiS} e^{\frac{-\nu_{o}}{2}}) \nonumber \\
 &\geq \ChiS (1 - \frac{\ChiS'}{\ChiS} e^{\frac{-\nu_{o}}{2}})  \nonumber \\
&\geq \sqrt{\frac{\ChiS}{\lambda e}} \label{achie.1}
\end{align}
where the last inequality holds by the choice of $\nu_{o}$. Moreover,
\begin{align}
b_{t+1} &\geq \ChiS e^{\lambda b_{t}} (1 - \frac{\ChiS'}{\ChiS} e^{\frac{-\nu_{o}}{2}}) \nonumber \\
& \overset{(a)}{\geq} \ChiS e \lambda b_{t} (1 - \frac{\ChiS'}{\ChiS} e^{\frac{-\nu_{o}}{2}}) \nonumber \\
&\overset{(b)}{\geq} \sqrt{\ChiS \lambda e} b_{t}
\end{align}
where $(a)$ holds because $e^{u} \geq eu$ for all $u > 0$ and $(b)$ holds by choice of $\nu_{0}$. Thus, for all $1 \leq t \leq t^{*}+1$: $b_{t} \geq \sqrt{\ChiS \lambda e} b_{t-1}$. Since $b_{1} \geq \sqrt{\frac{\ChiS}{\lambda e}}$, it follows by induction that:
\begin{equation}
\label{lower.eq.3}
b_{t} \geq \ChiS (\lambda \ChiS e)^{\frac{t}{2} -1}  \text{ for all } 1 \leq t \leq t^{*}+1
\end{equation}


We now divide the analysis into two cases. First, if $\bar{t}_{o}$ is such that $b_{\bar{t}_{o} -1} \geq \frac{\nu}{2(C-\lambda)}$. This implies that $\bar{t}_{o} -1 \geq t^{*} +1$ by the definition of $t^{*}$. Thus, $t^{*} \leq \bar{t}_{o} -2 \leq \bar{t}_{o} + \bar{t}_{1}$, which proves our claim for the first case.

If $\bar{t}_{o}$ is such that $b_{\bar{t}_{o} -1} < \frac{\nu}{2(C-\lambda)}$. Then, $\bar{t}_{o} \leq t^{*} +1$. Thus, $b_{\bar{t}_{o}} \geq \ChiS (\lambda Le)^{\frac{\bar{t}_{o}}{2} -1}$. Let $t_{o} = \min \{ t: b_{t} \geq \ChiS (\lambda \ChiS e)^{\frac{\bar{t}_{o}}{2} -1}\}$. Thus, by Lemma~\ref{Lem.lower.b.2}, we get $t_{o} \leq \bar{t}_{o}$. Moreover, by the choice of $t_{o}$ and $w_{o}$:
\begin{align}
w_{o} &\leq \lambda (\ChiS (\lambda \ChiS e)^{\frac{\bar{t}_{o}}{2}-1} -D) \leq \lambda (b_{t_{o}} - D)
\end{align}
Now define sequence $(w_{t}: t \geq0)$: $w_{t+1} = e^{w_{t}}$, where $w_{o}$ was chosen according to~\eqref{eq:condition-w0}. We already showed that $w_{o} \leq \lambda (b_{t_{o}} - D)  $. Assume that $w_{t-1} \leq \lambda (b_{t_{o}+t-1} - D)$ for $t_{o}+t-1 \leq t^{*}$. Then, 
\begin{align}\nonumber
\lambda (b_{t_{o}+t} - D) &\overset{(a)}{\geq} \lambda ( \ChiS e^{\lambda b_{t_{o}+t-1}} (1- \frac{\ChiS'}{\ChiS}e^{-\nu_{o}}) - D) \\ \nonumber
&\overset{(b)}{\geq} \lambda ( \ChiS e^{\lambda D + w_{t-1}} (1- \frac{\ChiS'}{\ChiS}e^{-\nu_{o}}) - D) \\ \nonumber
&\overset{(c)}{=} \lambda \ChiS e^{\lambda D} w_{t} (1- \frac{\ChiS'}{\ChiS}e^{-\nu_{o}}) - \lambda D \\ \nonumber
&\overset{(d)}{\geq} w_{t}
\end{align}
where $(a)$ holds by Lemma~\ref{Lem.lower.b.1}, $(b)$ holds by the assumption that $w_{t-1} \leq \lambda (b_{t_{o}+t-1} - D)$, $(c)$ holds by the definition of the sequence $w_{t}$ and $(d)$ holds by the choice of $w_{o}$ and the fact that $w_{t} \geq w_{o}$. Thus, we showed by induction that 
\begin{equation}
\label{lower.eq.4}
w_{t} \leq \lambda (b_{t_{o}+t} - D) \text{ for } 0 \leq t \leq t^{*} - t_{o}+1. 
\end{equation}

By the definition of $\bar{t}_{1}$ and since $w_{1} \geq 1$, we have $\nu \leq w_{\bar{t}_{1}+1}$. Thus, $w_{\bar{t}_{1}+1} \geq \nu - \lambda D$. Since, by the definition of $C$, $\lambda \leq 2(C-\lambda)$. Therefore, $w_{\bar{t}_{1}+1} \geq  \frac{\nu \lambda}{2(C-\lambda)} - \lambda D$. We will show that $t^{*} \leq \bar{t}_{o} + \bar{t}_{1}$ by contradiction. Let $t^{*} > \bar{t}_{o} + \bar{t}_{1}$. Thus, from~\eqref{lower.eq.4}, for $t = t_{o} + \bar{t}_{1} + 1$: 
\begin{equation}
\label{lower.eq.5}
b_{t_{o}+\bar{t}_{1} + 1} \geq \frac{w_{\bar{t}_{1} + 1}}{\lambda}  + D \geq \frac{\nu}{2(C-\lambda)}
\end{equation}
which implies that $t_{o}+\bar{t}_{1} + 1 \geq t^{*} +1$, i.e., $t_{o}+\bar{t}_{1}  \geq t^{*} $, which contradicts the assumption that $t^{*} > \bar{t}_{o} + \bar{t}_{1}$. 

To sum up, we have shown so far that if $\lambda > \frac{1}{\ChiS e}$, then $t^{*} \leq \bar{t}_{o} + \bar{t}_{1}$.

Since $t^{*}$ is the last iteration for $b_{t}  < \frac{\nu}{2(C - \lambda)}$. Then, $b_{t^{*}+1} \geq \frac{\nu}{2(C - \lambda)}$. We begin with $b_{t^{*}+1} =\frac{\nu}{2(C - \lambda)}$. Then by Lemma~\ref{Lem.lower.b.1}: 
\begin{align}
\label{lower.eq.6}
b_{t^{*}+2} &\geq \ChiS e^{\lambda b_{t^{*}+1}} (1 - \frac{\ChiS'}{\ChiS} e^{\frac{-\nu}{2}})
\end{align}

By Lemma~\ref{Lem.lower.b.2}, the sequence $b_{t}$ is non-decreasing in $t$. We also known $t^{*} +2 \leq \bar{t}_{o} + \bar{t}_{1} +2$. Using~\eqref{lower.eq.6}:
\begin{align}
\label{lower.eq.7}
b_{\bar{t}_{o} + \log^{*}(\nu)+2} \geq \ChiS e^{\frac{\lambda \nu}{2(C-\lambda)}} (1-\frac{\ChiS'}{{\ChiS}} e^{\frac{-\nu}{2}})
\end{align}
which concludes one case of the proof. 

When $b_{t^{*}+1} > \frac{\nu}{2(C - \lambda)}$, we use the truncation process~\cite[Lemma 6]{recovering_one}, which depends only on the tree structure. Applying this truncation process, it can directly be shown that the tree can be truncated such that with probability one the value of $b_{t^{*}+1}$ in the truncated tree is $\frac{\nu}{2(C - \lambda)}$. The truncation process~\cite[Lemma 6]{recovering_one} depends only on the structure of the tree. In this paper, the side information is independent of the tree structure given the labels, therefore the same truncation process holds for our case, which concludes the proof using~\eqref{lower.eq.6} and~\eqref{lower.eq.7}.



\section{Proof of Theorem~\ref{The.6}}
\label{BP-App.5.1}

The assumption $(np)^{\log^{*}(\nu)} = n^{o(1)}$ ensures that $(np)^{\TreeIter} = n^{o(1)}$. Since $\frac{K^{2}(p-q)^{2}}{q(n-K)} \to \lambda$,  $ p \geq q $ and $\frac{p}{q} = \theta(1)$, then $(\frac{n-K}{K})^{2} = O(np)$. Since $K = o(n)$, then $np \to \infty$. Thus, $(np)^{\TreeIter} = n^{o(1)}$ can be replaced by $(np+2)^{\TreeIter} = n^{o(1)}$, and hence, the coupling  Lemma~\ref{couple} holds. Moreover, since $(\frac{n-K}{K})^{2} = O(np)$ and $np = n^{o(1)}$,  $K = n^{1-o(1)}$.

Consider a modified form of Algorithm~\ref{tab} whose output is $\hat{C} = \{i: R_{i}^{\TreeIter} \geq \nu \}$. Then for deterministic $|C^{*}| = K$, the following holds:
\begin{align} \nonumber
p_{e} & = \mathbb{P}(\text{No coupling}) p_{e|\text{no coupling}} + \mathbb{P}(\text{coupling}) p_{e|\text{coupling}} \\ 
& \leq n^{-1 + o(1)} + \frac{K}{n} e^{-\nu(r + o(1))}\label{ache.1}
\end{align} 
where the last inequality holds by Lemmas~\ref{couple} and~\ref{lower_up_map} for some positive constant $r$. Multiplying~\eqref{ache.1} by $\frac{n}{K}$:
\begin{align}
\frac{\mathbb{E}[|C^{*} \triangle \hat{C} |]}{K} \leq \frac{n^{o(1)}}{K} + e^{-\nu(r + o(1))}  \to 0 \label{ache.2}
\end{align}
where the last inequality holds because $K = n^{1-o(1)}$ and $\nu \to \infty$.

Now going back to Algorithm~\ref{tab} and its output $\tilde{C}$, using Equation~\eqref{fix.est}:
\begin{equation}
\frac{\mathbb{E}[|C^{*} \triangle \tilde{C} |]}{K} \leq 2 \frac{\mathbb{E}[|C^{*} \triangle \hat{C} |]}{K} \to 0
\end{equation}
which concludes the proof under deterministic $|C^{*}| = K$.

When $|C^{*}|$ is random such that $K \geq 3\log(n)$ and $\mathbb{P}(| |C^{*}| - K | \geq \sqrt{3K\log(n)}) \leq n^{\frac{-1}{2} + o(1)}$, we have $\mathbb{E}[| |C^{*}| -K|] \leq n^{\frac{1}{2} + o(1)}$. Thus, for $\tilde{C}$, using Equation~\eqref{fix.est}: 
\begin{equation}
\frac{\mathbb{E}[|C^{*} \triangle \tilde{C} |]}{K} \leq 2 \frac{\mathbb{E}[|C^{*} \triangle \hat{C} |]}{K} + \frac{\mathbb{E}[| |C^{*}| -K|]}{K} \to 0
\end{equation}
which concludes the proof.


\section{Proof of Theorem~\ref{The.7}}
\label{BP-App.5.2}

Since $(np+2)^{\TreeIter} = n^{o(1)}$, the coupling Lemma~\ref{couple} holds. Moreover, since  $(\frac{n-K}{K})^{2} = O(np)$ and $np = n^{o(1)}$,  $K = n^{1-o(1)}$. Consider a deterministic $|C^{*}|=K$. Then, for any local estimator $\hat{C}$:
\begin{align} \nonumber
p_{e} & = \mathbb{P}(\text{No coupling}) p_{e|\text{no coupling}} + \mathbb{P}(\text{coupling}) p_{e|\text{coupling}} \\ 
& \geq  \frac{K(n-K)}{n^{2}} \mathbb{E}^{2}[e^{\frac{U_{0}}{2}}] e^{\frac{-\lambda \ChiS e}{4}} - n^{-1+o(1)}\label{conv.1}
\end{align} 
where the last inequality holds by Lemmas~\ref{couple} and~\ref{lower_up_map}. Multiplying~\eqref{conv.1} by $\frac{n}{K}$:
\begin{align}
\frac{\mathbb{E}[|C^{*} \triangle \hat{C} |]}{K} \geq \Big(1-\frac{K}{n}\Big) \mathbb{E}^{2}[e^{\frac{U_{0}}{2}}] e^{\frac{-\lambda \ChiS e}{4}} - o(1) \label{conv.2}
\end{align}
where the last inequality holds because $K = n^{1-o(1)}$. Thus, for $\lambda \leq \frac{1}{\ChiS e}$, $\frac{\mathbb{E}[|C^{*} \triangle \hat{C} |]}{K}$ is bounded away from zero for any local estimator $\hat{C}$.

It can be shown that under a non-deterministic $|C^*|$ that obeys a distribution in the class of distributions mentioned earlier, the local estimator will do no better, therefore the same converse will hold.


\section{Proof of Theorem~\ref{The.8}}
\label{BP-App.5.3}

Let $Z$ be a binomial random variable $\text{Bin}(n(1-\delta),\frac{K}{n})$. In view of Lemma~\ref{The.5.new}, it suffices to verify~\eqref{suff._exact.eq.1.new} when $\hat{C}_{k}$ for each $k$ is the output of belief propagation for estimating $C^{*}_{k}$ based on observing $\boldsymbol{G}_{k}$ and $\boldsymbol{Y}_{k}$. The distribution of $|C^{*}_{k}|$ is obtained by sampling the indices of the original graph without replacement. Thus, for any convex function $\phi$: $\mathbb{E}[\phi(|C^{*}_{k}|)] \leq \mathbb{E}[\phi(Z)]$. Therefore, Chernoff bound for $Z$ also holds for $|C^{*}_{k}|$. This leads to:
\begin{align}
\mathbb{P}\Big( \big| |C^{*}_{k}| - (1-\delta)K \big| \geq \sqrt{3K(1-\delta)\log(n)} \Big) & \overset{}{\leq} n^{-1.5 +o(1)} \nonumber \\
&\leq n^{\frac{-1}{2} + o(1)} \label{Suff.exact.BP.eq.1}
\end{align}
Thus, by Theorem~\ref{The.6}, belief propagation achieves weak recovery for recovering $C^{*}_{k}$ for each $k$. Thus:
\begin{equation}
\mathbb{P}\big(|\hat{C}_{k}\triangle C^{*}_{k}|  \leq \delta K \quad \text{ for } 1 \leq k \leq \frac{1}{\delta}\big) \to 1
\end{equation}
which together with Lemma~\ref{The.5.new} conclude the proof.


\section{Proof of Lemma~\ref{mean_variance_BP}}
\label{BP-App.6}

First, we expand $M(x)$ using Taylor series:
\begin{align}
M(x) & = \frac{\frac{p}{q}-1}{1 + e^{-(x-\nu)}} - \frac{1}{2} \big(\frac{\frac{p}{q}-1}{1 + e^{-(x-\nu)}}\big)^{2} \twocolbreak \includeonetwocol{}{\hspace{0.1in}}
+ O\bigg(\big(\frac{\frac{p}{q}-1}{1 + e^{-(x-\nu)}} \big)^{3} \bigg) \label{f_x_1}
\end{align}
Thus:
\begin{align}
\mathbb{E}[Z_{0}^{t+1}] =& -K(p-q) + Kq\mathbb{E}[M(Z_{1}^{t}+U_{1})] \twocolbreak 
 + (n-K)q \mathbb{E}[M(Z_{0}^{t}+U_{0})] \nonumber \\
 =&  -K(p-q) + K(p-q)\mathbb{E}\Big[\frac{1}{1+e^{-(Z_{1}^{t}+U_{1}-\nu)}}\Big] \twocolbreak 
+ (n-K)(p-q)\mathbb{E}\Big[\frac{1}{1+e^{-(Z_{0}^{t}+U_{0}-\nu)}}\Big] \nonumber \\
&   -\frac{K(p-q)^{2}}{2q}\mathbb{E}\Big[\big(\frac{1}{1+e^{-(Z_{1}^{t}+U_{1}-\nu)}}\big)^{2}\Big] \twocolbreak 
- \frac{(n-K)(p-q)^{2}}{2q}\mathbb{E}\Big[\big(\frac{1}{1+e^{-(Z_{1}^{0}+U_{0}-\nu)}}\big)^{2}\Big] \nonumber \\
& + O\Bigg(  \frac{K(p-q)^{3}}{q^{2}}\mathbb{E}\Big[\big(\frac{1}{1+e^{-(Z_{1}^{t}+U_{1}-\nu)}}\big)^{3}\Big] \twocolbreak 
+ \frac{(n-K)(p-q)^{3}}{q^{2}}\mathbb{E}\Big[\big(\frac{1}{1+e^{-(Z_{0}^{t}+U_{0}-\nu)}}\big)^{3}\Big] \Bigg) \label{mean_variance_BP.eq1}
\end{align}
Using Lemma~\ref{Lem.3} for $g(x) = \frac{1}{1+e^{-(x-\nu)}}$, 
\begin{align}
 K(p-q)\mathbb{E}&\Big[\frac{1}{1+e^{-(Z_{1}^{t}+U_{1}-\nu)}}\Big] \twocolbreak 
+ (n-K)(p-q)\mathbb{E}\Big[\frac{1}{1+e^{-(Z_{0}^{t}+U_{0}-\nu)}}\Big] \nonumber \\
& = K(p-q)  \label{mean_variance_BP.eq2}
\end{align}
Similarly:
\begin{align}
\frac{K(p-q)^{2}}{2q}\mathbb{E}&\Big[\big(\frac{1}{1+e^{-(Z_{1}^{t}+U_{1}-\nu)}}\big)^{2}\Big] \twocolbreak 
+ \frac{(n-K)(p-q)^{2}}{2q}\mathbb{E}\Big[\big(\frac{1}{1+e^{-(Z_{1}^{0}+U_{0}-\nu)}}\big)^{2}\Big]  \nonumber \\
& = \frac{K(p-q)^{2}}{2q}\mathbb{E}\Big[\frac{1}{1+e^{-(Z_{1}^{t}+U_{1}-\nu)}}\Big] \label{mean_variance_BP.eq3}
\end{align}
and,
\begin{align}
\frac{K(p-q)^{3}}{q^{2}}&\mathbb{E}\Big[\big(\frac{1}{1+e^{-(Z_{1}^{t}+U_{1}-\nu)}}\big)^{3}\Big] \twocolbreak  
+ \frac{(n-K)(p-q)^{3}}{q^{2}}\mathbb{E}\Big[\big(\frac{1}{1+e^{-(Z_{0}^{t}+U_{0}-\nu)}}\big)^{3}\Big]  \nonumber \\
& = \frac{K(p-q)^{3}}{q^{2}}\mathbb{E}\Big[\big(\frac{1}{1+e^{-(Z_{1}^{t}+U_{1}-\nu)}} \big)^{2}\Big] \label{mean_variance_BP.eq4}
\end{align}
Using~\eqref{mean_variance_BP.eq2},~\eqref{mean_variance_BP.eq3} and~\eqref{mean_variance_BP.eq4} and substituting in~\eqref{mean_variance_BP.eq1}:
\begin{align}
\mathbb{E}[Z_{0}^{t+1}] = & -\frac{\lambda}{2}b_{t} + O\bigg( \frac{K(p-q)^{3}}{q^{2}} \mathbb{E}[\big(\frac{1}{1+e^{-(Z_{1}^{t}+U_{1}-\nu)}}\big)^{2}]\bigg) \nonumber \\
=& -\frac{\lambda}{2}b_{t} + o(1) \label{mean_variance_BP.eq5}
\end{align}
where the last equality holds by the definition of $\lambda$ and $b_{t}$ and because $\frac{K(p-q)^{3}}{q^{2}} = \lambda \frac{n}{K} (1-\frac{K}{n})(\frac{p}{q}-1)$ which is $o(1)$ because of the assumptions of the lemma which also implies that $\frac{p}{q} \to 1$.

To show~\eqref{mean_2}, we use Taylor series: $M(x) = \frac{\frac{p}{q}-1}{1 + e^{-(x-\nu)}} +O(\big(\frac{\frac{p}{q}-1}{1 + e^{-(x-\nu)}}\big)^{2})$. Then,
\begin{align}
\mathbb{E}[Z_{1}^{t+1}] & = \mathbb{E}[Z_{0}^{t+1}] + K(p-q)\mathbb{E}[M(Z_{1}^{t}+U_{1})] \nonumber \\
& = \mathbb{E}[Z_{0}^{t+1}] + \frac{K(p-q)^{2}}{q} \mathbb{E}\Big[\frac{1}{1+e^{-(Z_{1}^{t}+U_{1}-\nu)}}\Big] \twocolbreak \includeonetwocol{}{\hspace{0.1in}}
 + O\bigg(\frac{K(p-q)^{3}}{q^{2}} \mathbb{E}\Big[\big( \frac{1}{1+e^{-(Z_{1}^{t}+U_{1}-\nu)}} \big)^{2}\Big] \bigg) \nonumber \\
& = \mathbb{E}[Z_{0}^{t+1}] + \lambda b_{t} + o(1) = \frac{\lambda}{2}b_{t} + o(1) \label{mean_variance_BP.eq6}
\end{align}

We now calculate the variance. For $Y = \sum_{i=1}^{L} X_{i}$, where $L$ is Poisson distributed and $\{X_i\}$ are independent of $Y$ and are i.i.d., it is well-known that $\text{var}(Y) = \mathbb{E}[L]\mathbb{E}[X_{1}^2]$. Thus,
\begin{align}
&\text{var}(Z_{0}^{t+1}) \nonumber \\
&=  Kq\, \mathbb{E}[M^{2}(Z_{1}^{t}+U_1)] + (n-K)q \, \mathbb{E}[M^{2}(Z_{0}^{t}+U_0)] \nonumber \\
& \overset{(a)}{=}  \frac{K(p-q)^{2}}{q^{2}}\mathbb{E}\Big[\big(\frac{1}{1+e^{-(Z_{1}^{t}+U_{1}-\nu)}}\big)^{2}\Big] \twocolbreak \includeonetwocol{}{\hspace{0.1in}} 
+ \frac{(n-K)(p-q)^{2}}{q^{2}}\mathbb{E}\Big[\big(\frac{1}{1+e^{-(Z_{0}^{t}+U_{0}-\nu)}}\big)^{2}\Big] \nonumber \\
& \includeonetwocol{}{\hspace{0.1in}} +  O\bigg(  \frac{K(p-q)^{3}}{q^{2}}\mathbb{E}\Big[\big(\frac{1}{1+e^{-(Z_{1}^{t}+U_{1}-\nu)}}\big)^{3}\Big] \twocolbreak \includeonetwocol{}{\hspace{0.1in}}
+ \frac{(n-K)(p-q)^{3}}{q^{2}}\mathbb{E}\Big[\big(\frac{1}{1+e^{-(Z_{0}^{t}+U_{0}-\nu)}}\big)^{3}\Big] \bigg)  \nonumber \\
& \overset{(b)}{=} \lambda b_{t} + o(1) \label{mean_variance_BP.eq7}
\end{align}
where $(a)$ holds because $\log^{2}(1+x) = x^{2} + O(x^{3})$ for all $x \geq 0$ and $(b)$ holds by similar analysis as in~\eqref{mean_variance_BP.eq5}. 

Similarly,
\begin{align}
\text{var}(Z_{1}^{t+1}) &= \text{var}(Z_{0}^{t+1}) \twocolbreak \includeonetwocol{}{\hspace{0.1in}} 
+ O\bigg( \frac{K(p-q)^{3}}{q^{2}} \mathbb{E}\Big[\big(\frac{1}{1+e^{-(Z_{1}^{t}+U_{1}-\nu)}}\big)^{2}\Big]\bigg) \nonumber \\
& = \lambda b_{t} + o(1) \label{mean_variance_BP.eq8}
\end{align}


\section{Proof of Lemma~\ref{Gaussian}}
\label{BP-App.7}

Before we prove the lemma, we need the following lemma from~\cite[Theorem 3]{Poisson_lemma}.
\begin{lemma}
\label{berry}
Let $S_{\gamma} = X_{1} + \cdots + X_{N_{\gamma}}$, where $X_{i}: i\geq 1$ are i.i.d. random variables with mean $\mu$, variance $\sigma^{2}$ and $\mathbb{E}[|X_{i}|^{3}] \leq  \rho^{3}$, and for some $\gamma >0$, $N_{\gamma}$ is a $\text{Poi}(\gamma)$ random variable independent of $(X_{i}: i\geq 1)$. Then,
\begin{align}
&\sup_{x} \big| \mathbb{P}\big( \frac{S_{\gamma} - \gamma\mu }{\sqrt{\gamma(\mu^{2}+\sigma^{2})}} \leq x \big) - \phi(x) \big|  \leq \frac{0.3041 \rho^{3}}{\sqrt{\gamma(\mu^{2}+\sigma^{2})^{3}}}
\end{align}
\end{lemma}

For $t \geq 0$, $Z_{0}^{t+1}$ can be represented as follows:
\begin{align}
Z_{0}^{t+1} &= -K(p-q) + \sum_{i=1}^{N_{nq}} X_{i} 
\label{Z_{0}}
\end{align}
where $N_{nq}$ is distributed according to $\text{Poi}(nq)$, the random variables $X_{i}, i\geq 1$ are mutually independent and independent of $N_{nq}$ and $X_{i}$ is a mixture: 
\[
X_{i} = \frac{(n-K)q}{nq} M(Z_{0}^{t} + U_{0}) + \frac{Kq}{nq}M(Z_{1}^{t} + U_{1}).
\]

Starting with~\eqref{Z_{0}}, using the properties of compound Poisson distribution, and then applying Lemma~\ref{mean_variance_BP}:
\begin{align}
nq\mathbb{E}[X_{i}^{2}] &= \text{var}(Z_{0}^{t+1}) = \lambda b_{t} + o(1) 
\label{gauss.3}
\end{align}
Also, using $\log^{3}(1+x) \leq x^{3}$ for all $x \geq 0$:
\begin{align}
 nq\mathbb{E}[|X_{i}^{3}|] 
& \leq \frac{K(p-q)^{3}}{q^{2}} \mathbb{E}\Big[\big(\frac{1}{1 + e^{-(Z_{1}^{t}+U_{1})+\nu }} \big)^{3}\Big] \twocolbreak \includeonetwocol{}{\hspace{0.1in}} 
+ \frac{(n-K)(p-q)^{3}}{q^{2}} \mathbb{E}\Big[\big(\frac{1}{1 + e^{-(Z_{0}^{t}+U_{0})+\nu } }\big)^{3}\Big] \nonumber \\
& \overset{(a)}{\leq}  \frac{K(p-q)^3}{q^2} \nonumber \\
& \overset{(b)}{=} o(1) \label{gauss.4}
\end{align}
where $(a)$ holds by Lemma~\ref{Lem.3} for $g(x) = \frac{1}{1+e^{-(x-\nu)}}$ and $(b)$ holds since $\frac{p}{q} \to 1$.

Combining~\eqref{gauss.3} and~\eqref{gauss.4} yields $\frac{\mathbb{E}[|X_{i}^{3}|]}{\sqrt{nq\mathbb{E}^{3}[X_{i}^{2}]}} = \frac{nq\mathbb{E}[|X_{i}^{3}|]}{\sqrt{(nq\mathbb{E}[X_{i}^{2}])^{3}}} \to 0$, which together with Lemma~\ref{berry} yields:
\begin{align}
&\sup_{x} \big| \mathbb{P}\big( \frac{Z_{0}^{t+1} + \frac{\lambda b_{t}}{2} }{\sqrt{\lambda b_{t}}} \leq x \big) - \phi(x) \big|  \to 0 
\label{gauss.5} 
\end{align}

Similarly, for $t \geq 0$, $Z_{1}^{t+1}$ can be represented as follows:
\begin{align}
Z_{1}^{t+1} &= -K(p-q) + \frac{1}{\sqrt{(n-K)q}} \sum_{i=1}^{N_{(n-K)q + Kp}} Y_{i} \label{Z_{1}}
\end{align}
where $N_{(n-K)q + Kp}$ is distributed according to $\text{Poi}((n-K)q + Kp)$, the random variables $Y_{i}, i\geq 1$ are mutually independent and independent of $N_{(n-K)q + Kp}$ and $Y_{i}$ is a mixture: 
\begin{align*}
Y_{i} = &\frac{(n-K)q}{(n-K)q + Kp}M(Z_{0}^{t} + U_{0}) \twocolbreak + \frac{Kp}{(n-K)q + Kp}M(Z_{1}^{t} + U_{1}).
\end{align*}

Starting with~\eqref{Z_{1}}, using the properties of compound Poisson distribution, and then applying Lemma~\ref{mean_variance_BP}:
\begin{align}
((n-K)q + Kp)\mathbb{E}[Y_{i}^{2}] &= \text{var}(Z_{1}^{t+1}) = \lambda b_{t} + o(1) \label{gauss.6}
\end{align}
Also, using $\log^{3}(1+x) \leq x^{3}$ for all $x \geq 0$:
\begin{align}
((n-K)q + Kp)\mathbb{E}[|Y_{i}^{3}|] & = nq\mathbb{E}[|X_{i}|^{3}] \twocolbreak \includeonetwocol{}{\hspace{-0.2in}}
+ K(p-q) \mathbb{E}\bigg[\Big(\frac{\frac{p}{q}-1}{1+e^{-(Z_{1}^{t}+U_{1})+\nu }}\Big)^{3}\bigg]\nonumber \\
& \leq o(1) \label{gauss.7}
\end{align}
where~\eqref{gauss.7} holds since $\frac{p}{q} \to 1$.

Combining~\eqref{gauss.6} and~\eqref{gauss.7} yields $\frac{\mathbb{E}[|Y_{i}^{3}|]}{\sqrt{(n-K)q + Kp)\mathbb{E}^{3}[Y_{i}^{2}]}} \to 0$, which together with Lemma~\ref{berry} yields:
\begin{align}
&\sup_{x} \big| \mathbb{P}\big( \frac{Z_{1}^{t+1} - \frac{\lambda b_{t}}{2} }{\sqrt{\lambda b_{t}}} \leq x \big) - \phi(x) \big|  \to 0 \label{gauss.8} 
\end{align}

Hence, using~\eqref{gauss.5} and~\eqref{gauss.8}, it suffices to show that $\lambda b_{t} \to v_{t+1}$, which implies that~\eqref{gauss.1} and~\eqref{gauss.2} are satisfied. We use induction to prove that $\lambda b_{t} \to v_{t+1}$. At $t=0$, we have: $v_{1} = \lambda \mathbb{E}[\frac{1}{e^{-\nu} + e^{-U_{1}}}] = \lambda b_{0}$. Hence, our claim is satisfied for $t=0$. Assume that $\lambda b_{t} \to v_{t+1}$. Then,
\begin{align}
b_{t+1} & = \mathbb{E}[\frac{1}{e^{-\nu} + e^{-(Z_{1}^{t+1}+U_{1})}}] = \mathbb{E}_{U_{1}}[\mathbb{E}_{Z_{1}}[\frac{1}{e^{-\nu} + e^{-(Z_{1}^{t} + u)}}]] \nonumber \\
& = \mathbb{E}_{U_{1}}[\mathbb{E}_{Z_{1}}[f(Z_{1}^{t+1}; u,\nu)]] = \mathbb{E}_{U_{1}}[\mathbb{E}_{Z_{1}}[\mathcal{E}_n]] \label{gauss.9}
\end{align}
where $f(z;u,\nu) = \frac{1}{e^{-\nu} + e^{-(z+u)}}$ and $\mathcal{E}_n$ is a sequence of random variables representing $f(Z;u,\nu)$ as it evolves with $n$. Let $G(s)$ denote a Gaussian random variable with mean $\frac{s}{2}$ and variance $s$. 

From~\eqref{gauss.8}, we have $\text{Kolm}\big(Z_{1}^{t+1}, G(\lambda b_{t})\big) \to 0$ where $\text{Kolm}(\cdot,\cdot)$ is the Kolmogorov distance (supremum of absolute difference of CDFs). Since $f(z; u,\nu)$ is non-negative and monotonically increasing in $z$ and since the Kolmogorov distance is preserved under monotone transformation of random variables, it follows that  $\text{Kolm}\big( f(Z_{1}^{t+1}; u, \nu) , f(G(\lambda b_{t}); u,\nu)\big) \to 0$ . Since $\lim_{z \to \infty} f(z;u\nu)  = e^{\nu}$, using the definition of Kolmogorov distance and by expressing the CDF of $f(G(\lambda b_{t}); u,\nu)$ in terms of the CDF of $G(\lambda b_{t})$ and the inverse of $f(z; u, \nu)$, we get:
\begin{align}
&\sup_{0<c<e^{\nu}} \Big|  F_{\mathcal{E}_n}(c) - F_{G(\lambda b_{t})}\Big(\log\big(\frac{c e^{-u}}{1-c e^{-\nu}}\big)\Big) \Big| \to 0 \label{gauss.10}
\end{align}
From the induction hypothesis, $\lambda b_{t} \to v_{t+1}$. Thus,
\begin{align}
&\sup_{0<c<e^{\nu}} \Big|  F_{\mathcal{E}_n}(c) - F_{G(v_{t+1})}\Big(\log\big(\frac{c e^{-u}}{1-c e^{-\nu}}\big)\Big) \Big| \to 0 \label{gauss.12}
\end{align}
which implies that the sequence of random variables $\mathcal{E}_n$ converges in Kolmogorov distance to a random variable $\frac{1}{e^{-\nu} + e^{-(G(v_{t+1})+u)}}$ as $n\to\infty$. This implies the following convergence in distribution:
\begin{align}
    \mathcal{E}_n \overset{i.d.}{\rightarrow} \frac{1}{e^{-\nu} + e^{-(G(v_{t+1})+u)}} \label{gauss.12.1}
\end{align}
Moreover, the second moment of $\mathcal{E}_n$ is bounded from above independently of $n$:
\begin{align}
\mathbb{E}[\mathcal{E}_{n}^{2}] & \overset{(a)}{\leq} e^{2\nu} \overset{(b)}{\leq} A \label{gauss.13}
\end{align}
where $(a)$ holds by the definition of $\mathcal{E}_n$, and $(b)$ holds for positive constant $A$ since based on the assumptions of the lemma, $\nu$ is constant as $n \to \infty$. 
%

By~\eqref{gauss.12},~\eqref{gauss.12.1} and~\eqref{gauss.13}, the dominated convergence theorem implies that, as $n \to \infty$, the mean of $\mathcal{E}_n$ converges to the mean of the random variable $\frac{1}{e^{-\nu} + e^{-(G(v_{t+1})+u)}}$. Since the cardinality of side information is finite and independent of $n$, it follows that:
\begin{align}
b_{t+1} & = \mathbb{E}_{U_{1}}\big[\mathbb{E}[\mathcal{E}_n]\big] \nonumber \\
& \overset{(a)}{\to} \mathbb{E}_{U_{1}}\bigg[\mathbb{E}_{Z}\bigg[\frac{1}{e^{-\nu} + e^{-(\frac{v_{t+1}}{2} + \sqrt{v_{t+1}} Z) - u} } \bigg] \bigg] \nonumber \\
& = \frac{v_{t+2}}{\lambda} \label{final.1}
\end{align}
where in $(a)$ we define $Z \sim \mathcal{N}(0,1)$. Equation~\eqref{final.1} implies that $\lambda b_{t+1} \to v_{t+2}$, which concludes the proof of the lemma.
 

\section{Proof of Lemma~\ref{The.weak.unbound}}
\label{BP-App.8}

Let $\kappa = \frac{n}{K}$. Since for all $\ell$: $| h_{\ell} | < \nu$, it follows that for any $t \geq 0$ and for sufficiently large $\kappa$:
\begin{align}
v_{t+1} & = \lambda \;\mathbb{E}_{Z,U_{1}}\bigg[\frac{1}{e^{-\nu} + e^{-(\frac{v_{t}}{2} +\sqrt{v_{t}}Z)-U_{1}}}\bigg] \nonumber \\
& = \lambda \sum_{\ell=1}^{L} \frac{\alpha_{+,\ell}^{2}}{\alpha_{-,\ell}} \;\mathbb{E}_{Z} \bigg[\frac{1}{e^{-\nu(1-\frac{h_{\ell}}{\nu})} + e^{-(\frac{v_{t}}{2} +\sqrt{v_{t}}Z)}}\bigg] \nonumber \\
& \overset{(a)}{=} \lambda \sum_{\ell=1}^{L} \frac{\alpha_{+,\ell}^{2}}{\alpha_{-,\ell}} \;\mathbb{E}_{Z} \bigg[\frac{1}{e^{-C_{l}\nu} + e^{-(\frac{v_{t}}{2} +\sqrt{v_{t}}Z)}}\bigg] \nonumber \\
& \overset{(b)}{=} \lambda \ChiS e^{v_{t}} (1+o(1)) \label{all_o.1}
\end{align}
where $(a)$ holds for positive constants $C_{\ell}$, $\ell \in \{ 1,\cdots , L\}$ and $(b)$ holds because $\mathbb{E}_{Z}[e^{\frac{v_{t}}{2} +\sqrt{v_{t}}Z}] = e^{v_{t}}$.

Consider the sequence $w_{t+1} = e^{w_t}$ with $w_{0} = 0$. Define $t^{*} = \log^{*}(\nu)$ to be the number of times the logarithm function must be iteratively applied to $\nu$ to get a result less than or equal to one. Since $w_{1} = 1$ and $w_{t}$ is increasing in $t$, we have $w_{t^{*}+1} \geq \nu$ (check by applying the $\log$ function $t^{*}$ times to both sides). Thus, as $\kappa$ grows, we have $\nu = o(w_{t^{*}+2})$.

Since $\ChiS \to \infty$ as $\kappa$ grows, it follows by induction that for any fixed $\lambda >0$:
\begin{align}
& v_{t} \geq w_{t} \label{all_o.3.1}
\end{align}
for all $t \geq 0$ and for all sufficiently large $\kappa$. Thus,
\begin{align}
& v_{t^{*}+2} \geq  w_{t^{*}+2} \label{all_o.3}
\end{align}
which implies that as $\kappa$ grows, $\nu = o(v_{t^{*}+2})$ and $h_{\ell} = o(v_{t^{*}+2})$ for all $\ell$. Since $v_{t}$ is increasing in $t$, using~\eqref{all_o.1} and~\eqref{all_o.3}, we get for all sufficiently large $\kappa$ and after $t^{*} +2$ iterations of belief propagation (or for a tree of depth $t^{*}+2$):
\begin{align}
& \mathbb{E}_{U_{0}}\Big[Q(\frac{\nu+\frac{v_{t^*+2}}{2} - U_{0}}{\sqrt{v_{t^*+2}}})\Big] = Q\Big( \frac{1}{2} \sqrt{v_{t^{*}+2} }(1+o(1))\Big) \label{Q.1.1}  \\
& \mathbb{E}_{U_{1}}\Big[Q(\frac{-\nu+\frac{v_{t^*+2}}{2} + U_{1}}{\sqrt{v_{t^*+2}}})\Big] = Q\Big( \frac{1}{2} \sqrt{v_{t^{*}+2}}(1+o(1))\Big) \label{Q.1.2}
\end{align}
%
Since $Q(x) \leq e^{-\frac{1}{2} x^{2}}$ for $x\geq0$, then using~\eqref{all_o.3},~\eqref{Q.1.1} and~\eqref{Q.1.2}:
\begin{align}
&\frac{n-K}{K} Q\Big( \frac{1}{2} \sqrt{v_{t^{*}+2}}(1+o(1)) \Big) \to 0 \label{all_o.21} \\ 
& Q\Big( \frac{1}{2}\sqrt{v_{t^{*}+2}}(1+o(1)) \Big) \to 0 \label{all_o.22}
\end{align}
Using~\eqref{all_o.21} and~\eqref{all_o.22} and Lemma~\ref{MAP_unbounded}, we get:
\begin{align}
&\lim_{\frac{n}{K} \to \infty} \lim_{nq,Kq \to \infty} \lim_{n\to\infty} \frac{\mathbb{E}[\hat{C} \triangle C^{*}]}{K} = 0
\end{align}


%
%
%
%

\bibliographystyle{IEEEtran}
\bibliography{IEEEabrv,ISIT2016_updated}

\end{document}